\documentclass[11pt,onecolumn]{IEEEtran}

\ifCLASSINFOpdf
  \usepackage[pdftex]{graphicx}
 \else
  \usepackage[dvips]{graphicx}
\fi

\usepackage[cmex10]{amsmath}
\usepackage{amssymb}

\usepackage{amsmath, amssymb, amsthm}
\usepackage{mathtools}
\usepackage{dsfont}
\usepackage{hyperref}

\usepackage{algorithm,algorithmic}
\usepackage{array}
\usepackage{etoolbox}

\usepackage{amssymb}
\usepackage{enumitem}
\usepackage{verbatim}
\usepackage{bbold}
\usepackage[normalem]{ulem}
\usepackage{xcolor}

\input{mathdef.tex}

\newtheorem{definition}{Definition}
\newtheorem{lemma}{Lemma}
\newtheorem{theorem}{Theorem}

\AfterEndEnvironment{proposition}{\noindent\ignorespaces}
\theoremstyle{remark}
\newtheorem{remark}{\bf Remark}
\newcommand{\E}{\mathbb{E}}

\newcommand{\beq}{\begin{equation}}
\newcommand{\eeq}{\end{equation}}

\newcommand{\bard}{\bar{d}}

\makeatletter

\newcommand{\Rmnum}[1]{\expandafter\@slowromancap\romannumeral #1@}
\makeatother
\usepackage{mathtools}

\long\def\/*#1*/{}

\begin{document}
%
% paper title
% can use linebreaks \\ within to get better formatting as desired
% Do not put math or special symbols in the title.
\title{Binary Classification with XOR Queries: \\Fundamental Limits and An Efficient Algorithm}

\author{Daesung Kim and Hye Won Chung % \thanks{Manuscript received January 20, 2002; revise d January 30, 2002. This work was supported by the I EEE.}%
\thanks{Daesung Kim (jklprotoss@kaist.ac.kr) and Hye Won Chung (hwchung@kaist.ac.kr) are with the School of Electrical Engineering at KAIST in South Korea.  %This work was partially supported by National Research Foundation of Korea under grant number 2017R1E1A1A01076340 and  by the Ministry of Science and ICT, Korea, under the ITRC support program (IITP-2020-2018-0-01402). 
This work was supported in part by National
Research Foundation of Korea under Grant 2017R1E1A1A01076340; in part
by the Ministry of Science and ICT, South Korea, under the ITRC support
program under Grant IITP-2021-2018-0-01402; and in part by the Institute of
Information and Communications Technology Planning \& Evaluation (IITP)
grant funded by the Korea Government MSIT under Grant 2020-0-00626.
This research was presented in part at 2020 IEEE International Symposium on Information Theory~\cite{9174227}.  

}}

\maketitle

%%% Many authors with many affiliations:
% \author{%
%   \IEEEauthorblockN{Albus Dumbledore\IEEEauthorrefmark{1},
%                     Olympe Maxime\IEEEauthorrefmark{2},
%                     Stefan M.~Moser\IEEEauthorrefmark{3}\IEEEauthorrefmark{4},
%                     and Harry Potter\IEEEauthorrefmark{1}}
%   \IEEEauthorblockA{\IEEEauthorrefmark{1}%
%                     Hogwarts School of Witchcraft and Wizardry,
%                     1714 Hogsmeade, Scotland,
%                     \{dumbledore, potter\}@hogwarts.edu}
%   \IEEEauthorblockA{\IEEEauthorrefmark{2}%
%                     Beauxbatons Academy of Magic,
%                     1290 Pyrénées, France,
%                     maxime@beauxbatons.edu}
%   \IEEEauthorblockA{\IEEEauthorrefmark{3}%
%                     ETH Zürich, ISI (D-ITET), ETH Zentrum, 
%                     CH-8092 Zürich, Switzerland,
%                     moser@isi.ee.ethz.ch}
%   \IEEEauthorblockA{\IEEEauthorrefmark{4}%
%                     National Chiao Tung University (NCTU), 
%                     Hsinchu, Taiwan,
%                     moser@isi.ee.ethz.ch}
% }

%% Abstract: 
%% If your paper is eligible for the student paper award, please add
%% the comment "THIS PAPER IS ELIGIBLE FOR THE STUDENT PAPER
%% AWARD." as a first line in the abstract. 
%% For the final version of the accepted paper, please do not forget
%% to remove this comment!
%%
\begin{abstract}

%Crowdsourcing systems have emerged as an effective platform to label data and classify objects with relatively low cost by exploiting non-expert workers. To ensure reliable recovery of unknown labels with as few number of queries as possible, we consider an effective query type that asks ``group attribute'' of a chosen subset of objects. In particular, we consider the problem of classifying $m$ binary labels with XOR queries that ask whether the number of objects having a given attribute in the chosen subset of size $d$ is even or odd. The subset size $d$, which we call query degree, can be varying over queries. Since a worker needs to make more efforts to answer a query of a higher degree, we consider a noise model where the accuracy of worker's answer changes depending both on the worker reliability and query degree $d$. For this general model, we characterize the information-theoretic limit on the optimal number of queries to reliably recover $m$ labels in terms of a given combination of degree-$d$ queries and noise parameters. Further, we propose an efficient inference algorithm that achieves this limit even when the noise parameters are unknown.
We consider a query-based data acquisition problem for binary classification of unknown labels, which has diverse applications in communications, crowdsourcing, recommender systems and active learning. 
To ensure reliable recovery of unknown labels with as few number of queries as possible, we consider an effective query type that asks ``group attribute'' of a chosen subset of objects. In particular, we consider the problem of classifying $m$ binary labels with XOR queries that ask whether the number of objects having a given attribute in the chosen subset of size $d$ is even or odd. The subset size $d$, which we call query degree, can be varying over queries. We consider a general noise model where the accuracy of answers on queries changes depending both on the worker (the data provider) and query degree $d$. For this general model, we characterize the information-theoretic limit on the optimal number of queries to reliably recover $m$ labels in terms of a given combination of degree-$d$ queries and noise parameters. Further, we propose an efficient inference algorithm that achieves this limit even when the noise parameters are unknown.

\end{abstract}

\section{Introduction}

Binary classification is one of the most fundamental problems in engineering and appears in a wide variety of fields such as communication, machine learning \cite{kotsiantis2006machine}, recommender systems, crowdsourcing \cite{karger2014budget}, and VLSI systems \cite{karypis2000multilevel}. In data acquisition for binary classification, one often acquires noisy answers from queries on the unknown binary labels, and by applying an inference algorithm the true labels are recovered. The common goal of such a data acquisition problem is to reliably recover the true labels at the minimum sample complexity (the number of queries). To achieve this goal,  an efficient querying strategy needs to be designed. There also exist some cases where we do not have a freedom to design queries, but rather it is directly inherited from the application. Designing a statistically-efficient inference algorithm is also important to fully exploit the information from answers so as to minimize the sample complexity, but at the same time, the algorithm should have affordable computational complexity to be used in large-scale problems.

Different types of queries have been studied, regarding different properties of target applications. For example, many literatures on crowdsourcing have investigated the most basic querying method named repetition query that asks one label at a time repeatedly to many workers \cite{gao2013minimax, li2014error, dalvi2013aggregating, zhang2014spectral}. More complex querying method such as pairwise comparison \cite{mazumdar2017clustering} asking whether or not two objects belong to the same class, or ``triangle'' queries \cite{vinayak2016crowdsourced}, which compare three objects simultaneously, were considered to increase the statistical-efficiency in answers. Pairwise comparison was also studied extensively through the stochastic block model \cite{abbe2015exact}, which is a standard model for studying community recovery problems in graphs. The homogeneity measurement that extends pairwise comparison to the hypergraph case to measure the homogeneity of more than two nodes was further investigated in some previous works \cite{kim2017community, ahn2019community, lee2020hypergraph}.

In this work, we focus on the XOR-based querying method, which asks whether the number of objects having a given attribute in the chosen subset of size (degree) $d$ is even or odd. Binary classification with XOR queries has long been studied in the context of linear codes in channel coding \cite{gallager1962low, mackay2005fountain, arikan2009channel}, XOR-based hypergraph clustering \cite{ahn2019community}, and random XOR-constrained satisfaction problem \cite{pittel2016satisfiability} with applications in statistical physics \cite{jia2004spin}. We study this problem in a very general setup where the error probability of the answers as well as the subset size $d$ is non-uniform over queries. We assume that the queries are answered by many workers having different error probabilities although each worker answers the queries of the same $d$ uniformly with assigned error probability. The terminology ``worker'' is inspired by the crowdsourcing system \cite{karger2014budget}, but it can be considered as different channels through which the data is transmitted, or different environments that affect the edge generation probability of a hypergraph. We additionally assume that the error probabilities of workers are unknown to reflect the scenarios, in crowdsourcing where worker reliabilities are unknown a priori, or in communication systems where the channel condition varies over time. Lastly, the assumption of non-uniform $d$ allows very general setups in the corresponding problems where we obtain measurements over different number of objects interacting with each other. % enabling our model to cover various problems simultaneously.

\subsection{Main Contributions}
In this work, we provide rigorous and general analysis of XOR queries in binary classification and demonstrate the sample efficiency of this type of querying strategy. We derive a sharp threshold on the required number of queries to recover $m$ binary labels in terms of a given combination of degree-$d$ XOR queries and worker noise parameters. Further, we provide an efficient inference algorithm to extract correct labels from the noisy answers of XOR queries, which achieves the information-theoretic limit even when the noise parameters are unknown.

Concretely, when the fraction of degree-$d$ queries is $\Phi_d\in[0,1]$ for $d\in\{1,\dots,D\}$ and the probability that a worker $k$ provides an incorrect answer to a degree-$d$ query is $\epsilon_{k,d}<1/2$, we show in Theorem~\ref{thm:thm1} that the number of queries should be at least 
\beq
\label{eqn:nmlogm}
n=\frac{ m\log m}{\sum_{d = 1}^D \sum_{k = 1}^w \frac{d\Phi_d}{w} ( \sqrt{1 - \epsilon_{k,d}}-\sqrt{\epsilon_{k,d}})^2}
\eeq to recover all $m$ binary labels with high probability as $m\to\infty$, where the maximum query degree is $D=\Theta(1)$ and each query is assigned to a worker chosen uniformly at random among total $w$ workers. We provide both upper and lower bounds that deviate only by an arbitrary small constant factor from~\eqref{eqn:nmlogm}.  The upper and lower bounds are derived by analyzing the optimal maximum likelihood (ML) decoder with the knowledge of worker noise parameters $\{\epsilon_{k,d}\}$.%\sout{If $\epsilon_{k,d}=\epsilon_k$, i.e., the worker error probability does not depend on the query type but only on the worker reliability as in Dawid and Skene model \cite{dawid1979maximum}, the sample complexity is inversely proportional to the average query degree $\sum d\Phi_d$. For a general set of $\{\epsilon_{k,d}\}$, the optimal query degree $d$ is determined as $d^*$ that maximizes $\sum_{k=1}^w d (\sqrt{1 - \epsilon_{k,d}}-\sqrt{\epsilon_{k,d}})^2$. However, in real crowdsourcing systems, the worker error model is unknown so that it is impossible to find the optimal query degree $d$ at the stage of query design.}

Our main contribution is on proposing an efficient inference algorithm that achieves the optimal sample complexity for any combination of the query degree $d$'s even when the worker noise parameters $\{\epsilon_{k,d}\}$ are unknown. 
The main idea is to boost the accuracy of our estimates on the correct labels by three steps, where the worker noise parameters are estimated after the first two steps and then used to refine the estimates on labels at the last step.
More specifically, we show that the weak recovery of labels (of which the formal definition will be provided in the next section) is possible without estimating the worker noise parameters by relying on one-step message passing and majority voting. 
The worker noise parameters can be estimated up to a desired accuracy based on the weakly recovered labels. Finally, the strong recovery of labels (of which the formal definition will be provided in the next section) is achieved by the weighted majority voting with the estimated noise parameters. The proposed algorithm is inspired by the recent two-step estimation approaches, which have been used in clustering or non-convex estimation problems, where the initial estimate, which is close to the solution up to a certain limit, is provided by some classical technique (e.g. spectral method) and the refinement step (e.g. gradient descent) is followed to boost the accuracy of the estimate \cite{keshavan2010matrix, candes2015phase,abbe2015exact}.

One important assumption posed on our algorithm is the presence of $\Theta(m)$ degree-1 queries. It is used to construct an initial estimate that makes our algorithm to move toward the correct direction. In applications of data acquisition where one can design querying strategy, this is a natural assumption. On the other hand, in other applications where degree-1 queries are unavailable, the degree-1 queries can be viewed as some sort of side information that provides an initial estimate that is better than random guess. Hence, the assumption can be related to the recent research direction where side information is used to help solve the original problem with smaller sample complexity \cite{xu2013speedup, ahn2018binary, saad2018community}.

We show that the proposed algorithm guarantees the recovery of $m$ labels at the optimal sample complexity provided by \eqref{eqn:nmlogm} as $m\to\infty$ even when the worker noise parameters are unknown (Theorem~\ref{thm:thm2}). We provide some experimental results with synthetic data that support our theoretical findings. In particular, to fairly compare the XOR quires of different degrees and with other types of queries, we consider an error model called $d$-coin flip model. We also apply our algorithm to crowdsourced binary classification with the data collected from Amazon Mechanical Turk and show the effectiveness of XOR query in reducing the sample complexity.

%\sout{We also provide the empirical performance analysis of the proposed algorithm for finite $m$'s both for synthetic data and real data collected from Amazon Mechanical Turk and demonstrate the effectiveness of the XOR query over other standard query types in reducing the number of queries required to recover $m$ binary labels (Section~\ref{sec:experiment}).}

\subsection{Related Works}
%\paragraph{Related Works}
%We describe some of the previous works that are closely related to this work. If the reliability of all workers are the same and the degree of all queries are uniform, the model reduces to the censored block model (CBM) on graph, when $d = 2$, or hypergraph, when $d \geq 3$. [5] studied community detection in a hypergraph and proposed that exact recovery is possible when $n = \Omega \left(\frac{k \log k}{ d ( \sqrt{\epsilon} - \sqrt{1 - \epsilon})^2} \right)$. However, they did not consider the case of nonuniform query degree nor provided a computationally efficient algorithm to achieve the bound. When the degree is fixed to 3, the problem is also a special case of tensor completion or 3-SAT problem. In these problems, it is conjectured that the information theoretic bound cannot be achieved with any polynomial time algorithm and $n = \Omega (n ^ {3 / 2})$ is required for the perfect recovery. We resolved this problem by adding degree 1 queries, which is possible in crowdsourcing, and achieved the optimal number of queries. The same croudsourcing model was studied in [6] but with noiseless model and degree distribution given by Soliton distribution.

\subsubsection{Constraint Satisfaction Problem}
The recovery of $m$ binary labels from noisy XOR queries can be viewed as an example of a planted constraint satisfaction problem (CSP), which is a subject of intense study in computer science, probability theory, and statistical physics, motivated by clustering, community detection, and cryptographic applications. Consider in particular a random planted $d$-XORSAT problem, where the goal is to recover $m$ binary variables (the planted solution) satisfying a set of $n$ constraints such that XOR of size-$d$ subsets of variables should be equal 0 (or 1) where the subsets are chosen uniformly at random among $m \choose d$ possibilities. This is the same recovery problem as ours except that for our setup, the subset size $d$ can be varying over constraints and the XOR value can be noisy. There have been many works~\cite{abbe2013conditional,achlioptas2008algorithmic,haanpaa2006hard} to answer 1) how large $n$ should  be to make the planted solution a unique solution, and 2) how large $n$ should be for the planted solution recoverable by an efficient algorithm. When $d=3$, this problem is also a special case of tensor completion problems, and the best known algorithm with polynomial-time complexity requires $n=\Omega(m^{3/2})$~\cite{barak2016noisy}. In this work, we show that by adding $\Theta(m)$ degree-1 queries, which does not change the order of sample complexity, the optimal sample complexity of $\Theta(m\log m)$ is achievable by an efficient algorithm of $O(m\log m)$ time steps. %\hl{\sout{This result demonstrates the benefit of mixing degree-1 queries with higher degree queries in reducing the gap between information-theoretic limit and computational limit.}}

\subsubsection{Hypergraph Clustering}
The problem we consider also has close connections to the XOR-based hypergraph clustering problem~\cite{ahn2019community,abbe2013conditional,watanabe2013message}. The goal of hypergraph clustering is to classify nodes (binary labels) by using randomly sampled XOR measurements among ${m\choose d}$ possibilities on subsets of node labels. 
The main difference from our model is that in the hypergraph clustering the subset size $d$ is often fixed as a constant and each hyperedge in the set of size ${m\choose d}$ is  independently sampled (without replacement) with a fixed probability $p$ (thus the number of total measurements is a random variable distributed by $\text{Binomial}({m \choose d},p)$.) Also, the error probability is often identical over the measurements. In our model, on the other hand, we can design queries so that for each query we randomly choose the query degree $d$ and select $d$ labels among $m$ uniformly at random. The error probability of each query may also depend on the query degree and the worker. In~\cite{ahn2019community}, the necessary and sufficient conditions on the expected number of hyperedge measurements  $p{m \choose d}$ were analyzed and shown to be $\frac{m\log m}{d(\sqrt{1-\epsilon}-\sqrt{\epsilon})^2}$. In our work, we generalize this result for our measurement model (where the number $n$ of queries is fixed and each query is randomly designed and assigned to one of workers in the system) and show that~\eqref{eqn:nmlogm} is the necessary and sufficient number of queries to recover all the $m$ labels reliably. Addition to this theoretical analysis, the main contribution of our work is that we provide the almost linear-complexity inference algorithm (up to logarithmic factor) with which we can infer all the $m$ labels with high probability even when the error parameters of the workers are unknown. 
%The workers can correspond to different channel conditions or different time-steps the hypergraphs were obtained. However, it is natural to assume that the error probability does not depend on $d$ or $d$ is uniform over the measurements in these problems.

\subsubsection{Linear Codes}
The result of this work can be applied to communication systems by viewing XOR queries as a linear code for binary labels. The randomly generated XOR queries correspond to the scenario where a receiver has no control over what packets to receive and just receives random $n$ packets from the transmitter (as in rateless coding setup with fountain codes~\cite{mackay2005fountain}). Also, the workers in our model can be thought of as many different channels with unknown error probabilities as in~\cite{8636012,8437703}. The degree constraint $D=\Theta(1)$ for the coded bits can be motivated from locally encodable coding~\cite{mazumdar2017semisupervised}: a data compression/transmission problem where each coded bit depends only on a small number of input bits. 

\subsubsection{Crowdsourcing}

Crowdsourcing systems usually rely on simple querying types that human workers are easy to answer. However, information efficiency of such simple query types is often limited; all the simple query types including repetition query \cite{karger2014budget}, pairwise comparison \cite{mazumdar2017clustering}, and homogeneity measurement \cite{ahn2019community} require larger sample complexity than the degree-3 XOR querying scheme under the same error probability. However, XOR queries are relatively harder to answer for human workers, so that the corresponding error rate in the answers may increase in practice. 
Hence, it is worth investigating whether the gain in sample complexity can offset the loss from the increased error rate of XOR queries in real crowdsourcing systems. Our experimental result provided in Section \ref{sec:experiment} answers this question in an affirmative way, and opens a possibility for XOR queries to be used in crowdsourced classification.

\subsection{Organization of the Paper}
The rest of the paper is organized as follows. In Section~\ref{sec:model}, we formulate the binary classification problem with XOR query for a general noise model where the accuracy of worker's answer changes depending both on the worker reliability and query degree $d$. Section~\ref{sec:fund} provides the information-theoretic limits on the required number of queries to recover all the labels with high probability, in terms of a given combination of degree-$d$ queries and worker noise parameters. In Section~\ref{sec:alg}, we present our computationally-efficient algorithm that achieves the information-theoretic limit on the optimal number of queries, even without the knowledge of worker noise parameters. In Section~\ref{sec:experiment}, we present simulation results that demonstrate the effectiveness of XOR querying strategy and the proposed inference algorithm compared to existing crowdsourcing strategies both for synthetic and real datasets. 
Section~\ref{sec:pfthm2} provides proofs on the performance of the proposed algorithm. 
 Section~\ref{sec:con} provides conclusions with future research directions. Technical proof details on the main results can be found in appendices.

\subsection{Notations}
We denote a vector by a bold face letter, e.g. $\vec{v}$, and the $i$-th component of it by $v_i$. Both of Bernoulli distribution (value 0 or 1) and Rademacher distribution (value $-1$ or $1$) with parameter $p$ are denoted by $\text{Bern}(p)$. We use $[n]$ to denote $\{1, 2, \cdots, n\}$, and for any set $A$, $|A|$ is the number of elements in $A$. We define the function $\sign(x)$ as $1$ if $x > 0$, $-1$ if $x < 0$, and $1$ with probability 1/2 and $-1$ with probability 1/2 if $x = 0$. Also, the function $\textsf{TRUNC}(x, I)$ is defined as the element in the interval $I$ that is closest to $x$.

\section{Model}\label{sec:model}
\begin{figure}
  \centering
  \includegraphics[width=5cm]{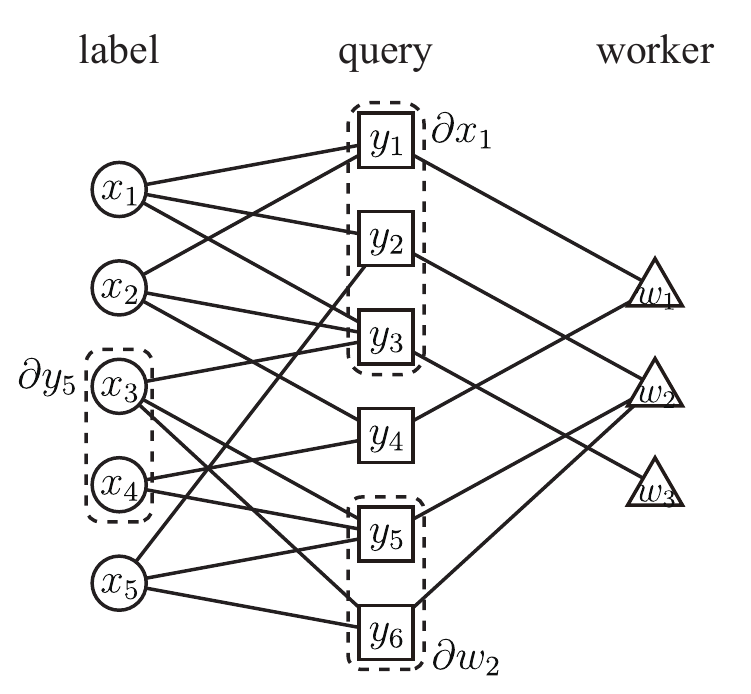}
  \caption{\normalsize A tripartite graph depicting query design and worker assignment. The figure is drawn for the case where the query degrees are either $d=2$ or $d=3$.}
  \label{fig:tripartitte}
\end{figure}
Let $\vec{x} \in \{1,-1\}^m$ be the ground truth label vector we aim to recover and $\hat{\vec{x}}\in \{1, -1\}^m$ be the estimate of the label vector. We ask in total of $n$ queries to $w$ workers, and $\vec{y} \in\{1, -1\}^n$ denotes a collection of all answers we get from the workers. There are three notions of recovery used in this paper. \textit{Strong recovery} refers to the case where $\vec{x}$ itself is recovered perfectly with high probability as $m\to\infty$, i.e., $\Pr{\hat{\vec{x}}\neq \vec{x}}\to 0$. \textit{Weak recovery} is the case where the error probability of each $x_i$ goes to 0 for all $i\in [m]$, i.e., $\Pr{\hat{x}_i\neq x_i}\to 0$. When the error probability of each $x_i$ is better than random guess, i.e., $\Pr{\hat{x}_i\neq x_i}<{1}/{2}$ for all $i \in [m]$, we say that \textit{detection} is possible.

\subsection{Query Design and Assignment}\label{sec:model_qdesign}
Each query is designed independently by first obtaining a query degree $d$ from a probability distribution ${{\Phi}}=\{\Phi_1,\dots, \Phi_D\}$  and then selecting $d$ components of $\vec{x}$, which will be contained in the query, uniformly at random among ${m \choose d}$ possibilities. 
We assume that both the maximum query degree $D$ and the degree distribution $\Phi_D$ do not scale with $m$.
Each query asks XOR of the $d$ labels to a worker chosen uniformly at random among total $w$ workers. Note that we are considering the non-adaptive model; all queries are designed in advance of getting any answer from the workers.

\subsection{Tripartite Graph Representation}\label{sec:model_qassign}
A tripartite graph is used to describe the query design and the worker assignment for each query as in Figure 1. 
%Inspired by this, we define some relevant terminologies. 
We use indices $i\in[m],j\in[n],k\in[w]$ for label, query, worker nodes, respectively.
Let $\partial x_{i}\subset[n]$ denote the set of queries that contain the label $x_i$, and $\partial y_j \subset [m]$ denote the set of labels that are contained in the $j$-th query. Also, let $\partial w_{k}\subset[n]$ denote the set of queries that are assigned to the worker $k \in [w]$. 
When $A \subset [n]$ is a subset of queries, denote by $A_d:=\{j\in A:|\partial y_j|=d\}$ the degree-$d$ queries in the set $A$.
Then, the set $\partial x_i\cap A_d$ includes the degree-$d$ queries in $A$ that are connected to the label node $x_i$, and the set $\partial w_k\cap A_d$ includes the degree-$d$ queries in $A$ that are assigned to the worker $k$. For the $j$-th query, the assigned worker and the query degree are denoted by $w(j)$ and by $d(j)$, respectively.

%We use $\partial x_{i}\subset[n]$ to denote the set of queries that contain $x_i$, and $\partial x_{i,d,A}$ to denote the set of queries in $\partial x_{i,A}$ whose degree is $d$. Also, we define $\partial y_j \subset [m]$, which is the set of labels that are contained in the $j$th query. Lastly, let $\partial w_{k,A}$ denote the set of queries in $A$ that are assigned to worker $k \in [w]$, and let $\partial w_{k,d,A}$ denote the set of degree $d$ queries in $\partial w_{k,A}$. The assigned worker and the degree of the $j$-th query are denoted by $w(j)$ and $d(j)$, respectively. We will sometimes omit the set $A$ from the subscript when it is clear from the text; e.g. $\partial x_i$ or $\partial x_{i,d}$.

\subsection{Noise Model}\label{sec:model_noise}
In order to study the problem in full generality, we assume that the error probabilities of the workers are not uniform and also depend on the query degree. In particular, we assume that the worker $k$ gives a wrong answer with probability $\epsilon_{k,d} \in [\lambda, 0.5)$ when answering a degree-$d$ XOR query, where $\lambda > 0$ is an arbitrary small constant. Hence, there is no perfectly reliable worker, i.e., $\epsilon_{k,d}>0$. This is a technical assumption though to make the weight $\log\frac{1-\hat{\epsilon}_{k,d}}{\hat{\epsilon}_{k,d}}$, which is used for a weighted majority voting in a step of our proposed inference algorithm, not diverge by bounding $\hat{\epsilon}_{k,d}\geq \lambda$ where $\hat{\epsilon}_{k,d}$ is the estimate worker reliability.

\section{Information-Theoretic Bounds on the Optimal Sample Complexity}\label{sec:fund}

We first analyze the optimal number of queries (the sample complexity) for the strong recovery of the label vector ${\vec{x}}$, i.e., to guarantee $\Pr{\hat{\vec{x}}\neq \vec{x}}\to 0$ as $m\to\infty$, in terms of a given fraction of degree-$d$ queries $\{\Phi_1,\dots, \Phi_D\}$ and the noise parameters $\{\epsilon_{k,d}\}$ for $k\in[w]$ and $d\in[D]$. We derive necessary and sufficient conditions on the sample complexity $n$ when the noise parameters $\{\epsilon_{k,d}\}$ of workers' answers are known at the inference algorithm. Thus, this result provides a lower bound on the optimal sample complexity for the case when $\{\epsilon_{k,d}\}$ is unknown, which is a more practical situation for applications such as crowdsourcing systems. In the next section, we develop an inference algorithm that does not require a prior knowledge of $\{\epsilon_{k,d}\}$ but still achieves the information-theoretic limit of the known $\{\epsilon_{k,d}\}$ case.
%We solve for the optimal number of queries required for the strong recovery as stated below.
\begin{theorem}\label{thm:thm1}
Assume that total $n$ XOR queries are randomly and independently generated among which the fraction of degree-$d$ queries is $\Phi_d\geq 0$ for $\sum_{d=1}^D \Phi_d=1$ and $\sum_{d \text{ odd}} \Phi_d>0$. Each query is randomly assigned to a worker $k\in[w]$ who provides an incorrect answer to a degree-$d$ query with probability $0 <\epsilon_{k,d}<1/2$.
With the maximum likelihood (ML) estimator  $\hat{\vec{x}}\in\{1,-1\}^m$, which minimizes $\Pr{\hat{\vec{x}} \neq \vec{x}}$ for a known $\{\epsilon_{k,d}\}$, the strong recovery is possible, i.e., $\Pr{\hat{\vec{x}} \neq \vec{x}} \to 0$ as $m\to\infty$, if the number of queries  is
\begin{equation}\label{eqn:thm1}
n \geq ( 1 + \eta) \frac{ m\log m}{\sum_{d = 1}^D \sum_{k = 1}^w \frac{d\Phi_d}{w} ( \sqrt{1 - \epsilon_{k,d}}-\sqrt{\epsilon_{k,d}})^2},
\end{equation}
and only if 
\begin{equation}
n \geq ( 1 - \eta) \frac{ m\log m}{\sum_{d = 1}^D \sum_{k = 1}^w \frac{d\Phi_d}{w} ( \sqrt{1 - \epsilon_{k,d}}-\sqrt{\epsilon_{k,d}})^2},
\end{equation}
for any arbitrarily small constant $\eta > 0$. %The converse holds when the number of workers $w=o(\log m)$.
\end{theorem}
\begin{IEEEproof}
The proof of this theorem is provided in Appendix~\ref{app:prThm1}. 
\end{IEEEproof}

\begin{remark}[{Efficiency of high-degree XOR queries}]
For a special case where the worker error probability is independent of the query degree $d$, i.e., $\epsilon_{k,d}=\epsilon_k$, Theorem~\ref{thm:thm1} shows that the sample complexity is inversely proportional to the average query degree $\sum_{d=1}^D d\Phi_d$ when $D=\Theta(1)$. This implies that increasing the query degree and asking a more complicated query to a worker helps reduce the required number of queries, if the error probability of a worker's answer does not change depending on the complexity of the queries.
\end{remark}

\begin{remark}[Optimal degree of XOR queries for a general error model]
For a general set of $\{\epsilon_{k,d}\}$, it can be inferred from~\eqref{eqn:thm1} that concentrating the degree distribution to a degree that has the maximum value of $\sum_{k=1}^w d (\sqrt{1 - \epsilon_{k,d}}-\sqrt{\epsilon_{k,d}})^2$ would be the optimal way to minimize the required number of queries. In other words, the optimal query degree that minimizes the required number of queries is $d^*=\argmax_{d\in\{1,\dots,D\}} \sum_{k=1}^w d (\sqrt{1 - \epsilon_{k,d}}-\sqrt{\epsilon_{k,d}})^2$.
However, in many applications, the query designer has no knowledge on the workers' reliabilities $\{\epsilon_{k,d}\}$ at the stage of query design, so determining the optimal query degree in advance is impossible. This motivates a query designer to mix queries with different degrees. A more important aspect of mixing queries with different degrees we argue in this paper is that it is possible to achieve the optimal sample complexity~\eqref{eqn:thm1} with an efficient algorithm even when $\{\epsilon_{k,d}\}$ is unknown, if there exist $\Theta(m)$ number of degree-1 queries. Thus, in the next section we will assume that the degree of the first $m$ queries is fixed to 1 regardless of the query-degree distribution $\Phi$. Note that the addition of $m$ queries has negligible effect on the optimal sample complexity, which scales as $\Theta(m \log m)$. More details will be found in the next section where we propose an efficient inference algorithm.
\end{remark}

\begin{remark}[Comparison to homogeneous query]
The homogeneous query, which asks whether all the items in a chosen subset of size $d$ belong to the same class or not, is another widely-studied group-query type.  
In~\cite{ahn2019community}, the required number of measurements to recover $m$ binary labels from random homogeneous query of degree-$d$ was analyzed in the context of hypergraph clustering. The paper considered the setup where the query degree and the error probability are fixed to $d$ and $\epsilon$, respectively, for all the queries. Under such a setup, it was shown that the required number of measurements (the answers from random homogeneous query) for strong recovery of $m$ binary labels scales as 
$\frac{2^{d-2}}{d}\frac{m\log m}{(\sqrt{1-\epsilon}-\sqrt{\epsilon})^2}$. Note that the information efficiency of homogeneous query decreases as the query degree $d$ increases; whereas that of XOR query  increases as in~\eqref{eqn:thm1}, as long as the error probability of workers' answer does not increase too fast to offset the information gain from the increased query degree. In Section~\ref{sec:experiment}, we provide simulation results to compare the information efficiency of XOR query, homogeneous query as well as repetition query under a fair error-model called $d$-coin flip model.
\end{remark}

%A more important aspect of mixing queries with different degrees we argue in this paper is that it is possible to achieve the optimal sample complexity~\eqref{eqn:thm1} with an efficient algorithm if there are $\Theta(m)$ number of degree-1 queries even when $\{\epsilon_{k,d}\}$ is unknown. 
%Thus, from this point on we will assume that the degree of the first $m$ queries are fixed to 1 regardless of $\Phi$.  Note that the addition of $m$ queries has negligible effect on the optimal sample complexity, which scales as $\Theta(m \log m)$. 

%We introduce two algorithms. First, we propose an algorithm that can be proved to achieve the optimal sample complexity in Theorem 1. It can be considered as a two-step BP algorithm. However, the algorithm converges slowly and requires large $m$ which is unsuitable in practice. Thus, we propose the second algorithm, which is the iterative version of the first algorithm with slight modification. The simulation result for the second algorithm is provided in Section 5.
\section{An Efficient Algorithm Achieving the Optimal Sample Complexity}\label{sec:alg}
In this section, we propose a computationally-efficient algorithm that guarantees the strong recovery of $m$ binary labels at the optimal sample complexity~\eqref{eqn:thm1} even when the worker reliabilities $\{\epsilon_{k,d} \}$ are unknown.
We assume that the first $m$ queries have a fixed degree $d=1$ and ask each of the $m$ labels exactly once.

\subsection{Four-Phase Inference Algorithm for XOR Queries} 
\begin{algorithm}[tb]
	\caption{Four-Phase Inference Algorithm for XOR Queries}
	\label{algorithm1}

\begin{algorithmic}[1]
\STATE \textbf{Data}: 	The observed query answers $\vec{y} \in\{1, -1\}^n$ and the tripartite graph (as in Fig.~\ref{fig:tripartitte}) depicting query design and worker assignment.
\STATE Phase 1. (Detection of labels): Let $A^{(1)}$ be the set of first $m$ degree-1 queries  asking each label. For each label $i \in [m]$, calculate the first estimate of $x_i$ as 
\begin{equation}\label{eqn:x1}
\hat{x}_i^{(1)} = \sign(y_{j_i}) ,
\end{equation}
where $j_i = \partial x_i \cap A^{(1)}$.
\STATE Phase 2. (Weak recovery of labels): Let $A^{(2)}$ be the next $n^{(2)} = m \left(\frac{\log m}{\log \log m}\right)$ queries. For each label $i \in [m]$ and query $j \in \partial x_{i}\cap A^{(2)}$, let 
$
m_{j\to i}^{(2)} = y_j \prod_{i^\prime \in \partial y_j \backslash \{i\}} \hat{x}_{i^\prime}^{(1)},
$
and calculate the second estimate of $x_i$ as
\beq\label{eqn:x2}
\hat{x}_i^{(2)} = \sign\Bigg( \sum_{j \in \partial x_{i}\cap A^{(2)}} m_{j \to i}^{(2)} \Bigg).
\eeq

\STATE Phase 3. (Estimating workers' reliabilities): Let $A^{(3)}$ be the next $n^{(3)} = w (\log m)(\log\log m)$ queries. For each query $j \in A^{(3)}$, define
$
E_{j}^{(3)} = \mathds{1} \bigg( y_j \neq \prod_{i \in \partial y_j} \hat{x}_i^{(2)} \bigg).
$
For each worker $k \in [w]$ and degree $d \in [D]$, choose the estimate of the noise parameter $\epsilon_{k,d}$ as
\begin{equation}\label{eqn:epkd}
\hat{\epsilon}_{k,d} =\textsf{TRUNC}\left( \frac{\sum_{j \in \partial w_{k}\cap A_d^{(3)}} E_{j }^{(3)} }{| \partial w_{k}\cap A_d^{(3)}|},[\lambda,0.5]\right).
\end{equation}
%where $\epsilon_{k,d}\in[\lambda,0.5]$ some some arbitrary small constant $\lambda>0$.
\STATE Phase 4. (Strong recovery of labels): Let $A^{(4)}$ be the rest $n^{(4)}=n-n^{(1)}-n^{(2)}-n^{(3)}$ queries. For each label $i \in [m]$ and query $j \in \partial x_{i}\cap A^{(4)}$, let
$
m_{j\to i}^{(4)} = y_j \prod_{i^\prime \in \partial y_j \backslash \{i\}} \hat{x}_{i^\prime}^{(2)}
$
and
$
M_{j \to i}^{(4)} = \log \bigg( \frac{1 - \hat{\epsilon}_{w(j),d(j)}}{\hat{\epsilon}_{w(j),d(j)}} \bigg) m_{j \to i}^{(4)},
$
and calculate the final estimate of $x_i$ as
\begin{equation}\label{eqn:x4}
\hat{x}_i^{(4)} = \sign \Bigg( \sum_{j \in \partial x_{i}\cap A^{(4)}}  M_{j \to i}^{(4)} \Bigg).
\end{equation}
\STATE \textbf{Output}: Final estimates $\hat{\vec{x}}:=\hat{\vec{x}}^{(4)}$ for labels.

\end{algorithmic}
\end{algorithm}

The algorithm we propose, presented as Algorithm~\ref{algorithm1}, is composed of four phases: detection of labels, weak recovery of labels, estimation of workers' reliabilities, and strong recovery of labels. We divide the total queries of size $n$ into four sets $A^{(1)}$, $A^{(2)}$, $A^{(3)}$, and $A^{(4)}$ of sizes $|A^{(1)}|=m$, $|A^{(2)}|=m\frac{\log m}{\log\log m}$, $|A^{(3)}|=w(\log m)(\log\log m)$, and $|A^{(4)}|=n-\sum_{l=1}^3 |A^{(l)}|$, and use each set only at the corresponding phase of the algorithm. We assume that the set $A^{(1)}$ is composed of only degree-1 queries each of which asks each label $i\in[m]$. 
The key intuition underlying Algorithm~\ref{algorithm1} is as follows.
\begin{itemize}
\item At Phase 1, we make an initial guess on each label $i\in[m]$ by using the $m$ degree-1 queries in $A^{(1)}$. When $y_{j_i}$ is the answer for the label $x_i$, we define the initial estimate of $x_i$, denoted by $\hat{x}_i^{(1)}$, to be equal to $\sign(y_{j_i})$. Since we assume $\epsilon_{k,d}<1/2$ regardless of the worker $k$, the answer is better than a random guess and we can guarantee the detection of all the labels from this step, i.e., $\Pr{\hat{x}^{(1)}_i\neq x_i}<{1}/{2}$ for all $i \in [m]$. This initial phase helps our estimates on the following phases converge toward the correct labels, and without this phase, the following phases could be nothing more than a random guess. 
\item At Phase 2, we use the estimates $\{\hat{x}_i^{(1)}\}$ from Phase 1 and the next set of $m\left(\frac{\log m}{\log\log m}\right)$ queries in $A^{(2)}$ to generate the second estimates $\{\hat{x}_i^{(2)}\}$ for the labels.  Each query node $j \in \partial x_{i}\cap A^{(2)}$ transmits a `message' $m_{j\to i}^{(2)} = y_j \prod_{i^\prime \in \partial y_j \backslash \{i\}} \hat{x}_{i^\prime}^{(1)}\in\{1,-1\}$ to its neighboring label node $i$, where the message is the estimate of $x_i$ based on the query answer $y_j$ and the estimates $\{\hat{x}_{i'}^{(1)}: {i^\prime \in \partial y_j \backslash \{i\}}\}$ from the previous phase. Then, the $i$-th label node collects all the messages from its neighboring query nodes and does the majority voting to calculate the second estimate $\hat{x}_i^{(2)}$ in~\eqref{eqn:x2}. Note that this phase has resemblance to the `hard-decision decoding algorithm' (Gallager's decoding algorithm) for LDPC codes~\cite{gallager1962low}, where the messages are allowed to take values only from $\{1,-1\}$ and the {\it check node} outputs a message along an edge $\mathsf{e}$ which is the product of all the incoming messages excluding the incoming message along  $\mathsf{e}$~\cite{gallager1962low,richardson2001capacity,bazzi2004exact,zarrinkhat2004threshold}.
We will show that even without any information on the workers' reliabilities, the weak recovery is possible using simple majority voting over the transmitted messages at this phase, i.e., $\Pr{\hat{x}^{(2)}_i\neq x_i}\to 0$ for all $i\in [m]$.
\item Phase 3 estimates the reliability of each worker $k\in[w]$ for each query degree $d\in[D]$ by using the next $|A^{(3)}|=w(\log m)(\log\log m)$ queries. The estimate for $\epsilon_{k,d}$ is generated by checking how many answers in $A^{(3)}$ from a worker $k$ for degree-$d$ queries agree with the weakly recovered labels $\{\hat{x}_i^{(2)}\}$ from Phase 2. We will show that the corresponding estimate $\hat{\epsilon}_{k,d}$  in~\eqref{eqn:epkd} converges to the true noise parameter $\epsilon_{k,d}$ as $m\to\infty$ under the condition that the number of workers $w=o(m/\log \log m)$. %Since the number of degree-$d$ queries assigned to a worker $k$ is larger than $(\log m)(\log\log m)$ with high probability when $w=o(m/(\log\log m)^2)$, by applying Hoeffding's inequality, we can prove the bound in~\eqref{eqn:ehatkd3} on the accuracy of the estimates $\hat{\epsilon}_{k,d}$.
\item Finally, at the last phase, we use the rest set of queries to generate the final estimate $\hat{x}_i^{(4)}$. We do the weighted majority voting for updated messages $\{m_{j\to i}^{(4)}\}_{j\in \partial x_i\cap A^{(4)}}$, where $
m_{j\to i}^{(4)} := y_j \prod_{i^\prime \in \partial y_j \backslash \{i\}} \hat{x}_{i^\prime}^{(2)}
$, with weight equal to $ \log \bigg( \frac{1 - \hat{\epsilon}_{k,d}}{\hat{\epsilon}_{k,d}} \bigg) $
where query $j$ is assigned to worker $k$ and has degree $d$. This phase refines the weakly recovered labels and generates the final estimate $\hat{\vec{x}}^{(4)}$ such that $\Pr{{\hat{\vec{x}}^{(4)}} \neq \vec{x}}\to 0$ as $m\to\infty$. After this phase, the strong recovery is achieved. 
\end{itemize}

\begin{remark}[Time Complexity] As for time complexity, the proposed algorithm takes $O(m\log m)$ time steps. The first phase requires $O(m)$ time steps, and the second phase takes $O( (m\log m)/\log \log m )$ since there are at most $(Dm\log m)/\log\log m$ different $m_{j\to i}^{(2)}$'s. The third phase requires $O(w(\log m)(\log\log m))$ time steps where $w$ is the number of workers. We later assume that $w=o(m/\log\log m)$. The last phase takes $O(m\log m)$ since there are at most $D m\log m$ different $M_{j\to i}^{(4)}$'s. 
\end{remark}

\begin{remark}[Importance of Phase 3--4 in Algorithm~\ref{algorithm1}] Without Phase 3--4, we can still guarantee the weak recovery of labels from Phase 1--2 by using only $\Theta \left(\frac{m\log m}{\log \log m} \right)$ queries (as will be proved in Lemma~\ref{lem:2} in Section~\ref{subsec:pflem2}). However, to guarantee the strong recovery of labels, especially with the exact constant factor as in~\eqref{eqn:thm1}, which depends on the worker reliabilities $\{\epsilon_{k,d}\}$, it is inevitable to estimate the worker reliabilities (Phase 3) and use them as weights in the weighted majority voting for label estimates (Phase 4). In Section~\ref{sec:experiment}, we provide some simulation results that compare the performance of Algorithm~\ref{algorithm1} with and without Phase 3--4 to demonstrate the effectiveness of these phases in strong recovery of labels. 
\end{remark}

%\paragraph{Algorithm 2} The previous algorithm divides the $O(m \log m)$ queries into three sets, $A^{(2)}$, $A^{(3)}$, and $A^{(4)}$. However, $\log m$ is a slowly increasing function, and each step would not have enough queries to converge to the right answer. Therefore, we remove Step 2 and let Step 3 and 4 iterate for sufficiently many cycles using all but the first $m$ degree 1 queries.
\subsection{Theoretical Performance Guarantee} 
\begin{algorithm}[tb]
	\caption{Modification of Algorithm \ref{algorithm1}}
	\label{algorithm2}

\begin{algorithmic}[1]
\STATE \textbf{Data}: 	The observed query answers $\vec{y} \in\{1, -1\}^n$ and the tripartite graph (as in Fig.~\ref{fig:tripartitte}) depicting query design and worker assignment.
\STATE Phase 0. (Removing a few query nodes generating loops) For a fixed $i\in [m]$, consider a graph $G_i$ for the inference of $x_i$ from the root $\hat{x}_i^{(4)}$ of Algorithm~\ref{algorithm1}. To remove loops in $G_i$, eliminate maximally three queries from  $\partial x_i\cap A^{(4)}$ (the set of queries in $A^{(4)}$ connected to $x_i$) and maximally five queries from $\{\partial w_k\}_{k\in \{w(j): j\in\partial x_i \cap A^{(4)}\}} \cap A^{(3)} $ (the set of queries in $A^{(3)}$ connected to any worker who answered any query in $\partial x_i\cap A^{(4)}$). More details on the queries to be removed are described in Algorithm~\ref{algorithm3}. After removing the queries, if $G_i$ still has any loop, claim an error. 
\STATE Phase 1--4 are the same as Algorithm~\ref{algorithm1} except that we generate the estimates $\{\hat{x}^{(1)}_{i''}\}$, $\{\hat{x}^{(2)}_{i'}\}$, $\{\hat{\epsilon}_{k,d,i}\}$, $\hat{x}_i^{(4)}$ only for the nodes that appear in $G_i$ (after the removal of the nodes generating loops). 
\STATE Repeat Phase 0--4 for each $i\in[m]$.
\STATE \textbf{Output}: Final estimates $\hat{\vec{x}}:=\hat{\vec{x}}^{(4)}$ for labels.
%\STATE Step 3. (Estimating workers' reliabilities): Define $A^{(3)}$, $n^{(3)}$, and $E_{j}^{(3)}$ as in Algorithm \ref{algorithm1}. Suppose we want to recover $x_i$ in Step 4. For each worker $k \in [w]$ and degree $d \in [D]$, select the set of queries $\partial \hat w_{k,d,i}^{(3)}$ from $\partial w_{k,d} \cap A^{(3)}$ that will be used in the estimation of the noise parameter $\epsilon_{k,d}$, and let the estimate be
%\begin{equation}\label{eqn:epkdi}
%\hat{\epsilon}_{k,d,i} =\textsf{TRUNC}\left(  \frac{\sum_{j \in \partial \hat w_{k,d,i}^{(3)}} E_{j,i }^{(3)} }{\abs{ \partial \hat w_{k,d,i}^{(3)}}},[\lambda,0.5]\right).
%\end{equation}
%\STATE Step 4. (Strong recovery of labels): Define $A^{(4)}$, $n^{(4)}$, and $m_{j\to i}^{(4)}$ as in Algorithm \ref{algorithm1}. For each label $i \in [m]$, select the set of queries $\partial \hat x_i^{(4)}$ from $\partial x_i^{(4)} \cap A^{(4)}$ that will be used in the estimation of $x_i$, and for each query $j \in \partial \hat x_{i}^{(4)}$ with $w(j) = k$ and $d(j) = d$, let
%$
%M_{j \to i}^{(4)} = \log \bigg( \frac{1 - \hat{\epsilon}_{k,d,i}}{\hat{\epsilon}_{k,d,i}} \bigg) m_{j \to i}^{(4)}.
%$
%Calculate the final estimate of $x_i$ as
%\begin{equation}\label{eqn:x4i}
%\hat{x}_i^{(4)} = \sign \Bigg( \sum_{j \in \partial \hat x_{i}^{(4)}}  M_{j \to i}^{(4)} \Bigg).
%\end{equation}
\end{algorithmic}
\end{algorithm}

Algorithm \ref{algorithm1} can be considered as a type of message-passing algorithm on a factor graph.
The analysis of message-passing algorithms becomes much simpler when the corresponding inference graph is a tree and the messages at each level of the graph are independent. However, if we draw a graph $G_i$ for the inference of $x_i$ for each $i\in[m]$ from the root $\hat x_i^{(4)}$ of Algorithm \ref{algorithm1}, the graph is not perfectly a tree with probability $\omega(\frac{1}{m})$, which cannot be ignored in the error analysis. Instead, by removing a few (constant) number of query nodes connected to $\hat x_i^{(4)}$ or to $\hat{\epsilon}_{k,d}$, we can make the inference graph from the root $\hat{x}_i^{(4)}$ a tree with probability $1-o(1/m)$. 
For the purpose, we modify Algorithm \ref{algorithm1} and add Phase 0 just to remove the query nodes generating loops in the graph $G_i$ for each $i\in[m]$. The modification is summarized in Algorithm \ref{algorithm2}, and the detailed definition of the queries removed from $G_i$ will be given in Algorithm~\ref{algorithm3} in Section~\ref{sec:pfthm2}, where we prove the performance of the Algorithm~\ref{algorithm2}.

Since the set of queries removed from the inference graph $G_i$ could be different for each $i\in[m]$, unlike Algorithm \ref{algorithm1}, the estimation of worker reliabilities $\{\epsilon_{k,d}\}$ can be different for each $i \in [m]$ and we denote the estimates as $\{\hat{\epsilon}_{k,d,i}\}$ to emphasize that the estimate depends on the survived nodes in $G_i$ after removing a few nodes generating loops. 

In Algorithm \ref{algorithm2}, to obtain the final estimate $\hat{x}_i^{(4)}$ for each label $i\in[m]$ we generate the estimates  $\{\hat{x}^{(1)}_{i''}\}$, $\{\hat{x}^{(2)}_{i'}\}$, $\{\hat{\epsilon}_{k,d,i}\}$ only for the nodes that appear in $G_i$ (after the removal of the nodes generating loops). The number of nodes (including query/label/worker nodes) appear in each $G_i$ is bounded by $\Theta\left(\log^3 m \right)$ with probability $1-o(1/m)$, and thus the total time-complexity of Algorithm \ref{algorithm2} is bounded by  $\Theta\left(m \log^3 m \right)$ with probability $1-o(1/m)$.

%We note that the modified algorithm is still computationally efficient, since for each $i\in [m]$ the complexity of the algorithm is $\Theta\left((\log m)^3 \right)$, which can be shown clearly after we describe $G_i$ in more details in the next section. 

We provide performance guarantee for  Algorithm \ref{algorithm2} in Theorem~\ref{thm:thm2} by using the fact that the inference graph $G_i$ is a tree with probability $1-o(1/m)$ after removing a few query nodes. We emphasize that this modification is purely for theoretical purpose. In Section~\ref{sec:experiment} we show through simulations that Algorithm \ref{algorithm1} without modification closely achieves the information-theoretic bounds on the minimum number queries at a finite $m$. 

%The difference is that we exclude the queries of Steps 3 and 4 that create loops in the factor graph of $\hat x_i^{(4)}$. Hence, unlike Algorithm \ref{algorithm1}, the estimation of worker reliabilities $\{\epsilon_{k,d}\}$ can be different for each $i \in [m]$ and we denote the estimates as $\{\hat{\epsilon}_{k,d,i}\}$ to emphasize that the estimate depends on the different set of queries eliminated in the factor graph for $\hat x_i^{(4)}$. The detailed definition of the queries used in Steps 3 and 4 of the algorithm will be given in the next section where we prove the performance of the algorithm. We note that the modified version is still computationally efficient. The performance guarantee of Algorithm \ref{algorithm2} is as below.

\begin{theorem}\label{thm:thm2}
Assume that the number $w$ of workers is $o(m/ \log\log m)$ and each of the labels in $[m]$ is queried at least once by degree-1 queries.  Then, Algorithm \ref{algorithm2} achieves the strong recovery, i.e., $\Pr{\hat{\vec{x}} \neq \vec{x}} \to 0$ as $m\to\infty$, with the information-theoretically optimal number of queries in~\eqref{eqn:thm1}. 
\end{theorem}
\begin{IEEEproof}
The proof of this theorem is provided in Section~\ref{sec:pfthm2}. 
\end{IEEEproof}
%Also, Algorithm~\ref{algorithm1}, which requires only $O(m\log m)$ time steps, achieves the strong recovery of the labels at the optimal number of queries (with an arbitrary small constant scaling gap $\delta>0$ as $m\to\infty$) even without the knowledge of noise parameters $\{\epsilon_{k,d}\}$ of workers. 

\medskip

{\bf Proof sketch: }Even though the full proof is presented in Section~\ref{sec:pfthm2}, here we provide the high-level ideas.
In Phase 1, we use $m$ degree-1 queries to have estimates $\{\hat{x}_{i''}^{(1)}\}$ better than a random guess. Since each label is answered by a worker whose error probability is less than $\frac{1}{2}$, the detection is guaranteed, i.e., $\Pr{\hat{x}_{i''}^{(1)} \neq x_{i''}}<{1}/{2}$ for all ${i''} \in [m]$.

In Phase 2, each label node $i'$ collects messages $\{m_{j\to i'}: j \in \partial x_{i'}\cap A^{(2)}\}$ from its neighboring query nodes in the set $A^{(2)}$ of size $|A^{(2)}|=m\frac{\log m}{\log\log m}$, and provides its second estimate $\hat{x}^{(2)}_{i'}$ by the majority voting over the messages. The message $m_{j\to i'}$ is the estimate of $x_{i'}$ based on the query answer $y_j$ and the estimates $\{\hat{x}_{i''}^{(1)}: {i^{\prime\prime} \in \partial y_j \backslash \{i'\}}\}$ from the previous phase. We show that the probability that $m_{j\to i'}^{(2)}$ is different from the true label $x_{i'}$ is less than 1/2. Thus, the label node collecting the average number of messages, $\Theta\left(\frac{\log m}{\log\log m}\right)$, can correctly recover the true label by simple majority voting with high probability as $m\to\infty$. 

In Phase 3, the error probability of each worker for a degree-$d$ query is estimated as the fraction of the worker's answers that do not match with the weakly recovered label nodes $\{\hat{x}_{i'}^{(2)}\}$.  For this phase, we use a new set of queries in $A^{(3)}$ of size $|A^{(3)}|=w(\log m)(\log\log m)$ where $w$ is the number of workers. Since the number of degree-$d$ queries assigned to a worker $k$ is $\Theta((\log m)(\log\log m))$ with high probability, by applying Hoeffding's inequality, we can prove that each $\hat{\epsilon}_{k,d,i}$ converges to the true noise parameter $\epsilon_{k,d}$ with the maximal error of $O(1/(\log \log m)^{1/4})$ as $m\to\infty$. The condition on the number workers $w=o(m/\log\log m)$ is required to make $|A^{(3)}|$ negligible compared to the overall number of queries $n=\Theta(m\log m)$.

In Phase 4, each query node $j \in \partial x_{i}\cap A^{(4)}$ transmits an `updated message' $m_{j\to i}^{(4)} = y_j \prod_{i^\prime \in \partial y_j \backslash \{i\}}  \hat{x}_{i^\prime}^{(2)}\in\{1,-1\}$, which is the estimate of $x_i$, to the $i$-th label node. Then, the $i$-th label node applies a weight $\log \frac{({1 - \hat{\epsilon}_{w(j),d(j),i}})}{({\hat{\epsilon}_{w(j),d(j),i}}) )}$ on each message and does the weighted majority voting on the collected messages. %Note that the weight is the same for every query $j$ having the same assigned worker and the query degree. 
By using the accuracy of the estimates $\{\hat{\epsilon}_{k,d,i}\}$ proved in the analysis of Phase 3, we can show that the weighted majority voting succeeds in recovering the true label vector ${\vec{x}}$ with high probability when the sample complexity $n$ satisfies~\eqref{eqn:thm1}.

\section{Experiments}\label{sec:experiment}

In this section, we report experimental results that illustrate the tightness of our theorems and the optimality of the proposed algorithm both for synthetic and real datasets. The first subsection demonstrates that our theoretical finding on the strong recovery is valid and tight even in non-asymptotic regimes, where we assume a simple error model such that the error probability varies over workers but not on query degrees. The comparison between different types of queries, including XOR, repetition, and homogeneous queries, are presented in the next subsection with a fair error model called $d$-coin flip model. Lastly, we apply the XOR query and the proposed algorithm to a real crowdsourcing platform, Amazon Mechanical Turk, and substantiate the practicality of the proposed algorithm.

\subsection{Performance of the Proposed Algorithm}
\begin{figure}
  \centering
  \includegraphics[width=8cm]{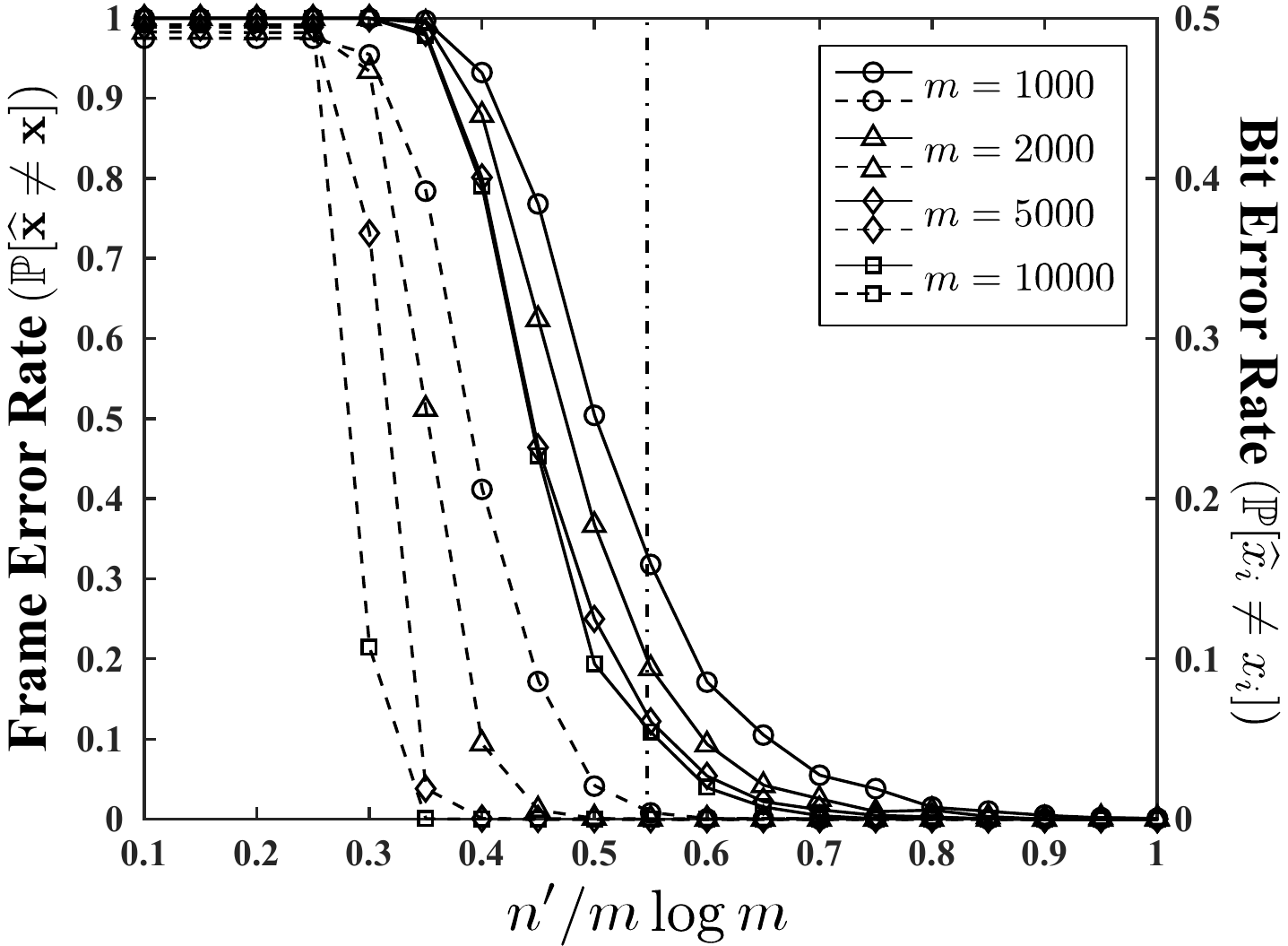}
  \caption{\normalsize Frame error rate, $\mathbb{P}[{\hat{\vec{x}}\neq \vec{x}}]$, (solid lines) and bit error rate $\mathbb{P}[\hat{x}_i \neq x_i]$ (dashed lines)  vs. (normalized) number of queries for four different values of $m$. Dash-dotted vertical lines are the information-theoretic limits for each $m$ given by Theorem \ref{thm:thm1}.}
  \label{fig:simulation1}
\end{figure}
We first show through simulation that Algorithm \ref{algorithm1} achieves the bound established in Theorem \ref{thm:thm1} for finite $m$. We set $m$ to have values of 1000, 2000, 5000, and 10000, while fixing the number of workers to $w = 100$. We vary the query degrees by letting them randomly sampled from $3$ to $6$ with equal probabilities, but we use a simple error model where the error probability does not depend on the query degree, as in the case of communication systems. Equal number of workers have the error probabilities, each from $\{0.02, 0.04, \cdots, 0.20\}$. Different from the original Algorithm~\ref{algorithm1} where each phase is conducted only once, to increase the accuracy of the estimates at a finite $m$, Phase 2 is iterated 10 times, and then Phases 3--4 are together iterated 10 times. Also, apart from the $m$ degree-1 queries used in Phase 1 of Algorithm \ref{algorithm1}, the queries are not divided into separates sets $A^{(2)}$, $A^{(3)}$, $A^{(4)}$, but all the queries are used together in both the iterations. In all the following experiments, we use this modified proposed algorithm. Denoting the number of queries used in the iterations as $n^\prime := n - m$, we measured the frame error rate, $\Pr{\hat{\vec{x}} \neq \vec{x}}$, and the average bit error rate, $\Pr{\hat{x}_i \neq x_i}$, of the proposed algorithm with respect to $n^\prime$ by repeating the experiment 1000 times. The result is shown in Fig \ref{fig:simulation1}. The solid lines and dotted lines indicate frame and bit error rate, respectively, and the information-theoretic limit given by Theorem \ref{thm:thm1} is shown with the vertical dash-dotted line. We observe that the proposed computationally-efficient algorithm for XOR query nearly achieves the optimal sample complexity even when the noise parameters are unknown, and the algorithm converges faster at bigger $m$. The bit error rate drops at much smaller sample complexity than the limit, but it does not imply the perfect recovery.

\begin{figure}
  \centering
  \includegraphics[width=8cm]{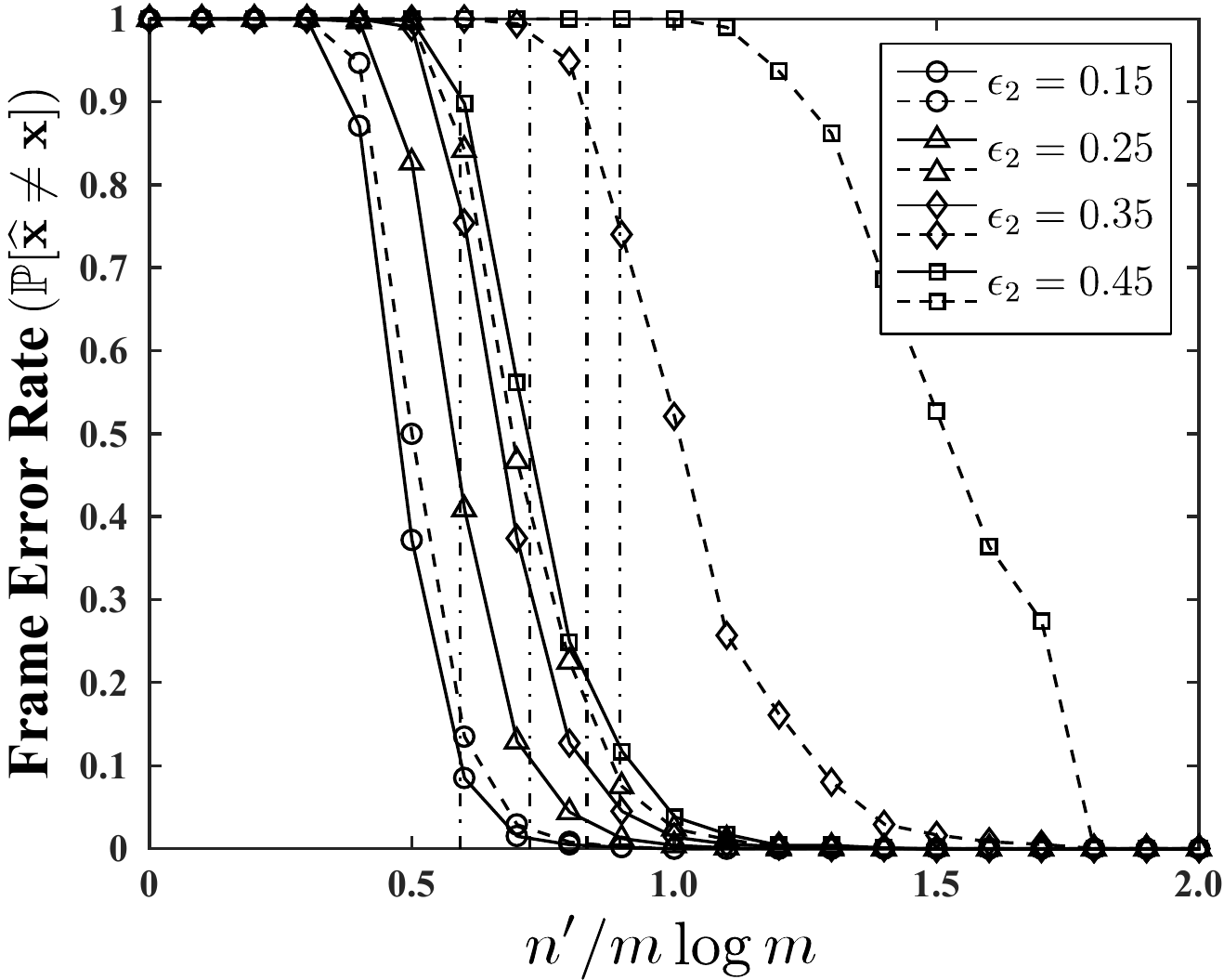}
  \caption{\normalsize Frame error rate $\mathbb{P}[{\hat{\vec{x}}\neq \vec{x}}]$ vs. (normalized) number of queries for four different combinations of worker reliabilities when the correct labels are inferred by Alg.~\ref{algorithm1} with (solid lines) and without (dashed lines) Phases 3--4. Dash-dotted vertical lines correspond to the information-theoretic limits of each case from $\epsilon_2=0.15$ (left most) to $\epsilon_2=0.45$ (right most).}
  \label{fig:simulation2}
\end{figure}

In the next experiment, we compare the performances of Algorithm~\ref{algorithm1} with and without Phases 3--4, respectively, to validate the significance of estimating worker reliabilities. We consider a setting where the number of object labels $m=5000$ and the number of workers $w=100$. All queries have a fixed degree $d=4$ except the first $m$ degree-1 queries. The error probability of half of the workers is fixed to $\epsilon_1 = 0.05$, but that of the other half is varied to $\epsilon_2 = 0.15, 0.25, 0.35, 0.45$. We assume that the degree-1 queries are assigned only to the first half of the workers. The error rates averaged over 1000 trials are summarized in Figure~\ref{fig:simulation2}, where the solid lines correspond to the proposed algorithm and the dashed lines correspond to the proposed algorithm without Phases 3--4. As the variance of the workers' reliability increases, i.e., for a higher $\epsilon_2$, the gap between the solid line and the dashed line increases. This shows that Phases 3--4 of the proposed algorithm, where worker reliabilities are estimated and used to refine the weakly recovered labels, become more important as the difference between workers' reliabilities is greater.

\subsection{Comparison of Different Schemes with a Fair Error Model}\label{subsec:exp1}
\begin{figure}
  \centering
  \includegraphics[width=8cm]{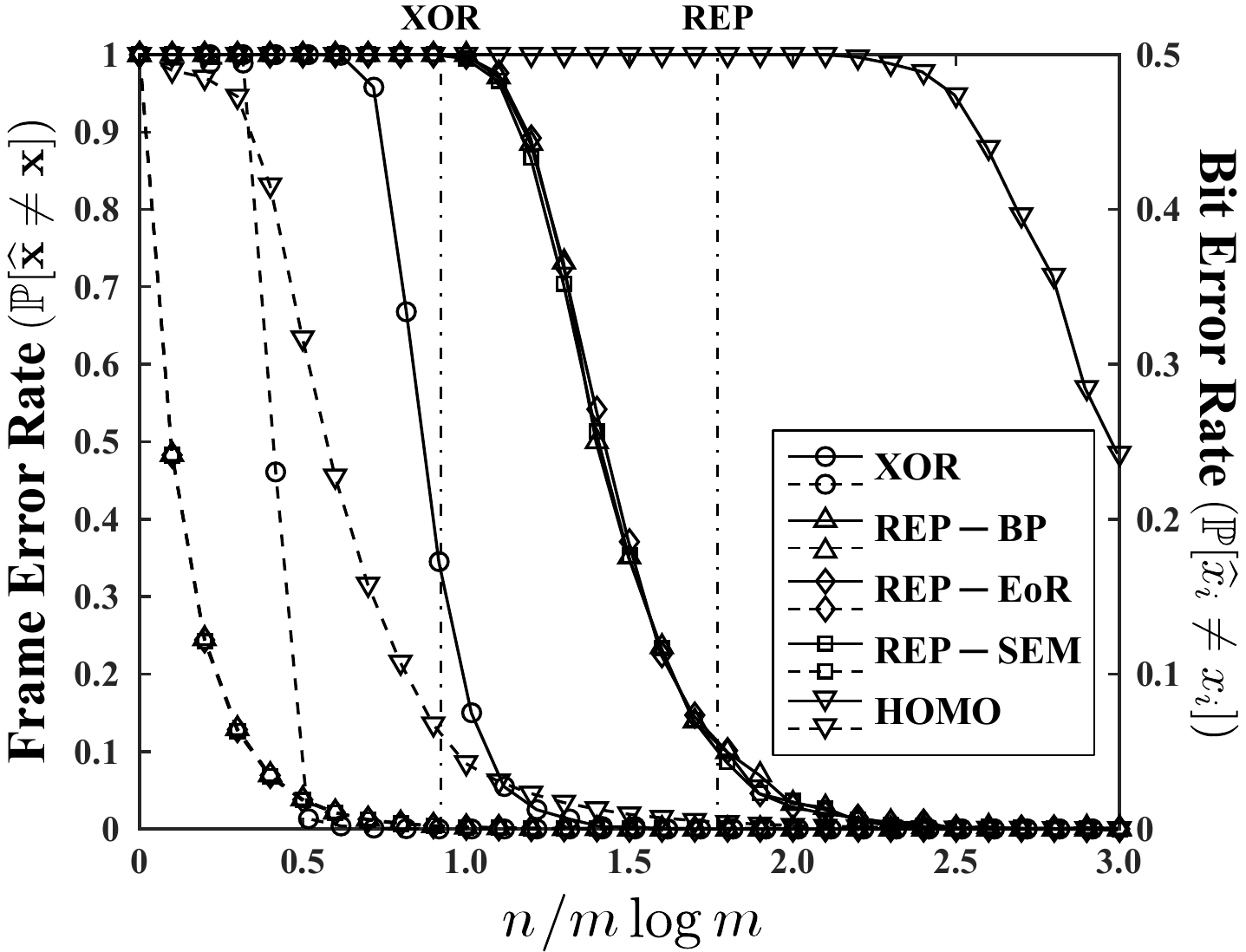}
  \caption{\normalsize Frame error rate, $\mathbb{P}[{\hat{\vec{x}}\neq \vec{x}}]$ (solid lines), and bit error rate $\mathbb{P}[\hat{x}_i \neq x_i]$ (dashed lines) vs. (normalized) number of queries for five different pairs of query types and inference algorithms. Dash-dotted vertical lines are information-theoretic limits of XOR (left) and REP query (right).}
  \label{fig:simulation3}
\end{figure}
In the first experiment, we assumed that the error probability of a worker remains constant regardless of query degree. However, we often encounter an application, e.g. crowdsourcing, such that the worker's error probability depends on query degree, or in general, querying method. Hence, in order to fairly compare different querying methods and the inference algorithms, a proper error model is required that describes the change in worker's error probability with respect to querying method. In this subsection, we propose such a noise model named $d$-coin flip model, and compare XOR queries with repetition (REP) and homogeneous (HOMO) queries.

In the $d$-coin flip model, given a query degree $d$, a worker $k$ independently flips $d$ coins, each of which gives head with probability $\epsilon_k$. When head has occurred, the worker makes wrong decision about the item corresponding to the coin, and after gathering $d$ decisions by the query operation, the final answer for the query is made. For example, the error probability of a degree-$d$ XOR query under the $d$-coin flip model becomes 
\beq
\epsilon_{k,d} = \sum_{\substack{l\in[1:d],\\l\text{ odd}}}{d\choose l}\epsilon_k^{l}(1-\epsilon_k)^{d-l}=\frac{1-(1-2\epsilon_k)^d}{2}.
\eeq
In this experiment, we chose the reliability parameter $\epsilon_k$ of each worker randomly from the set $\{0.010,0.020,\cdots,0.100\}$.

We use the number of object labels $m=5000$ and the number of workers $w=100$ as in the previous experiment. The query degrees of XOR query and HOMO query are uniformly sampled from 3 to 6, and the answers are collected with the $d$-coin flip noise model. As usual, the XOR query has additional $m$ degree-1 queries for the initialization, and the proposed algorithm is applied. For HOMO query, we apply the inference algorithm based on spectral clustering and local refinement, which has been shown to be order-wise optimal in~\cite{ahn2018hypergraph}.
For REP query, three state-of-the-art algorithms, based on belief propagation (BP)~\cite{karger2014budget}, spectral-EM (SEM)~\cite{zhang2014spectral}, and ratio of eigenvector (EoR)~\cite{dalvi2013aggregating} are applied. 
Figure~\ref{fig:simulation3} shows the frame and bit error rates measured in 1000 trials versus (normalized) number of queries for the five different pairs of query types and inference algorithms. The $m$ degree-1 queries of XOR querying is also included in the plot. The result indicates the benefit of using XOR queries with high degrees over REP and HOMO queries in reducing the sample complexity  for exact recovery. Although all the three algorithms for REP query nearly achieve the optimal sample complexity, the large gap between the fundamental limits of XOR query and that of repetition query  (plotted by vertical lines)  makes XOR query more efficient than REP query in terms of strong recovery. However, REP queries show better performance than XOR query in terms of bit error rate especially when the number of queries are small. The HOMO query turns out to be the worst among the three query types.

\/*
\subsubsection{Efficiency of XOR Queries with a More General Error Model}
\begin{figure}
  \centering
  \includegraphics[width=\columnwidth/2]{generalerror}
  \caption{\normalsize Information theoretic limits~\eqref{eqn:thm1} on the normalized number $\frac{n}{m\log m}$ of degree-$d$ XOR queries required for strong recovery of $m$ labels for a general error model~\eqref{eqn:generalerror} with the parameter $a$.}
  \label{fig:generalerror}
\end{figure}
We conduct another set of simulations to demonstrate the effectiveness of XOR querying for a more general error model. The new error model we consider is as follows. Suppose that a worker $k$ is not sure about $p_k\in(0,1)$ portion of item labels. For a degree-$d$ query, we assume that a worker provides the correct answer when he knows all the $d$ items and gives a random answer when he is not sure about any of $d$ items. The error probability is then equal to $\epsilon_{k,d}=(1-(1-p_k)^d)/2$. Note that this model is equivalent to the $d$-coin-flip model, which we considered in the previous experiment, when $p_k=2\epsilon_k$ and $\epsilon_k$ is the probability of making a wrong decision for an item. 
We remark that this new error model well approximates the empirical human error probability from MTurk: we had on average $p=2\epsilon=2*0.06$ (for $d=1$) and for $d=4$ the empirical human error probability was 0.19, which is close to $(1-(1-2*0.06)^4)/2=0.20$.
We further generalize this error model such that there exists a probability $f(d)$, increasing in $d$, that a worker provides a wrong answer even when he knows all the $d$ items (due to calculation error), and modify the error model for a degree-$d$ query as 
\beq\label{eqn:generalerror}
\epsilon_{k,d}=f(d)(1-p_k)^d+\frac{1}{2}(1-(1-p_k)^d)\quad \text{ for }\quad f(d)=0.5\tanh(a\cdot d)
\eeq 
where $a\geq 0$ controls how fast the error probability increases in $d$. The simulation results with this generalized error model are shown in Figure~\ref{fig:generalerror}. In the plots, $x$ and $y$ axes are the probability $p=p_k$ and query-degree $d$, respectively, and the z-axis is the required number $n$ of queries in~\eqref{eqn:thm1} normalized by $m\log m$. It is shown that for all the cases there exist some regimes where multi-degree XOR queries ($d>1$) are more efficient than the repetition query ($d=1$) for this general error model.
*/

\subsection{Real Experiment: Crowdsourcing}
\begin{figure}
  \centering
  \includegraphics[width=12cm]{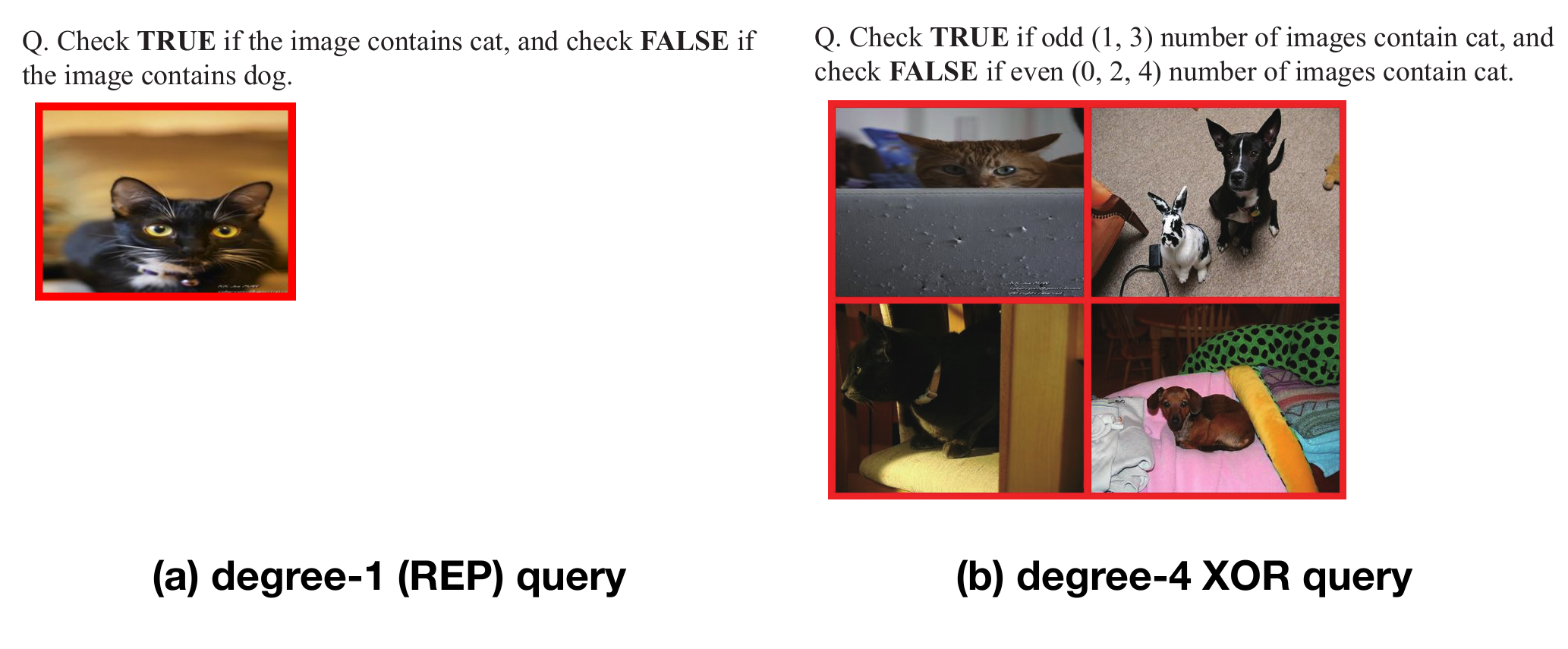}
  \caption{\normalsize Examples of degree-1 and degree-4 queries.}
  \label{fig:query}
\end{figure}

\begin{figure}
  \centering
  \includegraphics[width=8cm]{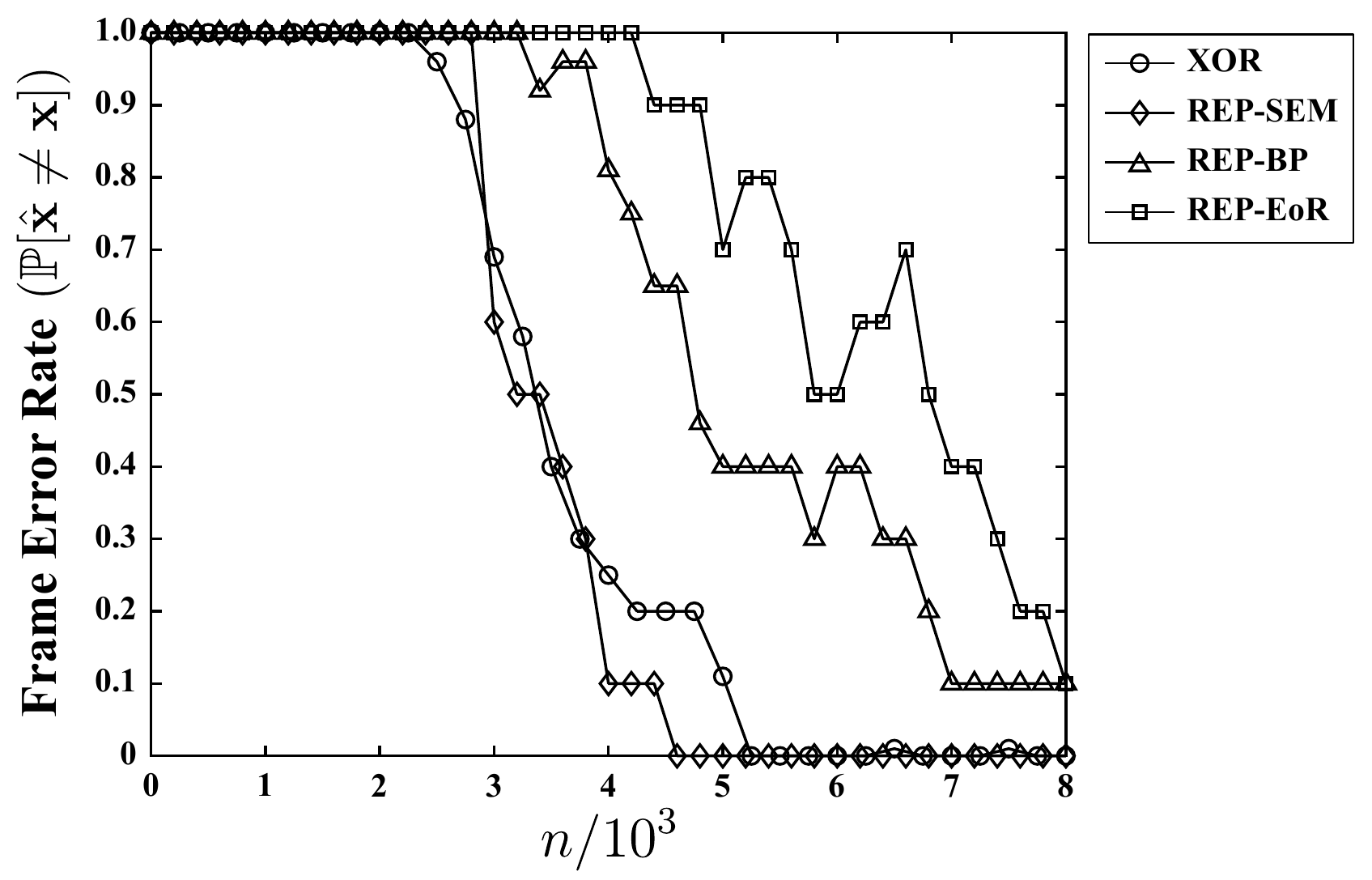}
  \caption{\normalsize Probability of error $\mathbb{P}[{\hat{\vec{x}}\neq \vec{x}}]$ in strong recovery vs. number of queries for four different algorithms applied to real dataset from human workers.}
  \label{fig:simulation5}
\end{figure}
In this subsection, we assess the practicality of XOR query and the proposed algorithm by applying them to a real crowdsourcing platform. We designed a binary classification task using 600 images of dogs and cats sampled from ImageNet~\cite{russakovsky2014imagenet}, and collected data from the workers in Amazon Mechanical Turk. Each human intelligent task (HIT) was designed to include 20 degree-1 queries and 20 degree-4 XOR queries. The examples of each query type are shown in Figure~\ref{fig:query}. We designed 400 HITs and assigned each of them to 400 workers. The reward of each query was fixed to \$0.01 regardless of the query degree. For the collected data, we compare how many queries $n$ are required to recover all the 600 labels when we use only degree-1 (REP) queries or we use 20 degree-4 XOR queries with additional 5 degree-1 queries from each HIT. For the REP queries we apply three different inference algorithms (BP, SEM, EoR) as in the previous experiment, and for the XOR queries we apply Algorithm~\ref{algorithm1}. We repeat this experiment 10 times and plot the empirical error rate in Figure~\ref{fig:simulation5}.
%For the collected data, we compared how many answers are required to recover all the 600 labels for four different pairs of query types and inference algorithms, 1) degree-4 XOR query and the proposed algorithm, 2) degree-1 (REP) query and BP, 3) REP query and Spec.+EM, and 4) REP query and EoR.
The result shows that XOR query with the proposed algorithm outperforms REP query with BP or EoR algorithms, but it has similar performance to REP query with SEM algorithm. 
The theoretical limits~\eqref{eqn:thm1} on the required number of queries calculated with the empirical noise parameters $\{\epsilon_{k,d}\}$ from the real dataset are  2300 for (degree-4) XOR query and 5200  for degree-1 REP query. 
In the experiment, the proposed algorithm with XOR query does not closely match this limit at the finite $m=600$, and thus the gain from the XOR query is not clearly seen. 
The reason could be that the number $m$ of images we use for the experiment is not large enough to meet the asymptotic information-theoretic limit.

\section{Proof of Theorem~\ref{thm:thm2}: Analysis of Algorithm~\ref{algorithm2}}\label{sec:pfthm2}
In this section, we provide the proof of Theorem~\ref{thm:thm2}. The proof is separated into two parts: in the first part (Section~\ref{sec:goodevents_state}), we will introduce a sequence of ``good events'' related to the query design and assignment that occurs with high probability; in the second part (Section~\ref{sec:pfthm2good}), we will consider the answers we get from the queries and analyze the error probability of Algorithm~\ref{algorithm2} conditioned on the good events.  In Section~\ref{sec:factor}, we start by introducing a factor graph and the related definitions and notations that are required to define and analyze the good events and the performance of Algorithm~\ref{algorithm2} 

\subsection{Factor Graph}\label{sec:factor}
\begin{figure}
  \centering
  \includegraphics[width=3cm]{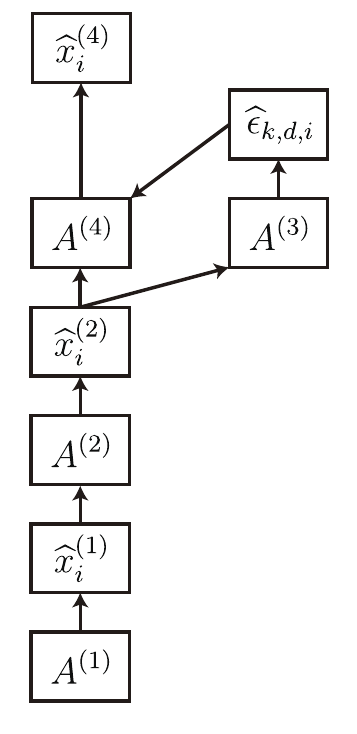}
  \caption{\normalsize The dependency between the types of nodes.}
  \label{fig:dependency}
\end{figure}

\begin{figure}
  \centering
  \includegraphics[width=12cm]{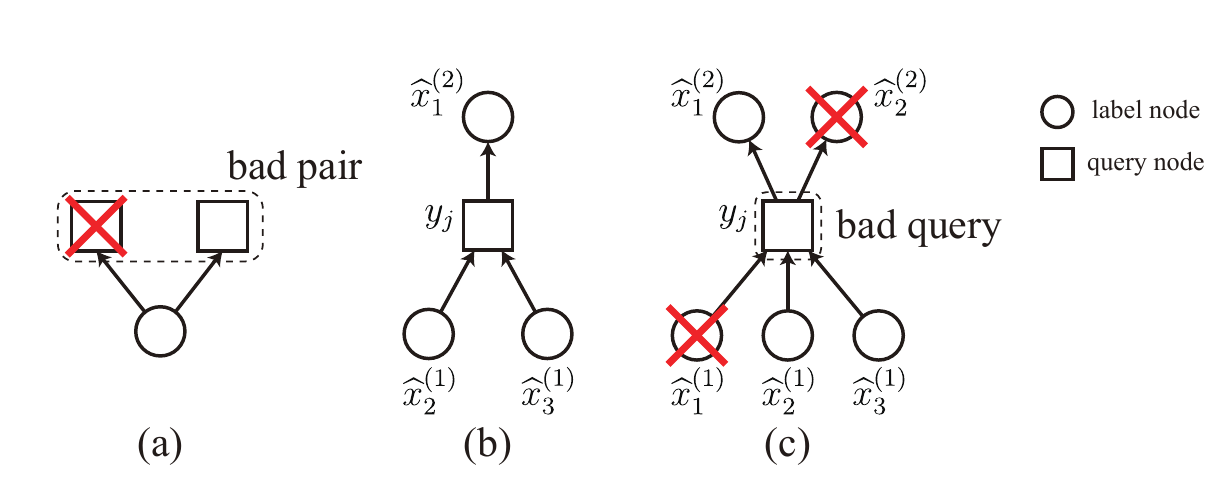}
  \caption{\normalsize Factor graph for (a) bad pair (b) non-bad query (c) bad query. The figure is drawn assuming the degree of Phase 2 query $y_j$ is $3$ and $\partial y_j = \{1, 2, 3\}$.}
  \label{fig:bad_nodes}
\end{figure}

When analyzing a message-passing type of algorithm, it is helpful to draw the factor graph, along the edges of which messages are transmitted. Thus, before we start proving Theorem~\ref{thm:thm2}, we introduce the symbolic meanings of the nodes in factor graph $G_i$ from the root $\hat x_i^{(4)}$ of Algorithm~\ref{algorithm2}.

First, we introduce three types of nodes present in the graph $G_i$: label node, query node, and worker node. Label nodes can be used to represent the three different estimates $\hat x_{i^\prime}^{(1)}$, $\hat x_{i^\prime}^{(2)}$, and $\hat x_{i^\prime}^{(4)}$ that are made on label $x_{i^\prime}$ for $i^\prime \in [m]$ in Phase 1, 2, and 4 of Algorithm~\ref{algorithm2}, respectively. Note that we depict the three estimates by different nodes although they have the same index $i^\prime$. There are four types of query nodes depending on in which phase of the algorithm (from Phase 1 to 4) the query is used. Each query node used in Phase 1, 2 and 4 outputs a message, $\sign(y_j)$, $m_{j \to i}^{(2)}$, $M_{j \to i}^{(4)}$, to each neighboring label node $i\in \partial y_j$, respectively, while each query node used in Phase 3 outputs a message $E_j^{(3)}$ to its neighboring worker node $w(j)$. 
Lastly, the worker node outputs a message $\hat \epsilon_{k,d,i}$, which is the estimate on worker reliability made in Phase 3, to a Phase 4 query node.

An edge in factor graph represents the dependency between the nodes, or the direction to which the message is transmitted. We next figure out the dependency between the types of the nodes introduced above. The estimates on labels made in Phase 1, 2, and 4 are based on the messages from queries in the same phase. In order to calculate the messages, query nodes use the estimates on label nodes made in the previous phase, but as a special case, Phase 4 query nodes use the reliability estimates for worker nodes made in Phase 3 as well as the estimates on label nodes made in Phase 2 . To calculate the reliability estimates, worker nodes use the messages from Phase 3 query nodes. We summarize the dependency between the types of nodes in Figure \ref{fig:dependency}, where the vertical position represents the level of nodes in the overall factor graph $G_i$. %Lower level nodes are used before the higher level nodes in the Algorithm~\ref{algorithm2}. 
For example, Phase 2 label node is located in a higher level than Phase 2 query node, and Phase 3 and 4 query nodes are located in the same level.

We will explain two different cases where a cycle is created in the factor graph. The first case is when a label node is connected to two different query nodes in a higher level, as depicted in Figure \ref{fig:bad_nodes}-(a). We call the two query nodes a \textit{bad pair}. When there exists a bad pair, we will delete one of the query nodes from the factor graph. The second case is when a query node sends messages (or connected) to more than one label nodes in a higher level. We call such a query a \textit{bad query}. We define the \textit{badness} of a bad query as the number of higher level nodes it is connected to, which is always larger than one by the definition. A non-bad query is depicted in Figure \ref{fig:bad_nodes}-(b) for comparison, with one higher-level label node and $d-1$ lower-level label nodes, where $d$ is the query degree.
%\sout{A bad query, depicted in Figure \ref{fig:bad_nodes}-(c), on the other hand, is connected to two higher-level label nodes and $d$ lower-level label nodes, since all the the estimates on the $d$ label nodes are required to calculate two different messages to the two higher-level label nodes. If there exists such a bad query, we will delete one of the higher-level label nodes connected to the bad query to remove the cycle, and one of the lower-level label nodes that is not required anymore.}
A bad query with badness equal to two is depicted in Figure \ref{fig:bad_nodes}-(c). Unlike a non-bad query, a bad query is always connected to $d$ lower-level label nodes since all the estimates on the $d$ labels are required to calculate the messages to any two or more different higher-level label nodes. 
%We can remove the bad query by deleting $(\textit{badness}-1)$ number of higher-level nodes 
If there exists a bad query, we delete $(\textit{badness}-1)$ number of higher-level nodes so as to make only one node left in the higher level, and also delete the lower-level node that has the same index with the remaining higher-level node since it is not required anymore.

In the next subsection, we will define ``good events" on the factor graph $G_i$ such that there are only a few bad pairs and a few bad queries so that by removing a few nodes we can make $G_i$ a tree with high probability.

\subsection{Good Events}\label{sec:goodevents_state}
\begin{figure*}
  \centering
  \includegraphics[width=15cm]{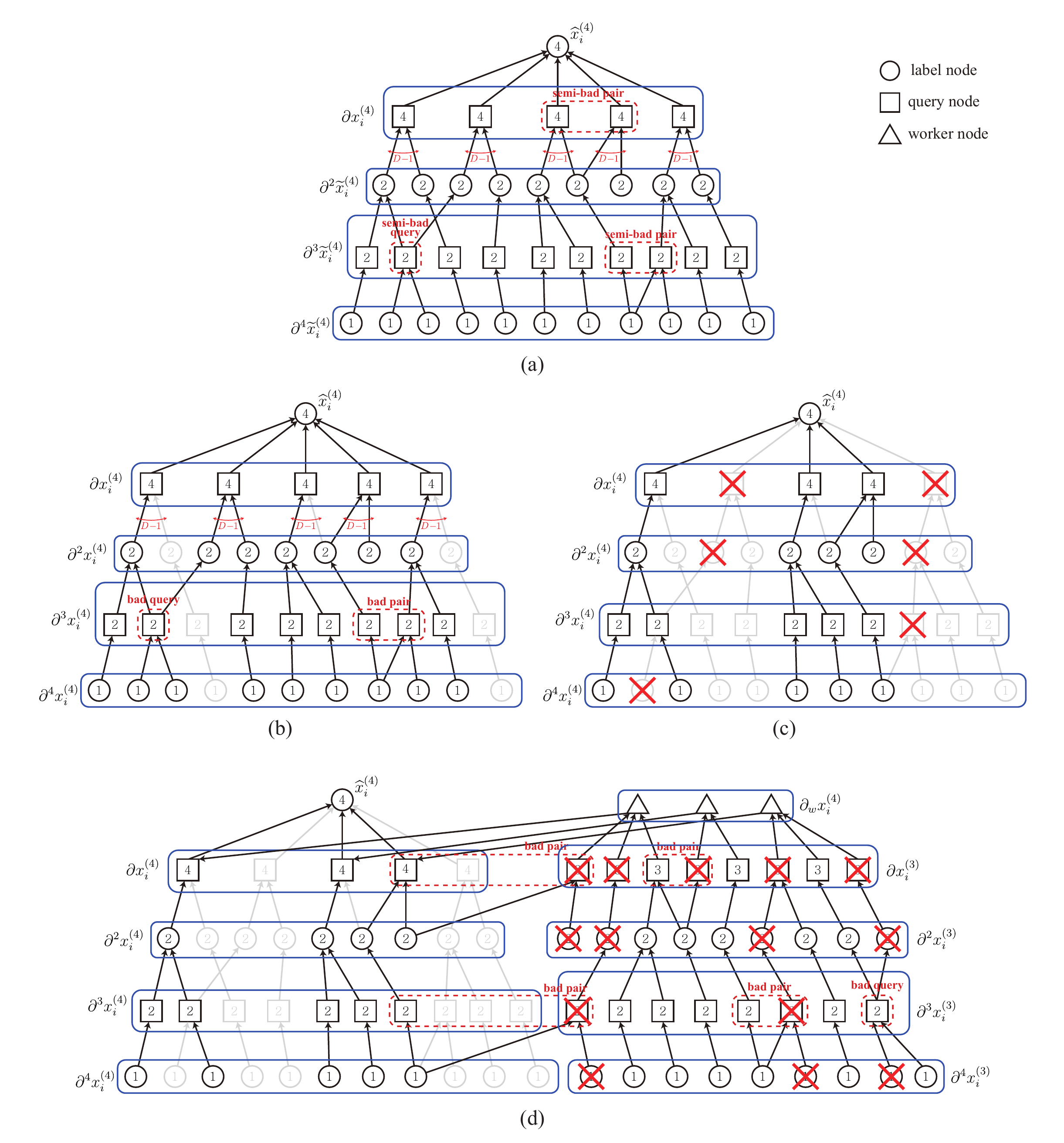}
  \caption{The numbers in nodes indicate the phase of Algorithm~\ref{algorithm2} to which each node belongs. (a) An example of the factor graph generated after step 4). (b) The factor graph after obtaining the true degree of each query in $\partial x_i^{(4)}$ and selecting $\partial y_j \setminus \{i\}$ from $\partial \tilde y_j$. (c) The factor graph after the removal of bad queries and bad pairs in step 5). (d) The final factor graph with an illustration of the bad pairs and bad queries made in steps 6) -- 10).}
  \label{fig:inf_graph}
\end{figure*}

We explain the detailed sequential process of drawing $G_i$ in this subsection. We will also introduce ``good events" related to each step of the process. 
The good events we define occur with probability exceeding $1 - o(1/m)$ so that it does not affect the error analysis. 
Basically, there are two types of good events regarding the random graph $G_i$. The events of the first type assert that there are not many bad pairs and bad queries in $G_i$ and by removing a few query nodes from $G_i$ we can make the remaining graph a tree with high probability. The second type is related to the number of queries connected to $\hat{x}_i^{(4)}$ and to $\hat{\epsilon}_{k,d,i}$ that are helpful in correctly estimating the label $x_i$ and the noise parameters $\{\epsilon_{k,d}\}$, respectively.

We explain the first type of good events related to independence of messages used for the estimation of $x_i$ and $\{\epsilon_{k,d}\}$, by considering the random process of generating a factor graph $G_i$ from the root $\hat{x}_i^{(4)}$ according to the random query design and assignment model, explained in Section~\ref{sec:model_qdesign}.
The process of generating $G_i$ and how the bad pairs and bad queries appear in the graph are depicted in Fig.~\ref{fig:inf_graph}.

\begin{enumerate}
\item At the first step, a subset of Phase 4 queries are connected to $x_i$ according to the random query design model, where each query in $A^{(4)}$ is connected to $x_i$ with probability $\bard/m$, independently, where $\bard=\sum_{d=1}^D d\Phi_d$ is the average query degree. We denote the set of Phase 4 queries connected to $x_i$ as 
\beq\label{eqn:partial x_i^4}
\partial x_i^{(4)}:=\partial x_i\cap A^{(4)}.
\eeq
\item Next, each query node $y_j$ for $j\in \partial x_i^{(4)}$ needs to select $d-1$ label nodes from $[m]\backslash\{i\}$ uniformly at random. The query degree $d$ of each query is sampled from the distribution $\left\{ \frac{d \Phi_d}{\bar d} \right\}$ instead of $\{ \Phi_d\}$, since conditioned on that $y_j$ is already connected to $x_i$ the degree distribution of each query increases proportional to $d$.
We divide this label-selection process for $y_j$ into the following three steps: (a) each query node $y_j$ first selects $D-1$ label nodes randomly  from $[m]\backslash\{i\}$, and we denote this set of label nodes by $\partial \tilde y_j$, (b) the true degree $d$ is sampled from the distribution $\left\{ \frac{d \Phi_d}{\bar d} \right\}$, and (c) $d-1$ label nodes are randomly sampled from $\partial \tilde y_j$ to construct $\partial y_j \setminus \{i\}$. We assume that only step (a) is performed here, and steps (b) and (c) are postponed until step 5).

The reason we introduced the set $\partial \tilde y_j$ is to acquire symmetry between the queries in $\partial x_i^{(4)}$ regardless of the true degree $d$. We also note that $\partial \tilde y_j$ is a kind of latent variable that is not revealed to Algorithm~\ref{algorithm2}, and some of the nodes in $\partial \tilde y_j$ may not be contained in $G_i$. Thus, $\partial \tilde y_j$ will be used not for the operations of Algorithm~\ref{algorithm2}, but only for the purpose of proving the performance of the algorithm later.

The set of label nodes selected in this step is denoted by
\beq\label{eqn:partial2 x_i^4}
\partial^2 \tilde x_i^{(4)} := \bigcup_{j \in \partial x_i^{(4)}} \partial \tilde y_j .
\eeq
In this step, if a pair of queries $(j_1,j_2)$ in $\partial x_i^{(4)}$ selects the same label from $[m]\backslash\{i\}$, i.e. $\partial \tilde y_{j_1} \cap \partial \tilde y_{j_2} \neq \emptyset$, we call this pair a {\it semi-bad pair}. We call it a \textit{semi}-bad pair, since it may or may not be a bad pair after we choose the actual $\partial y_{j_1}(\subset\partial \tilde y_{j_1} )$ and $\partial y_{j_2}(\subset\partial \tilde y_{j_2} )$. Also, note that a bad pair is always a semi-bad pair. We will later show that there is at most one semi-bad pair in $\partial x_i^{(4)}$, thus at most one bad pair in  $\partial x_i^{(4)}$. 
\item At the third step, we start considering Phase 2 query assignment, and specify the set of queries connected to each label node in $\partial^2 \tilde x_i^{(4)}$. First, the degree of each query in Phase 2 query set $A^{(2)}$ is sampled from $\{ \Phi_d \}$. We then separate the queries in $A^{(2)}$ into two sets depending on whether a query selects at least one label node in $\partial^2 \tilde x_i^{(4)}$ or not. Denote by $\partial^3 \tilde x_i^{(4)}$ the set of queries in $A^{(2)}$ that select at least one label node in $\partial^2  \tilde x_i^{(4)}$
\beq\label{eqn:partial3 x_i^4}
\partial^3 \tilde x_i^{(4)} := \left\{\bigcup_{i' \in \partial^2 \tilde x_i^{(4)}} \partial x_{i'} \right\}\bigcap A^{(2)}.
\eeq
Each query of degree-$d$ in $A^{(2)}$ is included in the set $\partial^3 \tilde x_i^{(4)}$ independently with probability $1-\frac{{m-|\partial^2 \tilde x_i^{(4)}| \choose d}}{{m\choose d}}$.

\item At the next step, each query in $\partial^3 \tilde x_i^{(4)}$ selects its labels conditioned on that it should select at least one label from $\partial^2 \tilde x_i^{(4)}$. Hence, we first let each degree-$d$ query in $\partial^3 \tilde x_i^{(4)}$ randomly select one label from $\partial^2 \tilde x_i^{(4)}$, and then let them select the remaining $d-1$ labels uniformly at random from the remaining $m - 1$ labels in $[m]$. If a query selects more than one label from $\partial^2 \tilde x_i^{(4)}$, we call such a query a \textit{semi-bad query}. The badness of a semi-bad query is defined the same as that of a bad query.  Also, if a pair of queries selects the same label not in $\partial^2 \tilde x_i^{(4)}$, the pair is called a semi-bad pair as before. We will show that there are at most one semi-bad query with badness equal to two, and at most one semi-bad pair in $\partial^3 \tilde x_i^{(4)}$. We will also show that there is no semi-bad query with badness larger than two.   Let us define $\partial^4 \tilde x_i^{(4)}$ as the index set of Phase 1 label nodes connected to Phase 2 query nodes in $\partial^3 \tilde x_i^{(4)}$. When there is no semi-bad query, it is equal to
\beq\label{eqn:partial4 x_i^4}
\partial^4 \tilde x_i^{(4)} = \left\{\bigcup_{j \in \partial^3 \tilde x_i^{(4)}} \partial y_j \right\} \setminus \partial^2 \tilde x_i^{(4)},
\eeq
but when there is a semi-bad query, $\partial^4 \tilde x_i^{(4)}$ may include some labels in $\partial^2 \tilde x_i^{(4)}$. The factor graph constructed up to this point is depicted in Figure~\ref{fig:inf_graph}-(a).
\item Now, we perform the steps (b) and (c) of the label-selection process for queries in $\partial x_i^{(4)}$ as introduced in step 2). After this, we will be given the set of label nodes in $\partial^2 \tilde x_i^{(4)}$ that are actually connected to the queries in $\partial x_i^{(4)}$ in $G_i$. We denote the set by
\begin{equation}
\partial^2 x_i^{(4)} := \bigcup_{j \in \partial x_i^{(4)}} \partial y_j \setminus \{i\}.
\end{equation}
%\hl{\sout{We remove the child nodes of $\partial^2 \tilde x_i^{(4)} \setminus \partial^2 x_i^{(4)}$ except the child nodes of $\partial^2 x_i^{(4)}$ from the factor graph} 
We remove from the factor graph the child nodes of $\partial^2 \tilde x_i^{(4)}$ that is not connected to $\partial^2 x_i^{(4)}$, and define the survived nodes in $\partial^3 \tilde x_i^{(4)}$ and $\partial^4 \tilde x_i^{(4)}$ as $\partial^3 x_i^{(4)}$ and $\partial^4 x_i^{(4)}$, respectively. The resulting factor graph after this step is depicted in Figure~\ref{fig:inf_graph}-(b). We check whether the semi-bad queries and semi-bad pairs are still present in the factor graph, and if so, since they are bad queries or bad pairs, we remove some of the queries in $\partial x_i^{(4)}$ and the child nodes of them to eliminate the bad queries and bad pairs. There are at most one bad pair in $\partial x_i^{(4)}$, one bad pair and one bad query with badness equal to two in $\partial^3 x_i^{(4)}$, and thus by removing at most three queries from $\partial x_i^{(4)}$ we can make the remaining graph a tree. An example of this process is depicted in Figure~\ref{fig:inf_graph}-(c). After the removal of the nodes, the definition of the sets $\partial x_i^{(4)}, \partial^2 x_i^{(4)},\partial^3 x_i^{(4)},\partial^4 x_i^{(4)} $ are updated to include only the survived nodes at each level of the factor graph. 

%When removing nodes in $G_i$, the definition of the sets without tilde signs are changed to only refer to the survived nodes.

\item The remaining steps are related to the estimation of workers' reliabilities $\{\epsilon_{k,d}\}$. We first specify the set of workers who answer for each query in $\partial x_i^{(4)}$. Remind that each query is assigned to a worker randomly selected from $[w]$. Let us define the set of $(k,d)$ for which the worker $k\in[w]$ is assigned to at least one degree-$d$ query in $\partial x_i^{(4)}$, i.e.,
\beq\label{eqn:partialw x_i^4}
\begin{split}
\partial_w x_i^{(4)}=&\{(k,d)\in[W]\times[D]: w(j)=k, \;\;d(j)=d\text{ for some }j\in\partial x_i^{(4)} \}.
\end{split}
\eeq
We can also interpret $\partial_w x_i^{(4)}$ as the set of worker nodes connected to $\partial x_i^{(4)}$.
%From this point, for simplicity, we will regard a worker of answering $D$ different degrees of queries as $D$ different workers answering queries of each fixed degree and consider $\partial_w x_i^{(4)}$ as the set of workers answering queries in $\partial x_i^{(4)} $. For a worker $(k,d)\in \partial_w x_i^{(4)}$, the probability that the answer is incorrect is $\epsilon_{k,d}$.

\item The next step is to assign Phase 3 queries to the worker nodes in $\partial_w x_i^{(4)}$. Let us define $\partial x_i^{(3)}$ as the set of queries in $A^{(3)}$ assigned to any worker in $\partial_w x_i^{(4)}$, i.e.
\beq\label{eqn:partial x_i^3}
\partial x_i^{(3)} = \bigcup_{(k, d) \in \partial_w x_i^{(4)}} \partial w_{k,d}^{(3)},
\eeq
where $\partial w_{k,d}^{(3)}=\partial w_k\bigcap A_d^{(3)}$ is the set of Phase 3 queries assigned to the worker node $(k,d)\in \partial_w x_i^{(4)}$.

\item Next, each query in $\partial x_i^{(3)}$ selects $d$ labels from $[m]$ uniformly at random, where $d$ is sampled from the degree distribution $\{\Phi_d\}$. Let us define the set of label nodes selected by $\partial x_i^{(3)}$ as
\beq\label{eqn:partial2 x_i^3}
\partial^2 x_i^{(3)} := \set{ \bigcup_{j \in \partial x_i^{(3)}} \partial y_j }.
\eeq
We will prove that at most one query in $\partial x_i^{(3)}$ selects a label in $\partial^2 \tilde x_i^{(4)} \cup \partial^4 \tilde x_i^{(4)}$. This good event implies that we should delete at most one query in $\partial x_i^{(3)}$ to remove any bad pair or bad query made by the following three cases. First, if any query in $\partial x_i^{(3)}$ selects a label in $\partial^2 x_i^{(4)}$, the query forms a bad pair with a query in $\partial x_i^{(4)}$. We remove all such bad pairs by removing the corresponding query in $\partial x_i^{(3)}$ and thus make $\partial^2 x_i^{(4)} \cap \partial^2 x_i^{(3)} = \emptyset$. Second, if any query in $\partial x_i^{(3)}$ selects a label in $\partial^4 x_i^{(4)}$, a bad query is created in $\partial^3 x_i^{(4)}$ that selects a label from both $\partial^2 x_i^{(4)}$ and $\partial^2 x_i^{(3)}$. Third, if any query in $\partial x_i^{(3)}$ selects a label in $\partial^2 \tilde x_i^{(4)} \setminus \partial^2 x_i^{(4)}$ or in $\partial^4 \tilde x_i^{(4)} \setminus \partial^4 x_i^{(4)}$ that is connected to a semi-bad pair or a semi-bad query in $\partial^3 \tilde x_i^{(4)}$, the semi-bad pair or the semi-bad query can turn back to a bad pair or a bad query in $G_i$. For both the second and third case, we remove the corresponding query in $\partial x_i^{(3)}$. Note that in these cases, it suffices to remove only one query in $\partial x_i^{(3)}$ regardless of the badness of the bad query.  %By doing so, we can effectively guarantee that there is no Step 2 query that selects a label both from $\partial^2 x_i^{(4)}$ and $\partial^2 x_i^{(3)}$. %, since Step 2 queries other than $\partial^3 \tilde x_i^{(4)}$ cannot select a label from $\partial^2 x_i^{(4)}$.
We can then guarantee that there is no Phase 2 query that selects a label from both $\partial^2 x_i^{(4)}$ and $\partial^2 x_i^{(3)}$, since no Phase 2 query other than $\partial^3 \tilde x_i^{(4)}$ can select a label from $\partial^2 x_i^{(4)}$.
We also prove that there is at most one bad pair in $\partial x_i^{(3)}$ with probability $1 - o(1/m)$. Hence, in total we remove at most $2$ query nodes from $\partial x_i^{(3)}$ in this step.

\item In this step, we assign Phase 2 queries to $\partial^2 x_i^{(3)}$. Since we have already finished assignment of the queries in $\partial^3 \tilde x_i^{(4)}$ in step 4), here we consider only the queries in $A^{(2)} \setminus \partial^3 \tilde x_i^{(4)}$. Again, we separate the queries in $A^{(2)} \setminus \partial^3 \tilde x_i^{(4)}$ into two sets depending on whether a query selects at least one label node in $\partial^2 x_i^{(3)}$ or not, and define $\partial^3 \tilde x_i^{(3)}$ to be the set of queries in $A^{(2)}\backslash\partial^3 \tilde x_i^{(4)}$ that select at least one label node in $\partial^2 x_i^{(3)}$, i.e.,
\beq\label{eqn:partial3 x_i^3}
\partial^3 \tilde x_i^{(3)} := \left\{\bigcup_{i' \in \partial^2 x_i^{(3)}} \partial x_{i'} \right\}\bigcap\left\{ A^{(2)} \backslash \partial^3 \tilde x_i^{(4)}\right\}.
\eeq
If there exists a label in $\partial^2 x_i^{(3)}$ that also belongs to $\partial^2 \tilde x_i^{(4)}$ or $\partial^4 \tilde x_i^{(4)}$ (there is at most one such label {as stated in step 8)}), this label is connected to some query in $\partial^3 \tilde x_i^{(4)}$.
Thus, some of the Phase 2 queries that are connected to $\partial^2 x_i^{(3)}$ may not be contained in $\partial^3 \tilde x_i^{(3)}$, and we define
\beq\label{eqn:partial3 x_i^3_2}
\partial^3 x_i^{(3)} := \left\{\bigcup_{i' \in \partial^2 x_i^{(3)}} \partial x_{i'} \right\}\bigcap A^{(2)}
\eeq
as the set of Phase 2 queries that select at least one label node in $\partial^2 x_i^{(3)}$. Note that the query-to-label assignment for the queries in $\partial^3 x_i^{(3)} \setminus \partial^3 \tilde x_i^{(3)}$ has already done in step 4).

\item We then make each query in $\partial^3 \tilde x_i^{(3)}$ select one label randomly  from $\partial^2 x_i^{(3)} \setminus \partial^2 \tilde x_i^{(4)}$. We have to exclude the labels in $\partial^2 \tilde x_i^{(4)}$, since any query that selects $\partial^2 \tilde x_i^{(4)}$ is already in $\partial^3 \tilde{x}_i^{(4)}$ and $\partial^3 \tilde{x}_i^{(3)} \cap \partial^3 \tilde{x}_i^{(4)}=\emptyset$ by definition. 
%selected in $\partial^3 \tilde x_i^{(3)}$ cannot select a label from $\partial^2 \tilde x_i^{(4)}$ by the definition.  
Lastly, each query in $\partial^3 \tilde x_i^{(3)} $ with degree $d$ selects the $d-1$ labels uniformly at random from the remaining nodes in $[m]$ except $\partial^2 \tilde x_i^{(4)}$. At most one query in $\partial^3 \tilde x_i^{(3)}$ selects a label in $\partial^4 \tilde x_i^{(4)} \setminus \partial^2 x_i^{(3)}$ (the set difference accounts for the case where there is a label in $\partial^2 x_i^{(3)}$ that also belongs to $\partial^4 \tilde x_i^{(4)}$), and by removing the corresponding query in $\partial x_i^{(3)}$ if necessary, we can exclude the cases where a query in $\partial^3 \tilde x_i^{(3)}$ forms a bad pair with a query in $\partial^3 x_i^{(4)}$ or with a query in $\partial^3 x_i^{(3)} \setminus \partial^3 \tilde x_i^{(3)}$. We also prove that there is at most one bad pair and at most one bad query with badness equal to two in $\partial^3 \tilde x_i^{(3)}$, and there is no bad query with badness larger than two. Thus, in this step, we remove at most three queries in $\partial x_i^{(3)}$, and in total at most five queries are removed from $\partial x_i^{(3)}$ to remove all the cycles. Figure~\ref{fig:inf_graph}-(d) shows an example of the inference graph and the nodes to be eliminated to remove the loops.

\end{enumerate}

\medskip

\begin{algorithm}[tb]
	\caption{Phase 0 of Algorithm \ref{algorithm2}}
	\label{algorithm3}

\begin{algorithmic}[1]
\STATE If there is a bad pair such that both of the queries are in $\partial x_i^{(4)}$, remove one of the query in the bad pair.
\STATE If there is a bad query in $\partial^3 x_i^{(4)}$ that is connected to two label nodes in $\partial^2 x_i^{(4)}$, remove a node in $\partial x_i^{(4)}$ that is the parent of one of the two nodes.  If there is a bad query in $\partial^3 x_i^{(4)}$ with badness larger than two, claim an error.
\STATE If there is a bad pair such that both of the queries are in $\partial^3 x_i^{(4)}$, remove a node in $\partial x_i^{(4)}$ that is the parent of one of the two query nodes.
\STATE If there is a bad pair such that one query is in $\partial x_i^{(4)}$ and the other is in $\partial x_i^{(3)}$, remove the query in $\partial x_i^{(3)}$ in the bad pair.
\STATE If there is a bad query in $\partial^3 x_i^{(4)}$ that is connected to two label nodes each in $\partial^2 x_i^{(4)}$ and $\partial^2 x_i^{(3)}$, remove the node in $\partial x_i^{(3)}$ that is the parent of the node in $\partial^2 x_i^{(3)}$.
\STATE If there is a bad pair such that both of the queries are in $\partial x_i^{(3)}$, remove one of the query in the bad pair.
\STATE If there is a bad pair such that one query is in $\partial^3 x_i^{(4)}$ and the other is in $\partial^3 x_i^{(3)}$, remove a node in $\partial x_i^{(3)}$ that is the parent of the query node in $\partial^3 x_i^{(3)}$.
\STATE If there is a bad query in $\partial^3 x_i^{(3)}$ that is connected to two label nodes in $\partial^2 x_i^{(3)}$, remove a node in $\partial x_i^{(3)}$ that is the parent of one of the two nodes.  If there is a bad query in $\partial^3 x_i^{(3)}$ with badness larger than two, claim an error.
\STATE If there is a bad pair such that both of the queries are in $\partial^3 x_i^{(3)}$, remove the query in $\partial x_i^{(3)}$ that is the parent of one of the two query nodes. 
\STATE Claim an error if we have to remove more than three queries in $\partial x_i^{(4)}$ or more than five queries in $\partial x_i^{(3)}$ from the above steps. 
\end{algorithmic}
\end{algorithm}

With the discussions made above, we can now state Phase 0 of the Algorithm~\ref{algorithm2} in an explicit way as in Algorithm~\ref{algorithm3}. The first three steps of Algorithm \ref{algorithm3} are related to step 5), the fourth to the sixth steps are related to step 8), and the seventh to the ninth steps are related to step 10). We claim that Algorithm \ref{algorithm3} succeeds with probability $1 - o(1/m)$.

\begin{lemma}\label{lem:app:pfLemmagood1}
By removing at most three queries in $\partial x_i^{(4)}$ and at most five queries in $\partial x_{i}^{(3)}$ as described in Algorithm \ref{algorithm3}, all the bad pairs and the bad queries in $G_i$ are removed with probability $1-o(1/m)$, and $G_i$ becomes a tree.
\end{lemma}

\begin{IEEEproof}
The proof of this lemma is provided in Appendix~\ref{app:pfLemmagood1}.
\end{IEEEproof}

The tree structure of $G_i$ implies the following four independence results (or good events) that we will use in the analysis of error probability.

\begin{enumerate}[label=(\roman*)]
\item For each $i^\prime \in \partial^2  x_i^{(4)} \cup \partial^2  x_i^{(3)}$, the messages $\{ m^{(2)}_{j \to i^\prime}\}_{j \in \partial x_{i^\prime} \cap A^{(2)}}$, which are used to generate $\hat{x}_{i'}^{(2)}$, are independent.

\item The estimators $\{\hat x_{i^\prime}^{(2)} \}_{i^\prime \in \partial^2  x_i^{(3)} \cup \partial^2  x_i^{(4)}}$ are independent.

\item For each $(k, d) \in \partial_w x_i^{(4)}$, the messages $\{E_{j}^{(3)}\}_{j \in \partial w_{k,d}^{(3)}}$, which are used to generate $\hat \epsilon_{k,d,i}$, are independent.

\item The messages $\{ m_{j \to i}^{(4)}\}_{j \in \partial x_i^{(4)}}$ and the estimators $\{ \hat \epsilon_{k,d,i} \}_{(k,d) \in \partial_w x_i^{(4)}}$, which are used to generate the final estimate $\hat{x}_i^{(4)}$,  are all independent.
\end{enumerate}

%Given the factor graph $G_i$ from the root $\hat{x}_i^{(4)}$ we are now ready to discuss a sequence of good events to guarantee $\Pr{\hat{x}_i^{(4)} \neq x_i} = o(1/m)$ for all $i \in [m]$, which is sufficient to prove the strong recovery of labels, i.e., $\Pr{{\hat{\vec{x}}^{(4)}} \neq \vec{x}}\to 0$ as $m\to\infty$, by the simple union bound. 
%We first state the sequence of good events related to independence needed for analyzing the performance of estimators $\{\hat{x}_{i'}^{(2)}\}$, $\{\hat\epsilon_{k,d,i}\}$ and $\hat{x}_i^{(4)}$ of Algorithm~\ref{algorithm2}. 
%The intersection of above good events holds with high probability after removing only a few query nodes.

We move on to the second type of good events that are required to generate accurate estimates $\{\hat{x}_i^{(2)}\}$, $\{\hat \epsilon_{k,d,i}\}$ and $\hat{x}_i^{(4)}$. We focus on controlling the number of three types of nodes, defined as {\it good labels},  {\it perfect queries}, and {\it good queries}.

\begin{definition}\label{def:perfectq} A label node $i' \in \partial^2 \tilde{x}_i^{(4)}\bigcup  \partial^2 x_i^{(3)}$ is called a {\it good label} if $|\partial x_{i'}\cap A^{(2)}|$, the number of queries in $A^{(2)}$ that have selected $x_{i^\prime}$,  is $\Theta\left(\frac{\log m}{\log\log m}\right)$. We call a query $j\in \partial x_i^{(4)}$ a perfect query if all the $D-1$ label nodes in $\partial \tilde y_j$ (before selecting the actual $d-1$ neighboring label nodes $\partial y_j\backslash \{i\}$ from $\partial \tilde y_j$) are good labels and it is not a parent node of any semi-bad pair or semi-bad query. We also call a query $j\in \partial x_i^{(3)}$ a good query if all the $d$-neighboring label nodes $\partial y_j $ are good labels.
\end{definition}

The set of good events for the accuracy of the estimator $\hat{x}_i^{(4)}$ is as below.
\begin{enumerate}[label=(\roman*)]\setcounter{enumi}{4}
\item The number of Phase 4 queries $A^{(4)}$ connected to $x_i$ is $|\partial x_i^{(4)}|=\Theta(\log m)$, and the number of perfect queries among $\partial x_i^{(4)}$ is at least $|\partial x_i^{(4)}|-C_{11}\log\log m - 3$ for some constant $C_{11}>0$. % after removing at most 3 queries that have generated cycles in $G_i$.
\item For each worker $(k,d)\in \partial_w x_i^{(4)}$, the number of Phase 3 queries assigned to $(k,d)$ is $|\partial w_{k,d}^{(3)}|=\Theta((\log m)(\log\log m))$, and there are at least $|\partial w_{k,d}^{(3)}|-C_{12}\log\log m - 6$ good queries in $\partial w_{k,d}^{(3)}$ for some constant $C_{12}>0$. 
\end{enumerate}

Note that the average numbers of $|\partial x_i^{(4)}|$ and $|\partial w_{k,d}^{(3)}|$ are $\Theta(\log m)$ and $\Theta((\log m)(\log\log m))$, respectively, and under good events (\romannumeral 5) and (\romannumeral 6) almost all the queries in $|\partial x_i^{(4)}|$ and $|\partial w_{k,d}^{(3)}|$ are perfect/good queries, respectively. We now claim that the intersection of the above two good events also holds with high probability.
\begin{lemma}\label{lem:app:pfLemmagood2}
The intersection of good events (\romannumeral 5)--(\romannumeral 6) holds with probability $1-o(1/m)$.
\end{lemma}
\begin{IEEEproof}
The proof of this lemma is provided in Appendix~\ref{app:pfLemmagood2}. 
\end{IEEEproof}

\subsection{Proof of Theorem~\ref{thm:thm2}: Error Analysis under Good Events}\label{sec:pfthm2good}
We prove Theorem~\ref{thm:thm2} conditioned on the intersection of good events (\romannumeral 1)--(\romannumeral 6) defined in the previous section, which occurs with probability $1-o(1/m)$ by Lemma \ref{lem:app:pfLemmagood1} and \ref{lem:app:pfLemmagood2}. Therefore, once we prove $\Pr{\hat{x}_i^{(4)} \neq x_i\big|\text{good events}} = o(1/m)$ for all $i \in [m]$, then it implies $\Pr{\hat{x}_i^{(4)} \neq x_i} = o(1/m)$ for all $i \in [m]$ and $\Pr{{\hat{\vec{x}}^{(4)}} \neq \vec{x}}\to 0$ as $m\to\infty$ by union bound. 

To prove Theorem~\ref{thm:thm2}, we first state Lemma~\ref{lem:1}--\ref{lem:3}, each of which describes the accuracy of the estimates $\{\hat{x}_{i''}^{(1)}\},\{\hat{x}_{i'}^{(2)}\},\{\hat{\epsilon}_{k,d,i}\}$, respectively, after Phase 1--3 of Algorithm~\ref{algorithm2}, respectively. 
In proving lemmas, we assume that we have removed at most three queries in $\partial x_i^{(4)}$ and at most five queries in $\partial x_{i}^{(3)}$ and obtained the independence as stated in Lemma \ref{lem:app:pfLemmagood1}.
%In the proof of the lemmas, we assume a set of good events that are necessary to prove the accuracy of the estimates. We will prove that the intersection of all the good events occur with high probability, $1-o(1/m)$. 
\begin{lemma}\label{lem:1}
After Phase 1 of Algorithm~\ref{algorithm2}, where we use total $n^{(1)}=m$ number of degree-1 queries to have an initial estimate on the labels, the detection of the labels $\{x_{i''}\}$ is guaranteed with $\{\hat{x}_{i''}^{(1)}\}$ in~\eqref{eqn:x1}, i.e., 
\beq
p^{(1)} := \Pr{\hat{x}_{i''}^{(1)} \neq x_{i''}} < 1/2, \quad \forall i'' \in \partial^{4} \tilde{x}_i^{(4)}\cup  \partial^{4} x_i^{(3)}.
\eeq
\end{lemma}
We dropped the subscript $i''$ in $p^{(1)}$ since every label node is queried exactly once by the first $m$ degree-1 queries and thus the accuracy of the estimates is the same for all $i'' \in \partial^{4} \tilde{x}_i^{(4)}\cup  \partial^{4} x_i^{(3)}$. Since every answer is better than a random guess by the assumption that $\epsilon_{k,1}<1/2$ for all $k\in[w]$, this lemma is obvious.

\begin{lemma}\label{lem:2} %After Step 2 of Algorithm~\ref{algorithm2}, the weak recovery of the labels $\{x_{i'}\}$ is guaranteed with $\{\hat{x}_{i'}^{(2)}\}$ in~\eqref{eqn:x2} for good labels $i' \in \partial^2 \tilde{x}_i^{(4)}\bigcup  \partial^2 x_i^{(3)}$ such that $|\partial x_{i'}\cap A^{(2)}|=\Theta\left(\frac{\log m}{\log\log m}\right)$; for some $p^{(2)}=o(1/\log m)$,
%\beq\label{eqn:bdp2}
%\Pr{\hat{x}_{i'}^{(2)} \neq x_{i'}} \leq p^{(2)}=o(1/\log m)
%\eeq 
%for all the good labels $i' \in \partial^2 \tilde{x}_i^{(4)}\bigcup  \partial^2 x_i^{(3)}$.

After Phase 2 of Algorithm~\ref{algorithm2} where we use total $n^{(2)}=m\left(\frac{\log m}{\log \log m}\right)$ number of queries to generate the second estimates on labels, the weak recovery of the labels $\{x_{i'}\}$ is guaranteed with $\{\hat{x}_{i'}^{(2)}\}$ in~\eqref{eqn:x2} for good labels $i' \in \partial^2 \tilde x_i^{(4)}\bigcup  \partial^2 x_i^{(3)}$ such that $|\partial x_{i'}\cap A^{(2)}|=\Theta\left(\frac{\log m}{\log\log m}\right)$; for some $p^{(2)}=o(1/\log m)$,
\beq\label{eqn:bdp2}
\Pr{\hat{x}_i^{(2)} \neq x_i} \leq p^{(2)}=o(1/\log m)
\eeq 
for all the good labels $i' \in \partial^2 \tilde{x}_i^{(4)}\bigcup  \partial^2 x_i^{(3)}$.

\end{lemma}

\begin{lemma}\label{lem:3} After Phase 3 of Algorithm~\ref{algorithm2} where we use $n^{(3)}=w (\log m)(\log \log m)$ number of queries to estimate reliability of workers, conditioned on good events  (\romannumeral 5)--(\romannumeral 6), for every $(k,d)\in\partial_w x_i^{(4)}$ the estimate $\hat{\epsilon}_{k,d,i}$ of the reliability of the $k$-th worker for the degree-$d$ query used for the estimation of $x_i$ satisfies
\beq\label{eqn:ehatkd3}
|\hat{\epsilon}_{k,d,i}-\epsilon_{k,d}|=O(1/(\log\log m)^{1/4})
\eeq
with probability at least $1-o(1/m)$, when the number of workers $w=o(m/ \log\log m)$.
%\begin{equation}
%\Mean*{\left( \frac{\hat{\epsilon}_{k,d}}{1 - \hat{\epsilon}_{k,d}} \right)^{\frac{1}{2}}} \to \left( \frac{\epsilon_{k,d}}{1 - \epsilon_{k,d}} \right)^{\frac{1}{2}} \quad\text{and}\quad
%\Mean*{\left( \frac{1 - \hat{\epsilon}_{k,d}}{\hat{\epsilon}_{k,d}} \right)^{\frac{1}{2}}} \to \left( \frac{1 - \epsilon_{k,d}}{\epsilon_{k,d}} \right)^{\frac{1}{2}}.
%\end{equation}
%as $m\to\infty$.
\end{lemma}

The proofs of Lemmas \ref{lem:2}--\ref{lem:3} are provided at the end of this section.

Finally, we are ready to bound the error probability $\Pr{\hat{x}_{i}^{(4)} \neq x_i}$. We start from not conditioning any good events.
By conditioning the number $a=|\partial x_i^{(4)}|$ of queries in $A^{(4)}$ connected to $x_i$, the error probability can be written as
\beq\label{eqn:err_conda}
\begin{split}
\Pr{\hat{x}_{i}^{(4)} \neq x_i}
&=\sum_{a\in [n^{(4)}]} \Pr{|\partial x_i^{(4)}|=a} \Pr{\hat{x}_{i}^{(4)} \neq x_i\Big| |\partial x_i^{(4)}|=a} \\
&= \sum_{a\in [n^{(4)}]} {n^{(4)} \choose a} \left( \frac{\bar d}{m} \right)^{a} \left(1 - \frac{\bar d}{m} \right)^{n^{(4)} - a} \Pr{\hat{x}_{i}^{(4)} \neq x_i\Big| |\partial x_i^{(4)}|=a},
\end{split}
\eeq
since each query in $A^{(4)}$ is connected to $x_i$ independently with probability $\bard/m$. We first show that $|\partial x_i^{(4)}|$ is $\Theta(\log m)$ with probability $1-o(1/m)$. Since the total number of $A^{(4)}$ queries is $\Theta(m\log m)$ and the number of labels is $m$, the average number of queries connected to each label is $\Theta(\log m)$.
\newcounter{lemmagood}
\setcounter{lemmagood}{\thelemma}
\begin{lemma}\label{good1_1}
Let us define the event $S_1$ as
\begin{center}
$S_1$: $c_1 \log m < |\partial x_i^{(4)}| < C_1 \log m$
\end{center}
for some constants $C_1>c_1>0$. Then, we have $\Pr{S_1} \geq 1 - o(1/m)$.
\end{lemma}
Note that the event $S_1$ is a part of the good event  (\romannumeral 5), which will be proved in Appendix~\ref{app:pfLemmagood1}. 
By Lemma~\ref{good1_1}, we have 
\beq\label{eqn:err_conda1}
\begin{split}
\Pr{\hat{x}_{i}^{(4)} \neq x_i}&\leq \sum_{a\in S_1} {n^{(4)} \choose a} \left( \frac{\bar d}{m} \right)^{a} \left(1 - \frac{\bar d}{m} \right)^{n^{(4)} - a} \Pr{\hat{x}_{i}^{(4)} \neq x_i\Big| |\partial x_i^{(4)}|=a}+o(1/m).
\end{split}
\eeq

Next, we analyze $ \Pr{\hat{x}_{i}^{(4)} \neq x_i\Big| |\partial x_i^{(4)}|=a}$ for $ a\in S_1$. From this point, the error analysis is conditioned on the intersection of the good events (\romannumeral 1)--(\romannumeral 6), which hold with probability $1-o(1/m)$ by Lemma \ref{lem:app:pfLemmagood1} and \ref{lem:app:pfLemmagood2}.
Remind that Phase 4 estimate of $x_i$ is defined as $\hat{x}_i^{(4)} = \sign \Bigg( \sum_{j \in \partial x_{i}^{(4)}}  M_{j \to i}^{(4)} \Bigg)$, where $
M_{j \to i}^{(4)} = \log \bigg( \frac{1 - \hat{\epsilon}_{w(j),d(j),i}}{\hat{\epsilon}_{w(j),d(j),i}} \bigg) m_{j \to i}^{(4)}.
$
Conditioned on the sequence of good events, by Lemma~\ref{lem:3} we can have the estimates on worker reliabilities $\{\hat{\epsilon}_{k,d,i}\}$ satisfying $|\hat{\epsilon}_{k,d,i}-\epsilon_{k,d}|=O(1/(\log\log m)^{1/4})$ for every $(k,d)\in\partial_w x_i^{(4)}$, with probability $1-o(1/m)$. Let us define this event as 
\beq
\begin{split}
\mathcal{E}:=\{&{|\hat{\epsilon}_{k,d,i}-\epsilon_{k,d}|=O(1/(\log\log m)^{1/4})}\text{ for every }(k,d)\in\partial_w x_i^{(4)}\}.
\end{split}
\eeq
So, by conditioning on the event $\mathcal{E}$ and by using the independence of the messages $\{M_{j\to i}\}_{i\in \partial j\in\partial x_i^{(4)}}$ implied by the good event  (\romannumeral 4), we get
\beq
\begin{split}\label{eqn:xi4analysis_newMarkov}
 \Pr{\hat{x}_{i}^{(4)} \neq x_i\Big| |\partial x_i^{(4)}|=a}&= \Pr{\sum_{j \in \partial {x}_{i}^{(4)}} -x_{i} M_{j \to i}^{(4)} \geq 0}\\
 &=\prod_{(k,d)\in \partial_w x_i^{(4)}} \prod_{\{j \in \partial {x}_{i}^{(4)}:w(j)=k,d(j)=d\}} \Mean {e^{-tx_{i} M_{j \to i}^{(4)}}\Bigg |\mathcal{E}}+o(1/m).
\end{split}
\eeq
for some $t>0$. 

We next analyze $\Mean {e^{-tx_{i} M_{j \to i}^{(4)}}\Bigg |\mathcal{E}}$, depending on whether $j\in \partial {x}_{i}^{(4)}$ is a perfect query or not. Remind the definition of the perfect query in Def.~\ref{def:perfectq}. 
When we define $\partial \hat x_i^{(4)}$ as the set of perfect queries in $\partial x_i^{(4)}$ and let $\hat{a}:=|\partial\hat x_i^{(4)}|$, 
good event (\romannumeral 5)  asserts that the number of perfect queries in $\partial x_i^{(4)}$ is $\hat{a}\geq a-C_4 \log\log m-3$.
If $j\in \partial {x}_{i}^{(4)}$ with $w(j)=k$ and $d(j)=d$ is a perfect query, the probability that the message $m_{j\to i}^{(4)}$ is different from the true label $x_i$, denoted by $q_j^{(4)}:=\Pr{m_{j \to i}^{(4)} \neq x_i}$, is in the range $\left[\epsilon_{k,d}, \epsilon_{k,d} + D p^{(2)} \right]$ for $p^{(2)}=o(1/\log m)$ (this can be shown similar to~\eqref{step3_bound}). 
Thus, by taking $t=\frac{1}{2}$, we have
\beq
\begin{split}
\Mean {e^{-tx_{i} M_{j \to i}^{(4)}}\Bigg |\mathcal{E}}&= q_j^{(4)} \left(\frac{1 - \hat \epsilon_{k,d,i}}{\hat \epsilon_{k,d,i}} \right)^{1/2} + (1 - q_j^{(4)})\left(\frac{\hat \epsilon_{k,d,i}}{1 - \hat \epsilon_{k,d,i}} \right)^{1/2} \\
\label{mgf_bound}
&\leq 2 \sqrt{\epsilon_{k,d} (1 - \epsilon_{k,d})} + o(1 / (\log \log m)^{1/2}),
\end{split}
\eeq
where the last inequality is from the condition $|\hat{\epsilon}_{k,d,i}-\epsilon_{k,d}|=O(1/(\log\log m)^{1/4})$.
Thus, we can find $A_{k,d}\in (0,1)$ such that $\Mean {e^{-tx_{i} M_{j \to i}^{(4)}}\Bigg |\mathcal{E}}\leq A_{k,d}=2 \sqrt{\epsilon_{k,d} (1 - \epsilon_{k,d})} + o(1 / (\log \log m)^{1/2})$.
If $j$ is not a perfect query, i.e., $j\in  \partial {x}_{i}^{(4)}\backslash \partial \hat{x}_{i}^{(4)}$, on the other hand, we just bound $\Mean {e^{-tx_{i} M_{j \to i}^{(4)}}\Bigg |\mathcal{E}}$ by some large constant $C_{13}>0$. (Since $\hat\epsilon_{k,d,i}\in[\lambda,0.5]$ for some $\lambda>0$, we can always find such a constant $C_{13}>0$.)
When we define $s_{k,d}:=|\{j\in \partial \hat x_i^{(4)}:  w(j)=k , d(j)=d\}|$ as the number of perfect queries having degree $d$ and assigned to worker $k$, we have
\beq
\begin{split}\label{eqn:expskd}
&\prod_{(k,d)\in \partial_w x_i^{(4)}} \prod_{\{j \in \partial {x}_{i}^{(4)}:w(j)=k,d(j)=d\}} \Mean {e^{-tx_{i} M_{j \to i}^{(4)}}\Bigg |\mathcal{E}} \leq C_{13}^{C_4 \log \log m + 3} \prod_{(k,d)\in \partial_w x_i^{(4)}} A_{k,d}^{s_{k,d}},
\end{split}
\eeq
since the number of non-perfect queries in $\partial x_i^{(4)}$ is bounded above by $a-\hat{a}\leq C_{11} \log\log m + 3$.

We then analyze the probability that the number $s_{k,d}$ of degree-$d$ queries in $\partial \hat x_i^{(4)}$ assigned to a worker $k\in[w]$ is $a_{k,d}$ for each $(k,d)\in \partial_w x_i^{(4)}$. We can determine whether a query is a perfect query or not after step 4), before we get the actual degree of queries at step 5). Also, perfect queries cannot be removed from $G_i$, since they do not belong to semi-bad queries or semi-bad pairs. Thus, the degree distribution of perfect queries is $\{ \frac{d \Phi_d}{\bar d}\}$, and their worker distribution is also uniform over all workers. Conditioned on that $|\partial \hat x_i^{(4)}|=\hat{a}$, the probability that $(s_{1,1},\dots, s_{w,D})=(a_{1,1},\dots, a_{w,D})$ for $\sum_{k,d} a_{k,d}=\hat{a}$ is thus
\beq
\begin{split}\label{eqn:skdprob}
&\Pr{(s_{1,1},\dots, s_{w,D})=(a_{1,1},\dots, a_{w,D})\Big| |\partial \hat x_i^{(4)}|=\hat{a}}\\
&=\frac{(\hat a)!}{(a_{1,1})! \cdots (a_{w,D})!} \prod_{k,d=(1,1)}^{(w,D)} \left( \frac{d \Phi_d}{w \bar d} \right)^{a_{k,d}}.
\end{split}
\eeq
Therefore, from~\eqref{eqn:xi4analysis_newMarkov}, \eqref{eqn:expskd} and \eqref{eqn:skdprob}, we have
\beq
\begin{split}
&\Pr{\hat{x}_{i}^{(4)} \neq x_i\Big| |\partial x_i^{(4)}|=a}\\
&\leq C_{13}^{C_{11} \log \log m + 3 } \sum_{\{(a_{11},\dots, a_{Dw}): \sum_{k,d=(1,1)}^{(w,D)}a_{k,d} =\hat{a}\}} \frac{(\hat a)!}{(a_{1,d})! \cdots (a_{w,d})!} \prod_{k,d} \left( \frac{d \Phi_d}{w \bar d} \right)^{a_{k,d}} \prod_{k,d} A_{k,d}^{a_{k,d}} +o(1/m)\\
&= C_{13}^{C_{11} \log \log m + 3} \left(\sum_{k,d} \frac{d \Phi_d }{w \bar d} A_{k,d} \right)^{\hat a}+o(1/m)\\
&\leq C_{13}^{C_{11} \log \log m + 3}\left(\sum_{k,d} \frac{d \Phi_d }{w \bar d} A_{k,d} \right)^{a-C_{11}\log\log m - 3}+o(1/m),
\end{split}
\eeq
where the last inequality is from the good event (\romannumeral 5) that the number of perfect queries is $\hat{a}\geq a-C_{11} \log\log m - 3$.
Since we consider $a=\Theta(\log m)\in S_1$, we finally have
\beq
\Pr{\hat{x}_{i}^{(4)} \neq x_i\Big| |\partial x_i^{(4)}|=a}\leq \left(\sum_{k,d} \frac{d \Phi_d }{w \bar d} A_{k,d} \right)^{a(1-o(1))}+o(1/m).
\eeq

By plugging this bound into~\eqref{eqn:err_conda1}, we can bound $\Pr{\hat{x}_{i}^{(4)} \neq x_i}$ as $o(1/m)$ as follows,
\beq\label{eqn:xi4_err_final}
\begin{split}
\Pr{\hat{x}_{i}^{(4)} \neq x_i}&\leq \sum_{a\in S_1} {n^{(4)} \choose a} \left( \frac{\bar d}{m} \right)^{a} \left(1 - \frac{\bar d}{m} \right)^{n^{(4)} - a} \left(\sum_{k,d} \frac{d \Phi_d }{w \bar d} A_{k,d} \right)^{a(1-o(1))}+o(1/m)\\
&\leq \left(1 - \frac{\bar d}{m} + \frac{\bar d}{m} \left(\sum_{k,d} \frac{d \Phi_d }{w \bar d} A_{k,d} \right)^{(1 - o(1))} \right)^{n^{(4)}}+o(1/m) \\
&= \left(1 - \frac{\bar d}{m} \sum_{k,d} \frac{d \Phi_d }{w \bar d} (1- A_{k,d}) +  o({\bar d}/{m}) \right)^{n^{(4)}}+o(1/m) \\
&= \left(1 - \sum_{k,d} \frac{d \Phi_d }{mw} \left(\sqrt{\epsilon_{k,d}} - \sqrt{1 - \epsilon_{k,d}}\right)^2 +  o({\bar d}/{m}) \right)^{n^{(4)}}+o(1/m) \\
&\leq \Exp{-n^{(4)} \left(\sum_{k,d} \frac{d \Phi_d}{m w} \left( \sqrt{\epsilon_{k,d}} - \sqrt{1 - \epsilon_{k,d}} \right)^2 + o({\bar d}/{m})\right)} +o(1/m)\\
&\leq \Exp{ -\left( 1 + \frac{\eta}{2} \right) \log m} = o(1/m)
\end{split}
\eeq
where we used $n^{(4)}= (1 + \eta) \frac{ m\log m}{\sum_{d = 1}^D \sum_{k = 1}^w \frac{d\Phi_d}{w} ( \sqrt{1 - \epsilon_{k,d}}-\sqrt{\epsilon_{k,d}})^2}$ at the last inequality. 

%For brevity in writing, we do not explicitly state the conditioning in  $ \Pr{\hat{x}_{i}^{(4)} \neq x_i\Big| |\partial x_i^{(4)}|=a}$. 

%\begin{lemma}\label{lem:4}
%After Step 4 of Algorithm~\ref{algorithm2}, the probability that each label $x_i$ is different from its estimate $\hat{x}_i^{(4)}$ in~\eqref{eqn:x4i} is bounded above by
%\beq
%p^{(4)} := \Pr{\hat{x}_i^{(4)} \neq x_i} \leq \exp \left( - \left(1 + {\eta}/{2}\right) \log m \right)=o(1/m),
%\eeq
%for a sufficiently small $\eta>0$ when $n^{(4)}\geq ( 1 + \eta) \frac{ m\log m}{\sum_{d = 1}^D \sum_{k = 1}^w \frac{d\Phi_d}{w} ( \sqrt{1 - \epsilon_{k,d}}-\sqrt{\epsilon_{k,d}})^2}.
%$ Therefore, the strong recovery of labels $\{x_i\}$ is guaranteed with $\{\hat{x}_i^{(4)}\}$, i.e., $\Pr{{\hat{\vec{x}}^{(4)}} \neq \vec{x}}\to 0$ as $m\to\infty$. 
%\end{lemma}

%We provide the proofs for each lemma in the following subsections.
%\medskip
\subsubsection{Proof of Lemma~\ref{lem:2}}\label{subsec:pflem2}
In this lemma, we prove the weak recovery of good labels $i'\in \partial^2 \tilde{x}_i^{(4)}\bigcup  \partial^2 x_i^{(3)}$ when $\hat{x}_{i'}^{(2)}$ receives messages $\{m^{(2)}_{j\to i'}: j\in \partial x_{i'}\cap A^{(2)}\}$ from $|\partial x_{i'}\cap A^{(2)}|=\Theta\left(\frac{\log m}{\log\log m}\right)$-number of queries. Let us first analyze the probability that the message $m_{j\to i'}^{(2)}\in\{1,-1\}$ is different from the true label $x_{i'}\in\{1,-1\}$.  Since $m_{j\to i'}^{(2)} = y_j \prod_{i^{\prime\prime} \in \partial y_j \backslash \{i'\}} \hat{x}_{i^{\prime\prime}}^{(1)}$, $m_{j\to i'}^{(2)}$ is different from $x_{i'}$ when the received answer $y_{j}$ is incorrect and there are even number of wrong estimates in $\{\hat{x}_{i^{\prime \prime}}^{(1)}\}$ for $i^{\prime \prime} \in \partial y_{j} \setminus \{i^\prime\},$ or when $y_{j}$ is correct and there are odd number of wrong estimates in $\{\hat{x}_{i^{\prime\prime}}^{(1)}\}$. Thus, for a query $y_j$ with $w(j) = k$ and $d(j) = d$, we have
\beq
\begin{split}\label{m2bound}
\Pr{m_{j\to i^\prime}^{(2)} \neq x_{i^\prime}}
&\leq \epsilon_{k,d}+(1-2\epsilon_{k,d})\Bigg[\sum_{\substack{l\in[0:d-1],\\l\text{ odd}}}{d-1\choose l} (p^{(1)})^{l}(1-p^{(1)})^{d-1-l}\Bigg] \\
&= \frac{1 - (1 - 2\epsilon_{k,d}) (1 - 2 p^{(1)})^{d-1}}{2} < \frac{1}{2},
\end{split}
\eeq
and we can find $q^{(2)} < \frac{1}{2}$ such that $\Pr{m_{j\to i^\prime}^{(2)} \neq x_{i^\prime}} < q^{(2)}$ for all $j \in \partial x_{i^\prime} \cap A^{(2)}$. 

From the Chernoff bound, we have
\begin{equation}\label{eqn:chrnf_x2}
\Pr{\hat{x}_{i^\prime}^{(2)} \neq x_{i^\prime}}
= \Pr{\sum_{j \in \partial x_{i^\prime} \cap A^{(2)}} -x_{i^\prime} m_{j \to i^\prime}^{(2)} \geq 0}
\leq \prod_{j \in \partial x_{i^\prime} \cap A^{(2)}} \Mean {e^{-tx_{i^\prime} m_{j \to i^\prime}^{(2)}}},
\end{equation}
for any $t > 0$, assuming the good event (\romannumeral 1) such that the messages in $\{m_{j\to i'}^{(2)}: j \in \partial x_{i^\prime} \cap A^{(2)} \}$ are independent. If we let $q_j^{(2)}:= \Pr{m_{j \to i^\prime}^{(2)} \neq x_{i^\prime}} < q^{(2)}<1/2$, we get
\begin{equation}
\Mean{e^{-tx_{i^\prime} m_{j \to i^\prime}^{(2)}}}
= q_j^{(2)} e^t + (1 - q_j^{(2)}) e^{-t}
< q^{(2)} e^t + (1 - q^{(2)}) e^{-t}.
\end{equation}
By taking  $t = \frac{1}{2} \ln \left(\frac{1 - q^{(2)}}{q^{(2)}}\right)> 0$, we have
\beq\label{eqn:bd_oneterm}
\Mean{e^{-tx_{i^\prime} m_{j \to i^\prime}^{(2)}}}
<2\sqrt{q^{(2)} (1-q^{(2)})}<1.
\eeq

From~\eqref{eqn:chrnf_x2} and~\eqref{eqn:bd_oneterm}, for a good labdel $x_{i'}$ having $\abs{\partial x_{i^\prime} \cap A^{(2)}} > c_4 \left(\frac{\log m}{\log \log m} \right)$ for some constant $c_4>0$, we can bound the estimation error of $\hat{x}_{i'}^{(2)}$ as
\begin{equation}
\Pr{\hat{x}_{i^\prime}^{(2)}\neq x_{i^\prime}}
< \Exp{\left(c_4 \left(\frac{\log m}{\log \log m} \right)  \right)\log \left(2\sqrt{q^{(2)} (1-q^{(2)})} \right) } = o(1 / \log m).
\end{equation}

\subsubsection{Proof of Lemma~\ref{lem:3}}

Next, we prove that the estimates on worker reliabilities $\{\epsilon_{k,d,i}\}$ for the $k$-th worker for a degree-$d$ query for any $(k,d) \in \partial_w  x_i^{(4)}$ are accurate as
\begin{equation}
\label{step3_accuracy}
\abs{\hat{\epsilon}_{k,d,i} - \epsilon_{k,d}} = O(1 / (\log \log m)^{1/4})
\end{equation}
with probability $1 - o(1/m)$. 

%To prove this result, we use the fact that for any $(k,d)\in \partial_w  x_i^{(4)}$,  the number of Step 3 degree-$d$ queries assigned to worker $k$ (denoted by $ \partial w_{k,d}^{(3)}$) is $\Theta(\log m (\log\log m))$ with probability $1-o(1/m)$ when the number of Step 3 queries is $n^{(3)}=w(\log m)(\log \log m)$ and $w$ is the total number of workers. 
%\begin{lemma}
%For any $(k, d) \in \partial_w x_i ^{(4)}$, $ c_6 (\log m) (\log \log m) < \abs{\partial w_{k,d}^{(3)}} < C_6 (\log m) (\log \log m)$.
%\end{lemma}
%\begin{IEEEproof}
%For given $(k, d) \in \partial_w x_i^{(4)}$, $\abs{\partial w_{k,d}^{(3)}}$ follows $\Bin{n^{(3)}, \frac{\Phi_d}{w}}$, and the expectation of it is $\Phi_d (\log m)(\log \log m)$. Hence, by Lemma \ref{ber} (Chernoff bound for Binomial distribution), there exists $c_6$ and $C_6$ such that
%\begin{equation*}
%\Pr{c_6 (\log m)(\log \log m) < \abs{\partial w_{k,d}^{(3)}} < C_6 (\log m)(\log \log m)} \geq 1 - \Theta\left( e^{-(\log m)(\log \log m)} \right).
%\end{equation*}
%The union bound gives
%\begin{equation*}
%\Pr{c_6 (\log m)(\log \log m) < \abs{\partial w_{k,d}^{(3)}},\;\;\forall(k, d) \in \partial_w x_i^{(4)}} \geq 1 - \Theta\left( (\log m) e^{-(\log m)(\log \log m)} \right) = 1 - o(1/m).
%\end{equation*}
%\end{IEEEproof}
For a fixed $(k,d)\in \partial_w x_i^{(4)}$, we first show that $\epsilon_{k,d}$ is very close to the expectation of $\hat \epsilon_{k,d,i}$. 
Consider a query $y_j$ for some $j \in \partial  w_{k,d}^{(3)}$. 
Note that $E_{j}^{(3)} = \mathds{1} \bigg( y_j \neq \prod_{i' \in \partial y_j} \hat{x}_{i'}^{(2)} \bigg)$ equals 1 when the received answer $y_{j}$ is incorrect and there are even number of wrong estimates in $\{\hat{x}_{i^{\prime}}^{(2)}: i^{\prime} \in \partial y_{j}\}$, or when $y_{j}$ is correct and there are odd number of wrong estimates in $\{\hat{x}_{i^{\prime}}^{(2)}: i^{\prime} \in \partial y_{j}\}$. If $y_j$ is a good query such that all the connected labels $\{i'\in \partial y_j\}$ are good labels, by Lemma~\ref{lem:2} we have $\Pr{\hat{x}_{i'}^{(2)} \neq x_{i'}} \leq p^{(2)}=o(1/\log m)$ for all $i'\in \partial y_j$, and thus
\beq
\begin{split}
\label{step3_bound}
\epsilon_{k,d} \leq \Pr{E_{j}^{(3)} = 1}
&\leq \epsilon_{k,d}+(1-2\epsilon_{k,d})\Bigg[\sum_{\substack{l\in[0:d],\\l\text{ odd}}}{d\choose l} (p^{(2)})^{l}(1-p^{(2)})^{d-l}\Bigg] \\
&= \epsilon_{k,d} (1 - 2p^{(2)})^{d} + \frac{1 - (1 - 2p^{(2)})^{d}}{2} \\
&\leq \epsilon_{k,d} + d p^{(2)} \\
&\leq \epsilon_{k,d} + D p^{(2)}.
\end{split}
\eeq
If $y_j$ is not a good query for some $j \in \partial w_{k,d}^{(3)}$, we can use the trivial bound such that $\abs{\Pr{E_{j}^{(3)} = 1} - \epsilon_{k,d}} < 1$.
Conditioned on the good event (\romannumeral 6), we have $|\partial w_{k,d}^{(3)}|=\Theta((\log m)(\log\log m))$  and there are at least $(|\partial w_{k,d}^{(3)}|-C_{12}\log\log m-6)$-number of good queries in $\partial w_{k,d}^{(3)}$. Since $\E[\hat{\epsilon}_{k,d,i}] = \frac{\sum_{j \in \partial w_{k,d}^{(3)}} \E[E_{j }^{(3)} ]}{\abs{ \partial w_{k,d}^{(3)}}}=\frac{\sum_{j \in \partial w_{k,d}^{(3)}} \Pr{E_{j }^{(3)} =1 }}{\abs{ \partial w_{k,d}^{(3)}}}$,
we have
\beq
\abs{\epsilon_{k,d} - \mean \hat \epsilon_{k,d,i}} \leq \frac{\abs{\partial w_{k,d}^{(3)}} D p^{(2)} + C_{12} \log \log m + 6}{\abs{\partial w_{k,d}^{(3)}}} = o \left( \frac{1}{\log m}\right).
\eeq
We next show that $\hat \epsilon_{k,d,i}$ is close to its mean. Since $\{E_{j}^{(3)}\}_{j \in \partial w_{k,d}^{(3)}}$ are independent by good event (\romannumeral 3), using Hoeffding's inequality, we have
\beq
\Pr{\abs{\hat{\epsilon}_{k,d,i} - \mean \hat{\epsilon}_{k,d,i}} > \delta}
\leq e^{-2 \abs{\partial w_{k,d}^{(3)}}\delta^2}
\lesssim e^{-2 (\log m )(\log \log m)\delta^2},
\eeq
for any $\delta > 0$. If we take $\delta = 1 / (\log \log m)^{1/4}$ and use the union bound, we have
\beq
\Pr{\abs{\hat{\epsilon}_{k,d,i} - \mean \hat{\epsilon}_{k,d,i}} > \delta\text{ for some }(k,d) \in \partial_w  x_i ^{(4)}}
\lesssim \abs{\partial_w  x_i ^{(4)}} e^{-2 (\log m) (\log \log m)^{1/2}} = o(1/m)
\eeq
since $|\partial_w x_i^{(4)}|\leq |\partial x_i^{(4)}|$ and $ |\partial x_i^{(4)}|=\Theta(\log m)$ by the good event (\romannumeral 5).
Then,  from the triangle inequality we know that
\beq
\abs{\hat \epsilon_{k,d,i} - \epsilon_{k,d}}
\leq \abs{\hat \epsilon_{k,d,i} - \mean \hat{\epsilon}_{k,d,i}} + \abs{ \epsilon_{k,d} - \mean \hat{\epsilon}_{k,d,i}}
= \delta + o(1/ \log m)
= O(1 / (\log \log m)^{1/4})
\eeq
holds for any $(k, d) \in \partial_w x_i ^{(4)}$ with probability at least $1 - o(1/m)$. %Because of its high probability, we consider this good recovery as sort of a good event. Thus, we assume that we are given $\{\hat{\epsilon}_{k,d}\}_{(k,d) \in \partial_w \hat x_i^{(4)}}$ satisfying \eqref{step3_accuracy}, and $\hat \epsilon_{k,d}$'s are not random variables anymore in the subsequent analysis.

\section{Conclusions}\label{sec:con}

We considered binary classification of $m$ labels with XOR queries, where the query degree $d$ can be varying over queries and the error probability of the answer can change depending on the query degree as well as on the worker. We characterized the optimal number of queries required to reliably recover all the $m$ labels with high probability, and proposed an efficient inference algorithm that achieves this limit even without the knowledge of noise parameters. 
Simulations on synthetic data and real data show the effectiveness of the XOR queries and the proposed algorithm.

The problem considered here is an example of a more general planted constraint satisfaction problem (CSP), which has wide applications in clustering, community detection, and matrix/tensor completion problems. In these problems, intensive research is going on to bridge the gap between  information-theoretic limit and  computational limit, where the information-theoretic limit is determined by the required number of measurements to make the planted solution a unique solution with high probability while the computational limit is determined by the required number of measurements that allow a feasible algorithm in recovering the unique solution. 
In this work, we considered an example of the CSP in the context of binary classification with XOR queries and provided an inference algorithm achieving the exact information-theoretic limit even without the knowledge of noise parameters for measurements. 
One of the interesting future directions related to this work is to apply the algorithmic ideas from this work to possibly bridge the gap between the information-theoretic limit and the computational limit for other applications related to the planted CSP such as graph clustering or community detection.

%
% paper title
% can use linebreaks \\ within to get better formatting as desired
% Do not put math or special symbols in the title.

%\vskip0.25in

% As a general rule, do not put math, special symbols or citations
% in the abstract or keywords.

%\begin{IEEEkeywords}
%Sample complexity, query difficulty, Fountain codes, Soliton distribution, crowdsourcing.%, information theoretic surrogates.
%\end{IEEEkeywords}

\IEEEpeerreviewmaketitle

\appendices

\section{Proof of Theorem~\ref{thm:thm1}}\label{app:prThm1}

In the proof of Theorem~\ref{thm:thm1}, for simplicity in notation, we assume that the true label vector is $\vec{x}\in\{0,1\}^m$ instead of $\vec{x}\in\{-1,1\}^m$, i.e., we consider $2\vec{x}-1$ as the true label vector $\vec{x}$, and find the estimate $\hat{\vec{x}}\in\{0,1\}^m$. 
 In Theorem~\ref{thm:thm1}, we assume that total $n$ XOR queries are randomly and independently generated among which the fraction of degree-$d$ queries is $\Phi_d$ for $\sum_{d=1}^D \Phi_d=1$, and each query is randomly assigned to a worker $k\in[w]$ who provides an incorrect answer to a degree-$d$ query with probability $\epsilon_{k,d}<1/2$. 
When $\hat{\vec{x}}\in\{0,1\}^m$  is the optimal estimate (maximum likelihood estimate) of the label vector $\vec{x}\in\{0,1\}^m$  that minimizes the probability of error $\Pr{\hat{\vec{x}} \neq \vec{x}}$ using a known $\{\epsilon_{k,d}\}$, we assert that the strong recovery is possible, i.e., $\Pr{\hat{\vec{x}} \neq \vec{x}} \to 0$ as $m\to\infty$, if the number of queries  is
\begin{equation}
n \geq ( 1 + \eta) \frac{ m\log m}{\sum_{d = 1}^D \sum_{k = 1}^w \frac{d\Phi_d}{w} ( \sqrt{1 - \epsilon_{k,d}}-\sqrt{\epsilon_{k,d}})^2},
\end{equation}
and only if
\begin{equation}
n \geq ( 1 - \eta) \frac{ m\log m}{\sum_{d = 1}^D \sum_{k = 1}^w \frac{d\Phi_d}{w} ( \sqrt{1 - \epsilon_{k,d}}-\sqrt{\epsilon_{k,d}})^2}
\end{equation}
for any arbitrarily small constant $\eta > 0$. %The converse holds when the number of workers $w=o(\log m)$.

%In the proof below, for simplicity in notation, we assume that the true label vector is $\vec{x}\in\{0,1\}^m$ instead of $\vec{x}\in\{-1,1\}^m$, i.e., we consider $2\vec{x}-1$ as the true label vector $\vec{x}$, and find estimate $\hat{\vec{x}}\in\{0,1\}^m$. 

\subsection{Proof of Achievability}
Theorem~\ref{thm:thm1} is an extension of Theorem 2 in \cite{ahn2019community}, where a similar setup was analyzed for the case that the query degree $d$ is fixed over all queries and the noise parameter is fixed as $\epsilon_{k,d}=\epsilon$, $\forall k\in[w]$ and $\forall d\in[D]$. 
The main difference in our analysis compared to that in \cite{ahn2019community} occurs due to the fact that the ML decoding rule, which generates the optimal estimate $\hat{\vec{x}} $ that minimizes the error probability $\Pr{\hat{\vec{x}}\neq \vec{x}}$, should use weighted majority voting instead of majority voting when we aggregate answers from different workers with noise parameters $\{\epsilon_{k,d}\}$ that depend both on the worker reliability and query degree.

Denote by $\vec{0}$ the $m$-dimensional all-zero label vector. We assume that the ground truth label vector is $\vec{0}$ without loss of generality.
The ML decoding rule results in an error if there exists $\vec{v} \neq \vec{0}$ such that $\Pr{\vec{x} = \vec{v}|\vec{y} } \geq \Pr{ \vec{x} = \vec{0}|\vec{y} }$. Denote by $\hat{\vec{x}}(\vec{y})$ the estimated label vector of the ML decoding rule. For brevity, we just use $\hat{\vec{x}}=\hat{\vec{x}}(\vec{y})$. 
%, and assume that the ground truth label vector is $\vec{0}$ without loss of generality. When only even degree queries are used, we cannot distinguish $\vec{0}$ with $\vec{1}$. Therefore, we assume that odd degree queries account for at least a constant fraction of the denominator of~\eqref{eqn:thm1}, i.e., there exists a constant $c_1$ such that

By using union bound, the error probability is bounded by
\beq
\Pr{ \hat{\vec{x}}\neq \vec{0}}\leq \sum_{\vec{v}\neq \vec{0}}\Pr{\hat{\vec{x}}=\vec{v}}=\sum_{s=1}^m {m\choose s}\Pr{\hat{\vec{x}}=\vec{v}|\|\vec{v}\|_1=s },
\eeq
where the last equality is due to the symmetry in the way we design queries. When $\vec{v}_s$ denote the length-$m$ vector whose first $s$ components are $1$ and the rest are $0$, it can be shown that $\Pr{\hat{\vec{x}}=\vec{v}|\|\vec{v}\|_1=s }=\Pr{\hat{\vec{x}}=\vec{v}_s}$ for all $\vec{v}$ with $\|\vec{v}\|_1=s $. 
Thus, the bound on the error probability can be written as
\beq\label{eqn:overallunion}
\Pr{ \hat{\vec{x}}\neq \vec{0}}\leq \sum_{s=1}^m {m\choose s}\Pr{\hat{\vec{x}}=\vec{v}_s}.
\eeq

Let $n_{k,d}$ denote the number of degree-$d$ queries assigned to the worker $k$, where the total number of queries is $n=\sum_{k,d} n_{k,d}$. 
We expand $\Pr{\hat{\vec{x}}=\vec{v}_s}$ conditioned on $\{n_{k,d}\}$.
\beq
\begin{split}\label{eqn:hatxvsallnkd}
\Pr{\hat{\vec{x}}=\vec{v}_s}&= \sum_{\{n_{k,d}\}} {n \choose n_{1,1} \cdots n_{w,D}}\left( \prod_{k,d} \left(\frac{\Phi_d}{w}\right)^{n_{k,d}}\right) \Pr{\hat{\vec{x}}=\vec{v}_s | \{ n_{k,d} \}}.
\end{split}
\eeq

We next analyze $\Pr{\hat{\vec{x}}=\vec{v}_s | \{ n_{k,d} \}}$. Let $l_{k,d}$ be the number of queries among $n_{k,d}$ of which the correct answers are different for $\vec{v_s}$ and $\vec{0}$. Since $\vec{v_s}$ and $\vec{0}$ are different only at the first $s$ components and each query randomly chooses $d$ components among $m$ and asks the XOR of the chosen $d$ components, the probability that a degree-$d$ query has a different answer for $\vec{v_s}$ and $\vec{0}$ is
\beq
p_{s,d}:=\frac{\sum_{\substack{1 \leq i \leq d\\ i\text{: odd}}} {s \choose i}{m - s \choose d - i}}{{m \choose d}}.
\eeq
Therefore, $l_{k,d}$ follows a binomial distribution, $B$($n_{k,d}, p_{s,d}$) for all $k\in[w]$. 
By using this, it can be shown that
\beq\label{eqn:hatxvsnkd}
\Pr{\hat{\vec{x}}=\vec{v}_s | \{ n_{k,d} \}}=\sum_{\{l_{k,d}\}} \left( \left(\prod_{k,d} {n_{k,d} \choose l_{k,d}} p_{s,d}^{l_{k,d}} (1 - p_{s,d})^{n_{k,d} - l_{k,d}}\right) \Pr{\hat{\vec{x}}=\vec{v}_s | \{ l_{k,d} \}} \right).
\eeq

Let us analyze $ \Pr{\hat{\vec{x}}=\vec{v}_s | \{ l_{k,d} \}}$. Since we assume that $\vec{0}$ is the ground truth vector, the correct answers for all XOR queries should be equal to 0. Let $r_{k,d}$ denote the number of queries among $l_{k,d}$ such that the received answer is equal to 1 .
Remind that the probability of receiving incorrect answer for degree-$d$ query assigned to worker $k$ is equal to $\epsilon_{k,d}$.
Given the answer vector $\vec{y}$, the ML decoder claims that $\hat{\vec{x}}=\vec{v}_s$ if
\begin{equation}
\prod_{k,d} (1 - \epsilon_{k,d})^{r_{k,d}} \epsilon_{k,d}^{l_{k,d} - r_{k,d}} \geq \prod_{k,d} (1 - \epsilon_{k,d})^{l_{k,d} - r_{k,d}} \epsilon_{k,d}^{r_{k,d}}.
\end{equation}
Applying $\log$ to both sides and rearranging terms, the above inequality can be written as
\begin{equation}
\label{eq4}
\sum_{k,d}\log\left(\frac{1 - \epsilon_{k,d}}{\epsilon_{k,d}}\right) r_{k,d} \geq \frac{1}{2} \sum_{k,d} \log\left(\frac{1 - \epsilon_{k,d}}{\epsilon_{k,d}}\right) l_{k,d},
\end{equation}
which is basically the weighted majority voting. Note that for any $t>0$,
\beq
\begin{split}
 \Pr{\hat{\vec{x}}=\vec{v}_s | \{ l_{k,d} \}}&=\Pr{\sum_{k,d}\log\left(\frac{1 - \epsilon_{k,d}}{\epsilon_{k,d}}\right) r_{k,d} \geq \frac{1}{2} \sum_{k,d} \log\left(\frac{1 - \epsilon_{k,d}}{\epsilon_{k,d}}\right) l_{k,d}\middle| \{ l_{k,d}\} }\\
 &= \Pr{e^{t\sum_{k,d}\log\left(\frac{1 - \epsilon_{k,d}}{\epsilon_{k,d}}\right) r_{k,d}} \geq e^{ \frac{1}{2}t \sum_{k,d} \log\left(\frac{1 - \epsilon_{k,d}}{\epsilon_{k,d}}\right) l_{k,d}}\middle| \{ l_{k,d}\} }\\
& \leq \frac{\prod_{k,d} \mean\bigg[ \exp \bigg(t \log\left(\frac{1 - \epsilon_{k,d}}{\epsilon_{k,d}}\right) r_{k,d} \bigg) \bigg] }{\prod_{k,d} \exp \bigg(  \frac{1}{2}t  \log\left(\frac{1 - \epsilon_{k,d}}{\epsilon_{k,d}}\right) l_{k,d}\bigg)}\\
&= \prod_{k,d} \bigg( \epsilon_{k,d} \left(\frac{1 - \epsilon_{k,d}}{\epsilon_{k,d}} \right)^{\frac{1}{2}t} + (1 - \epsilon_{k,d}) \left(\frac{1 - \epsilon_{k,d}}{\epsilon_{k,d}} \right)^{-\frac{1}{2}t} \bigg)^{l_{k,d}}
\end{split}
\eeq
where the inequality is by the Chernoff bound and the last equality holds since $r_{k,d}$ follows a binomial distribution, $B(l_{k,d}, \epsilon_{k,d})$.
By choosing $t=1$,
\beq\label{eqn:xhatvsCher}
 \Pr{\hat{\vec{x}}=\vec{v}_s | \{ l_{k,d} \}}\leq \prod_{k,d} \left( 2 \sqrt{\epsilon_{k,d} ( 1- \epsilon_{k,d}}) \right)^{l_{k,d}}.
\eeq
By using~\eqref{eqn:hatxvsnkd} and ~\eqref{eqn:xhatvsCher},
\beq
\begin{split}\label{eqn:hatxvsnkd2}
\Pr{\hat{\vec{x}}=\vec{v}_s | \{ n_{k,d} \}}&\leq \sum_{\{l_{k,d}\}} \bigg( \prod_{k,d} {n_{k,d} \choose l_{k,d}} p_{s,d}^{l_{k,d}} (1 - p_{s,d})^{n_{k,d} - l_{k,d}} \left( 2 \sqrt{\epsilon_{k,d} ( 1- \epsilon_{k,d}}) \right)^{l_{k,d}} \bigg) \\
&= \prod_{k,d} \bigg( \sum_{\{l_{k,d}\}} {n_{k,d} \choose l_{k,d}} p_{s,d}^{l_{k,d}} (1 - p_{s,d})^{n_{k,d} - l_{k,d}} \left( 2 \sqrt{\epsilon_{k,d} ( 1- \epsilon_{k,d}}) \right)^{l_{k,d}} \bigg) \\
&= \prod_{k,d} (1 - p_{s,k,d}^\prime )^{n_{k,d}},
\end{split}
\eeq
where $p_{s,k,d}^\prime = ( \sqrt{1 - \epsilon_{k,d}}-\sqrt{\epsilon_{k,d}} )^2 p_{s,d}$. 
Thus, by using~\eqref{eqn:hatxvsallnkd} and~\eqref{eqn:hatxvsnkd2}, we get
\beq
\begin{split}\label{eqn:hatxvsallnkd2}
\Pr{\hat{\vec{x}}=\vec{v}_s}& \leq \sum_{\{n_{k,d}\}} {n \choose n_{1,1} \cdots n_{w,D}} \prod_{k,d} \left(\frac{\Phi_d}{w}\right)^{n_{k,d}} (1 - p_{s,k,d}^\prime )^{n_{k,d}} \\
&= \bigg(1 - \sum_{k,d} \frac{\Phi_d}{w} p_{s,k,d}^\prime \bigg)^n \\
&\leq \exp \bigg( -n\sum_{k,d} \frac{\Phi_d}{w} p_{s,k,d}^\prime \bigg).
\end{split}
\eeq
Lastly, by using~\eqref{eqn:overallunion} and~\eqref{eqn:hatxvsallnkd2}
\beq\label{eqn:lastbd}
\Pr{ \hat{\vec{x}}\neq \vec{0}}\leq \sum_{s=1}^m {m\choose s} \exp \bigg( -n\sum_{k,d} \frac{\Phi_d}{w} p_{s,k,d}^\prime \bigg)
\eeq
for $p_{s,k,d}^\prime = (\sqrt{\epsilon_{k,d}} - \sqrt{1 - \epsilon_{k,d}})^2 p_{s,d}$.

We next use similar techniques used in the proof of Theorem 2 in \cite{ahn2019community} to show that the right-hand side of~\eqref{eqn:lastbd} goes to 0 when
\beq\label{eqn:nachiev}
n\geq  ( 1 + \eta) \frac{ m\log m}{\sum_{d = 1}^D \sum_{k = 1}^w \frac{d\Phi_d}{w} ( \sqrt{1 - \epsilon_{k,d}}-\sqrt{\epsilon_{k,d}})^2}
\eeq
for a small universal constant $\eta > 0$.

We divide the sum in the right-hand side of~\eqref{eqn:lastbd} into three regimes: $s \leq \delta m$, $\delta m < s \leq m - \delta m$, and $m - \delta m < s$, where $0 < \delta < 1$ is a small constant chosen later. First, we consider the case where $s \leq \delta m$. Note that $p_{s,d}$ is bounded below as
\begin{equation}
p_{s,d} = \frac{\sum_{\substack{1 \leq i \leq d\\ i\text{: odd}}} {s \choose i}{m - s \choose d - i}}{{m \choose d}} \geq \frac{{s \choose 1}{ m - s \choose d - 1}}{{m \choose d}} \geq \frac{{s \choose 1}{(1 - \delta)m \choose d - 1}}{{m \choose d}} \geq c_2 s \frac{d}{m},
\end{equation}
where $c_2 = (1 - 2 \delta)^D$. Thus, the summation over $s=1$ to $\delta m$ is bounded above by
\beq
\begin{split}\label{eqn:sumcase1}
\sum_{s=1}^{\delta m} {m \choose s} \exp \bigg( -n\sum_{k,d} \frac{\Phi_d}{w} p_{s,k,d}^\prime \bigg)
&\leq \sum_{s=1}^{\delta m} {m \choose s} \exp \bigg( -\frac{c_2 sn}{m} \sum_{k,d} \frac{d\Phi_d}{w} (\sqrt{\epsilon_{k,d}} - \sqrt{1 - \epsilon_{k,d}})^2 \bigg) \\
&\leq \sum_{s=1}^{\delta m} m^s \exp \left( - \left(1 + \frac{\eta}{2} \right) s \log m \right) \\
&= \sum_{s=1}^{\delta m} \exp\left( - \frac{\eta}{2} s \log m \right),
\end{split}
\eeq
where the last term goes to 0 for a sufficiently small $\delta$.

For the second case, we also bound $p_{s,d}$ using the first term as
\begin{equation}
p_{s,d} = \frac{\sum_{\substack{1 \leq i \leq d\\ i\text{: odd}}} {s \choose i}{m - s \choose d - i}}{{m \choose d}} \geq \frac{{s \choose 1}{ m - s \choose d - 1}}{{m \choose d}} \geq \frac{\delta m { \delta m \choose d - 1}}{{m \choose d}} \geq c_3 d,
\end{equation}
where $c_3 = \left(\frac{\delta}{2}\right)^D$. The summation over $s=\delta m$ to $m-\delta m$ goes to $0$ since
\beq
\begin{split}\label{eqn:sumcase2}
&\sum_{s=\delta m}^{m - \delta m} {m \choose s} \exp \bigg( -n\sum_{k,d} \frac{\Phi_d}{w} p_{s,k,d}^\prime \bigg)\\
&\leq \sum_{s=\delta m}^{m - \delta m} {m \choose s} \exp \bigg( -c_3 n \sum_{k,d} \frac{d\Phi_d}{w} (\sqrt{\epsilon_{k,d}} - \sqrt{1 - \epsilon_{k,d}})^2 \bigg) \\
&\leq \exp(m - (1 + \eta) c_3 m \log m ) \to 0.
\end{split}
\eeq
For the last case, we only use odd $d$'s to bound the summation. Using the last term, we have
\begin{equation}
p_{s,d} = \frac{\sum_{\substack{1 \leq i \leq d\\ i\text{: odd}}} {s \choose i}{m - s \choose d - i}}{{m \choose d}} \geq \frac{{s \choose d}{ m - s \choose 0}}{{m \choose d}} \geq \frac{{(1 - \delta) m \choose d}}{{m \choose d}} \geq c_4 d,
\end{equation}
where $c_4 = \frac{(1 - 2 \delta)^D}{D}$. The summation over $s=m-\delta m$ to $m$ goes to $0$ since
\beq
\begin{split}\label{eqn:sumcase3}
\sum_{s=m - \delta m}^{m} {m \choose s} &\exp \bigg( -n\sum_{k,d} \frac{\Phi_d}{w} p_{s,k,d}^\prime \bigg)
\leq \sum_{s=m - \delta m}^{m} {m \choose s} \exp \bigg( -n\sum_{k,d\text{ odd}} \frac{\Phi_d}{w} p_{s,k,d}^\prime \bigg) \\
&\leq \sum_{s=m - \delta m}^{m} {m \choose s} \exp \bigg( -c_4 n\sum_{k,d\text{ odd}} \frac{d \Phi_d}{w} (\sqrt{\epsilon_{k,d}} - \sqrt{1 - \epsilon_{k,d}})^2 \bigg) \\
&\leq \sum_{s=m - \delta m}^{m} {m \choose s} \exp(-(1 + \eta) c_1 c_4 m \log m) \\
&\leq \exp(m - (1 + \eta) c_1 c_4 m \log m) \to 0,
\end{split}
\eeq
where the third inequality holds since if $\Phi_d > 0$ for at least one odd $d$, we can a constant $c_1 > 0$ such that
\begin{equation}
\label{eq2}
\sum_{k, d\text{ odd}} \frac{d \Phi_d}{w} (\sqrt{1 - \epsilon_{k,d}}-\sqrt{\epsilon_{k,d}} )^2 \geq c_1 \sum_{k, d} \frac{d \Phi_d}{w} (\sqrt{1 - \epsilon_{k,d}}-\sqrt{\epsilon_{k,d}} )^2 .
\end{equation}
By combing~\eqref{eqn:sumcase1}, ~\eqref{eqn:sumcase2}, and~\eqref{eqn:sumcase3}, it can be shown that the upper bound on $\Pr{ \hat{\vec{x}}\neq \vec{0}}$ in~\eqref{eqn:lastbd} goes to $0$ with the number $n$ of queries  satisfying~\eqref{eqn:nachiev}.

\subsection{Proof of Converse}

%The converse result is an extension of that in \cite{ahn2019community} to the multiple degree and nonuniform error case. However, the extension is not trivial due to the difference in query assignment model. Later in the proof we will clarify the difficulty in the analysis that comes from our model.

We note that the converse result is similar to that in \cite{ahn2019community}, except that our result holds for any combination of degree $d$ queries and noise parameters $\{\epsilon_{k,d}\}$ while that in \cite{ahn2019community} considers the case of a fixed query degree $d$ and a fixed error probability $\epsilon_{k,d}=\epsilon$. However, the extension is not trivial due to the difference in query assignment model; in our model, we fix the number $n$ of total queries and randomly choose $d$ labels for each degree-$d$ query, while that in \cite{ahn2019community} independently samples every $n\choose d$ queries with a fixed probability $p>0$. Later in the proof we will clarify where the difficulty in the analysis of our model comes.

Let $\mathcal{E}_\eta$ be the event that the ground truth label vector $\vec{0}$ is more probable than any other $\vec{v}\neq 0$, i.e., $\Pr{\vec{x} = \vec{0}|\vec{y} } > \Pr{ \vec{x} = \vec{v}|\vec{y} }$ for all $\vec{v}\neq 0$ so that the ML decoding rule provides the correct estimate $\hat{\vec{x}}=\vec{0}$, where $n$ is given by 
\begin{equation}\label{eqn:nboundconverse}
n = (1 - \eta) \frac{m \log m}{\sum_{k,d} \frac{d \Phi_d}{w} (\sqrt{1 - \epsilon_{k,d}}-\sqrt{\epsilon_{k,d}} )^2}
\end{equation}
for $0<\eta < 1$. Our goal is to prove that $\Pr{\mathcal{E}_\eta}<1$ for any $\eta\in(0,1)$. Since $\Pr{\mathcal{E}_\eta}$ increases as $\eta$ decreases, once we prove that $\Pr{\mathcal{E}_\eta}<1$ for some $\eta>0$ it implies that $\Pr{\mathcal{E}_\eta'}<1$ for any $\eta'>\eta$.
Note that for an arbitrarily small $\eta>0$, there always exist $\zeta>0$ such that
\begin{equation}\label{eqn:nboundlower}
(1 + \zeta) \frac{m \log m}{\sum_d d \Phi_d} < (1 - \eta) \frac{m \log m}{\sum_{k,d} \frac{d \Phi_d}{w} (\sqrt{1 - \epsilon_{k,d}}-\sqrt{\epsilon_{k,d}} )^2}
\end{equation}
since $(\sqrt{1-\epsilon_{k,d}}-\sqrt{\epsilon_{k,d}})^2<1$. We will use this fact to prove $\Pr{\mathcal{E}_\eta}<1$ for an arbitrarily small $\eta>0$.

We next introduce two ``good events," related to query design and assignment, that occur with high probability and on whose intersection it can be shown that $\Pr{\mathcal{E}_\eta}<1$ for an arbitrarily small $\eta>0$. The first good event is about the number of items in $[m]$ that are not simultaneously selected by any query in $[n]$.
By Lemma 2 of \cite{ahn2019community}, with high probability there exist $r = \frac{m}{2 \log ^7 m}$ components of $\vec{x}$ that are not simultaneously contained in any query when the number of queries $n=O(m\log m)$ and the maximum query degree $D=\Theta(1)$. Denote this event by $\Delta_1$ and let such $r$ components be the first $r$ components of $\vec{x}$, i.e., $(x_1,\dots,x_r)$, without loss of generality. Then, one can bound $\Pr{\mathcal{E}_\eta}$ as 
\begin{equation}
\label{eq7}
\Pr{\mathcal{E}_\eta}-o(1) =\Pr{\mathcal{E}_\eta | \Delta_1} %\leq \pr \bigg[ \bigcap_{s=1}^r A_s \bigg| \Delta \bigg] = \Pr{A_1 | \Delta}^r.
\end{equation}
Below, we always condition on $\Delta_1$ and write $\Pr{\mathcal{E}_\eta | \Delta_1}$ as $\Pr{\mathcal{E}_\eta}$ from brevity.

Let $\mathcal{E}_{\eta,i}$ be the event that $\vec{0}$ is more probable than $\vec{e}_i$, where $\vec{e}_i$ is the $m$-dimensional unit vector with its $i$-th component equal to 1. Then, we have the bound
\begin{equation}\label{eqn:errorintersect}
\Pr{\mathcal{E}_\eta} \leq \pr \bigg[ \bigcap_{i=1}^r \mathcal{E}_{\eta,i} \bigg].
\end{equation}

In the query assignment model in \cite{ahn2019community}, the events $\{\mathcal{E}_{\eta, i} \}_{i \in [r]}$ are mutually independent conditioned on $\Delta_1$, since the subsets of queries in $[n]$ that determine $\mathcal{E}_{\eta, i}$ for each $i\in [r]$ do not overlap and the sizes of the subsets are independently determined. Therefore, in \cite{ahn2019community}, it became $\pr \bigg[ \bigcap_{i=1}^r \mathcal{E}_{\eta,i} \bigg]=\prod_{i=1}^r \Pr{\mathcal{E}_{\eta,i}}$, which made the analysis simple. 
 However, in our model, $\{\mathcal{E}_{\eta, i} \}_{i \in [r]}$ are not anymore independent, since the number of total queries is fixed to $n$ and  the sizes of the subsets of queries that have selected the item $i$ for $i\in[r]$ are dependent on each other.  Let $n_i$ denote the size of the subset of queries that have selected item $i$. 
In our model, $\{\mathcal{E}_{\eta, i} \}_{i \in [r]}$ become independent conditioned on $\{n_i\}_{i \in [r]}$.
 
 To this end, we decompose the query assignment process conditioned on $\Delta_1$ into three steps as follows.
\begin{enumerate}
\item Each query selects one of the item in $[r]$ with probability $\bar \Psi = \sum_d \Psi_d$, where
\beq\label{eqn:psid}
\Psi_d = \frac{\Phi_d {m - r \choose d - 1}}{r {m - r \choose d - 1} + {m - r \choose d}} = \frac{d \Phi_d }{m + (d - 1)(r - 1)},
\eeq
and selects no item in $[r]$ with probability $1 - r \Psi$.
\item The degree of queries that selected an item in $[r]$ is drawn from the distribution $\left\{ \frac{\Psi_d}{\bar \Psi} \right\}$. Note that it is not $\{ \Phi_d \}$.
\item Each query that selected an item in $[r]$ selects the remaining items from $[m] \setminus [r]$.
\end{enumerate}
Since a query cannot select more than one items in $[r]$ conditioned on $\Delta_1$, $\{n_i \}_{i \in [r]}$ is determined after the first step.

The second good event is related to the number of queries that select each item $i\in[r]$. 
\begin{lemma}
\label{good2_converse}
Let us define the event $\Delta_2$ as
\begin{center}
$\Delta_2$: The number of queries that select $x_i$ is $c_1 \log m < n_i < C_1 \log m$ for all $i \in [r]$ and for some $C_1>c_1>0$.
\end{center}
Conditioned on $\Delta_1$, when we have total number of queries $n=\Theta(m\log m)$,  we have $\Pr{\Delta_2} \geq 1 - o(1)$.
\end{lemma}
\begin{proof} For each $i\in[r]$, $n_i$ is distributed by $\text{Bin}(n,\bar{\Psi})$ where $n>(1+\zeta)\frac{m\log m}{\sum_d d\Phi_d}$ for some $\zeta>0$ by ~\eqref{eqn:nboundlower}. Note that $\bar{\Psi}=\sum_d \Psi_d$ for $\Psi_d$ defined in~\eqref{eqn:psid} converges to $(\sum_d d \Phi_d)/m$ as $m\to \infty$ since $r=\frac{m}{2\log^7 m}$. Thus, the good event $\Delta_2$ can be proved in the exactly same way as the proof of the good event $S_1$, which is defined in Lemma~\ref{good1_1} and proved in Appendix~\ref{app:pfLemmagood1} by using the Chernoff bound. 
\end{proof}

From $\Pr{\Delta_2} \geq 1 - o(1)$, we have
\beq
\begin{split}
\pr \bigg[ \bigcap_{i=1}^r \mathcal{E}_{\eta,i} \bigg]
&\leq \sum_{\{n_i\}_{i \in [r]} \in \Delta_2} \Pr{\{n_i\}_{i \in [r]}} \pr \bigg[ \bigcap_{i=1}^r \mathcal{E}_{\eta,i} \bigg| \{n_i\}_{i \in [r]} \bigg] + o(1) \\
\label{Enisum}
&= \sum_{\{n_i\}_{i \in [r]}} \Pr{\{n_i\}_{i \in [r]}} \prod_{i=1}^r \Pr{ \mathcal{E}_{\eta,i} \middle | n_i} + o(1).
\end{split}
\eeq

Next, we will derive a lower bound on $\Pr{ \mathcal{E}_{\eta,i}^c \middle | n_i}$ to get an upper bound on $\Pr{ \mathcal{E}_{\eta,i} \middle | n_i}$. Let $i \in [r]$ be given. We define $n_{i,k,d}$ to be the number of queries among $[n_i]$ having degree $d$ and assigned to worker $k$. We denote the answer for the $j$th query among $[n_{i,k,d}]$ by $Y_{k,d}^j$, and let $X_{k,d}^j = \log \left(\frac{1 - \epsilon_{k,d}}{\epsilon_{k,d}}\right) Y_{k,d}^j$. Then, we can explicitly write $\Pr{ \mathcal{E}_{\eta,i}^c \middle | \{ n_{i,k,d} \}}$ as
\begin{equation}\label{Eibound}
\Pr{ \mathcal{E}_{\eta,i}^c \middle | \{ n_{i,k,d} \}} = \Pr{\sum_{k,d} \sum_{j=1}^{ n_{i,k,d}} X_{k,d}^j \geq \frac{1}{2} \sum_{k,d} \sum_{j=1}^{ n_{i,k,d}} \log \left(\frac{1 - \epsilon_{k,d}}{\epsilon_{k,d}}\right)}
\end{equation}
\begin{lemma}\label{lem:boundEzetaic}
Conditioned on $\Delta_2$, we can bound $\Pr{ \mathcal{E}_{\eta,i}^c \middle | \{ n_{i,k,d} \}} $ in~\eqref{Eibound} as
\beq\label{eqn:boundEzetaic}
\Pr{ \mathcal{E}_{\eta,i}^c \middle | \{ n_{i,k,d} \}} \geq c_2 e^{- \sqrt {\log m}} \prod_{k,d} \left( 2 \sqrt{\epsilon_{k,d} ( 1 - \epsilon_{k,d})}\right)^{n_{i,k,d}}.
\eeq
\end{lemma}
\begin{proof} 
To get the lower bound on \eqref{Eibound}, we use the technique used in the proof of the Cramer-Chernoff bound~\cite{klenke2014probability}. Let us define a new random variable $Z_{k,d}^j$ that has the same support of $X_{k,d}^j$, but has a different probability distribution given by
\begin{equation}
\Pr{Z_{k,d}^j = z} = \frac{e^z \Pr{X_{k,d}^j = z}}{\Mean{e^{X_{k,d}^j}}}.
\end{equation}
In other words, $Z_{k,d}^j$ is a random variable such that
\begin{equation}
Z_{k,d}^j = \begin{cases}
\log \left(\frac{1 - \epsilon_{k,d}}{\epsilon_{k,d}}\right) &\text{w.p. } \frac{1}{2} \\
0 &\text{w.p. } \frac{1}{2}
\end{cases}.
\end{equation}
With $Z_{k,d}^j$, we can rewrite \eqref{Eibound} as
\begin{equation}\label{Eibound2}
\eqref{Eibound} = \Mean{\mathds{1}\left\{ \sum_{k,d,j} Z_{k,d}^j \geq  t_i \right\} \prod_{k,d,j} e^{-Z_{k,d}^j} \Mean{e^{X_{k,d}^j}}},
\end{equation}
where $t_i =\frac{1}{2} \sum_{k,d,j} \log \left(\frac{1 - \epsilon_{k,d}}{\epsilon_{k,d}}\right)$. The variance of each $Z_{k,d}^j$ is bounded and we have $\sum_{k,d} n_{i,k,d} = n_i = \Theta(\log m)$ terms in the summation $\sum_{k,d,j} Z_{k,d}^j$ by the second good event $\Delta_2$. Hence, by applying the Berry-Esseen theorem ~\cite{klenke2014probability} to the summation, we have
\begin{equation}
\Pr{t_i \leq \sum_{k,d,j} Z_{k,d}^j \leq t_i + \sqrt{\log m}} \geq c_2
\end{equation}
for some constant $0 < c_2 < 1$. With this result, we acquire a bound on \eqref{Eibound2} such that
\begin{align}
\eqref{Eibound2} &\geq \Mean{\mathds{1}\left\{ t_i \leq \sum_{k,d,j} Z_{k,d}^j \leq t_i + \sqrt{\log m} \right\} e^{-\sum_{k,d,j} Z_{k,d}^j} \prod_{k,d,j}  \Mean{e^{X_{k,d}^j}}} \\
&\geq \Pr{t_i \leq \sum_{k,d,j} Z_{k,d}^j \leq t_i + \sqrt{\log m}} e^{- \sqrt {\log m}} e^{-t_i} \prod_{k,d,j}  \Mean{e^{X_{k,d}^j}} \\
\label{Eibound3}
&\geq c_2 e^{- \sqrt {\log m}} e^{-t_i} \prod_{k,d,j} \Mean{e^{X_{k,d}^j}}
\end{align}
With the simple calculation $\Mean{e^{X_{k,d}^j}} = 2(1 - \epsilon_{k,d})$, \eqref{Eibound3} is equal to
\begin{equation}
c_2 e^{- \sqrt {\log m}} \prod_{k,d} \left( 2 \sqrt{\epsilon_{k,d} ( 1 - \epsilon_{k,d})}\right)^{n_{i,k,d}}.
\end{equation}

\end{proof}

By summing up $\Pr{ \mathcal{E}_{\eta,i}^c \middle | \{ n_{i,k,d} \}}$ for all $\{ n_{i,k,d} \}$ with the lower bound in~\eqref{eqn:boundEzetaic}, we get
\begin{align}
\Pr{\mathcal{E}_{\eta,i}^c \middle| n_i}
&= \sum_{\{ n_{i,k,d} \}} { n_i \choose n_{i,1,1} \cdots n_{i,w,D}} \prod_{k,d} \left( \frac{\Psi_d}{w \bar \Psi} \right) \Pr{\mathcal{E}_{\eta,i}^c \middle| \{\hat n_{i,k,d} \} } \\
&\geq c_2 e^{- \sqrt{\log m}} \sum_{\{ n_{i,k,d} \}} { n_i \choose n_{i,1,1} \cdots n_{i,w,D}} \prod_{k,d} \left( \frac{\Psi_d}{w \bar \Psi} \right) \left( 2 \sqrt{\epsilon_{k,d} ( 1 - \epsilon_{k,d})}\right)^{ n_{i,k,d}} \\
&= c_2 e^{- \sqrt{\log m}} \left( \sum_{k,d} \frac{\Psi_d}{w \bar \Psi} 2 \sqrt{\epsilon_{k,d} ( 1 - \epsilon_{k,d})} \right)^{ n_i}.
\end{align}
By using $c_1 \log m < n_i < C_1 \log m$ conditioned on $\Delta_2$, we obtain
\begin{equation}\label{Enibound}
\Pr{\mathcal{E}_{\eta,i} \middle| n_i} \leq 1- \left(c_2^{\frac{1}{c_1 \log m}} e^{-\frac{1}{c_1 \sqrt{\log m}}}\sum_{k,d} \frac{\Psi_d}{w \bar \Psi} 2 \sqrt{\epsilon_{k,d} ( 1 - \epsilon_{k,d})} \right)^{ n_i} = 1 - A^{n_i}.
\end{equation}
where
\beq\label{eqn:defA}
A:=c_2^{\frac{1}{c_1 \log m}} e^{-\frac{1}{c_1 \sqrt{\log m}}}\sum_{k,d} \frac{\Psi_d}{w \bar \Psi} 2 \sqrt{\epsilon_{k,d} ( 1 - \epsilon_{k,d})}.
\eeq

Plugging \eqref{Enibound} into the first term of \eqref{Enisum}, and using the symmetry, the first term of \eqref{Enisum} is bounded as
\begin{align}
\sum_{\{n_i\}_{i \in [r]} \in S_2} \Pr{\{n_i\}_{i \in [r]}} \prod_{i=1}^r \Pr{ \mathcal{E}_{\eta,i} \middle | n_i}
&\leq \sum_{\{n_i\}_{i \in [r]} \in S_2} \Pr{\{n_i\}_{i \in [r]}} \prod_{i=1}^r (1 - A^{n_i}) \\
&\leq \sum_{\{n_i\}_{i \in [r]}} \Pr{\{n_i\}_{i \in [r]}} \prod_{i=1}^r (1 - A^{n_i}) \\
\label{Sumrs}
&= \sum_{s=0}^r {r \choose s} (-1)^s \sum_{\{n_i\}_{i \in [r]}} \Pr{\{n_i\}_{i \in [r]}} A^{n_{1:s}}, 
\end{align}
where $n_{1:s} = n_1 + \cdots + n_s$. For given $s$, the second summation in \eqref{Sumrs} is calculated as
\begin{equation}
\sum_{\{n_i\}_{i \in [r]}} \Pr{\{n_i\}_{i \in [r]}} A^{n_{1:s}}=\sum_{\{n_i\}_{i \in [r]}} {n \choose n_1 \cdots n_r \ (n - n_{1:r})} (\bar \Psi)^{n_{1:r}} (1 - r\bar \Psi)^{n - n_{1:r}} A^{n_{1:s}}
= (1 - s\bar{\Psi} (1 - A))^n,
\end{equation}
and finally we have the following upper bound on $\Pr{\mathcal{E}_\eta}$
\begin{equation}\label{Rsbound}
\Pr{\mathcal{E}_\eta} \leq \sum_{s=0}^r (-1)^s (1 - s B)^n,
\end{equation}
where $B = \bar \Psi (1 - A)$. Lemma below shows that the right-hand side of~\eqref{Rsbound} is less than 1 (it actually converges to 0 for $n$ in~\eqref{eqn:nboundconverse}), and this completes the proof of converse. 
\begin{lemma} For $B = \bar \Psi (1 - A)$ with $A$ in~\eqref{eqn:defA}, we have
\beq
 \sum_{s=0}^r (-1)^s (1 - s B)^n\to 0
\eeq
for $n$ in~\eqref{eqn:nboundconverse}.
\end{lemma} 
\begin{proof}
The bound in \eqref{Rsbound} can be related to the balls into bins problem. Suppose we have $n$ balls and there are $r$ bins. For each throw of a ball, the probability that a bin receives a ball is $B$ for all $r$ bins, and the ball is thrown into a ``dummy'' bin with probability $1 - rB$. Then, from the inclusion--exclusion principle, \eqref{Rsbound} is the probability that all $r$ bins receive at least one ball. Let $I_i$ be the indicator that the $i$th bin receives at least one ball, and define $W = I_1 + \cdots + I_r$. Then, our goal is to prove that $\Pr{W = 0}$ is bounded away from $1$. We use the second moment method to prove it.
The expectation of $W$ and $W^2$ are calculated as
\begin{equation}
\mean W = r (1 - B)^n \quad\text{and}\quad \mean W^2 = r (1- B)^n + r(r-1) (1 - 2B)^n.
\end{equation}
From the second moment method, we have
\begin{equation}\label{Smbound}
\Pr{W = 0} \leq \frac{\var W}{(\mean W)^2} = \frac{\mean W^2}{(\mean W)^2} - 1 = \frac{1}{r (1 - B)^n} + \frac{r - 1}{r} \frac{(1 - 2B)^n}{(1 - B)^{2n}} - 1.
\end{equation}
We note that $mB$ approaches to $C = \sum_{k,d} \frac{d \Phi_d}{w} (\sqrt{1 - \epsilon_{k,d}}-\sqrt{\epsilon_{k,d}} )^2$ as $m$ goes to infinity. The second term of \eqref{Smbound} goes to $1$, since
\begin{equation}
\frac{(1 - 2B)^n}{(1 - B)^{2n}} = \left(1 - \frac{B^2}{(1-B)^2}\right)^n \leq \Exp{-n\frac{B^2}{(1-B)^2}} = \Exp{\Omega \left(-\frac{\log m}{m} \right)} \to 1.
\end{equation}
It remains to prove that the first term of \eqref{Smbound} goes to $0$. From the inequality $1-x \geq e^{-\frac{x}{1-x}}$, we have
\begin{equation}
\frac{1}{r (1 - B)^n} \leq \frac{1}{r} e^{n \frac{B}{1-B}} = \frac{\log^7 m}{m} e^{(1-\eta) \frac{m \log m}{C} \frac{B}{1-B}} \leq \Exp{-\frac{\eta}{2} \log m} \to 0,
\end{equation}
and it completes the proof.
\end{proof}

\section{Proof of the Lemma~\ref{lem:app:pfLemmagood1}}\label{app:pfLemmagood1}

In Section \ref{sec:goodevents_state}, we discussed the process of constructing $G_i$ and described a sequence of good events bounding the number of bad queries and bad pairs generated in each step. By proving that this sequence of good events occur with probability $1-o(1/m)$, Lemma \ref{lem:app:pfLemmagood1}, which state the independence of messages used in each phase of Algorithm~\ref{algorithm2}, can be proved. Therefore, 
%Lemma \ref{lem:app:pfLemmagood1} has already been proved when we discuss the process of constructing $G_i$ in Section \ref{sec:goodevents_state} except that the proofs of good events bounding the number of bad queries and bad pairs in each step were not provided. Hence, 
in this section, we make the proof of Lemma \ref{lem:app:pfLemmagood1} complete by proving that the good events occur with probability $1-o(1/m)$. 
Since the number of bad queries and bad pairs at each level of the graph $G_i$ can depend on the number of nodes that appear in each level of the graph, we will additionally define and prove some good events to control the size of the sets at each level of $G_i$ defined in Section \ref{sec:goodevents_state}.

%It is more likely to have more bad queries and bad pairs if the size of the corresponding set is increased. Thus, we will also define and prove some good events that control the size of the sets defined in Section \ref{sec:goodevents_state}.

Before we move on to the proof, we first state the basic Chernoff bound on Bernoulli random variables that will be used repeatedly throughout the proofs of good events.
\begin{lemma}
\label{ber}
For a sequence of independent $n$ Bernoulli random variables $X_1, \cdots, X_n$ having mean $p$, we have
\begin{align*}
\Pr{\sum_{i=1}^n X_i \geq (1 + \delta) np} &\leq \left( \frac{e^\delta}{(1+\delta)^{1+ \delta}} \right)^{np}\text{ for any } \delta > 0,\\
\Pr{\sum_{i=1}^n X_i \leq (1 - \delta) np} &\leq \left( \frac{e^{-\delta}}{(1-\delta)^{1- \delta}} \right)^{np}\text{ for any } 0 < \delta < 1.
\end{align*}
\end{lemma}
Note that when $np=\omega( \log m)$, the upper bounds on the tail probability is both $o(1/m)$ if $\delta$ is a constant.

%\subsection{Analysis of Good Events}
%It is sufficient to prove that $\Pr{\hat{x}_i^{(4)} \neq x_i} = o(1/m)$ for all $i \in [n]$ because then the union bound will give the desired result. Thus, we call an event a \textit{good event} when the probability that it occurs is at least $1 - o(1/m)$. Basically, there are two types of good events. Let us call the factor graph of $\hat x_i^{(4)}$ for algorithm 1 as $G_i$. The events of first type assert that there are not many cycles in $G_i$, and the events of second type assert that the number of assigned queries to an item or a worker is close to average in most cases. The proof for each good event is provided in the appendix.

%Suppose we are given $i \in [n]$. Among all Step 4 queries, we only need to consider the queries that are connected to $x_i$ when analyzing $\Pr{\hat{x}_i^{(4)} \neq x_i}$. Thus, we divide the query assignment process of Step 4 in several steps in a $x_i$-centric way. First, we determine whether a query will contain $x_i$ or not for every query in $A^{(4)}$. A query in $A^{(4)}$ contains $x_i$ with probability $\frac{\bar d}{m}$, independently, where $\bar d = \sum_{d=1}^D d \Phi_d$ is the average query degree. We are given $\partial x_i \cap A^{(4)}$ after this step, and let us define $a = \abs{\partial x_i \cap A^{(4)}}$. The first good event is related to this step.

\subsection{Good Event Regarding Step 1) of Generating $G_i$}
We restate Lemma \ref{good1_1} for completeness.

\newcounter{lemmatemp}
\setcounter{lemmatemp}{\thelemma}
\setcounter{lemma}{\thelemmagood}

\begin{lemma}\label{good1_2}
Let us define the event $S_1$ as
\begin{center}
$S_1$: $c_1 \log m < |\partial x_i^{(4)}| < C_1 \log m$
\end{center}
for some constants $C_1>c_1>0$. Then, we have $\Pr{S_1} \geq 1 - o(1/m)$.
\end{lemma}
\setcounter{lemma}{\thelemmatemp}

\begin{IEEEproof}
The distribution of $|\partial x_i^{(4)}|$ follows $\Bin{n^{(4)}, \frac{\bar d}{m}}$. For any $0 < \epsilon < 1$, since $(\sqrt{1-\epsilon} - \sqrt{\epsilon})^2 < 1$, 
\beq
n^{(4)} :=( 1 + \eta) \frac{ m\log m}{\sum_{d = 1}^D \sum_{k = 1}^w \frac{d\Phi_d}{w} ( \sqrt{1 - \epsilon_{k,d}}-\sqrt{\epsilon_{k,d}})^2}\geq (1+ \eta) \frac{m \log m}{\sum_{d=1}^D d \Phi_d } = (1+ \eta) \frac{ m \log m}{\bar d},
\eeq
and
\beq
n^{(4)} \frac{\bar d}{m} \geq (1+\eta){\log m}.
\eeq
Let us take $c_1$ and $C_1$ such that
$
\Min{\log \left(\frac{c_1^{c_1}}{e^{c_1-1}} \right), \log \left(\frac{C_1^{C_1}}{e^{C_1 - 1}} \right)} \geq 1 - \frac{\eta}{2}.
$
Then, from Lemma \ref{ber}, we get
\beq
\begin{split}
\Pr{c_1 \log m < a_1 < C_1 \log m} &\geq 1- \Exp{- \left( 1 - \frac{\eta}{2} \right) (1+ \eta) \log m} \\
&\geq 1 - \Exp{-\left( 1 + \frac{\eta}{4}\right) \log m} = 1- o(1/m).
\end{split}
\eeq
\end{IEEEproof}

\subsection{Good Event Regarding Step 2) of Generating $G_i$}
\begin{lemma}\label{good2}
Let us define the event $S_2$ as
\begin{center}
$S_2$: There is at most one semi-bad pair in $\partial x_i^{(4)}$.
\end{center}
Then, we have $\Pr{S_2} \geq 1 - o(1/m)$.
\end{lemma}
\begin{IEEEproof}
The probability that query pairs $(j_1, j_2)$ and $(j_3, j_4)$ in $\partial x_i^{(4)}$ share items $x_{i_1}$ and $x_{i_2}$, respectively, i.e., $\partial \tilde y_{j_1} \cap \partial \tilde y_{j_2} = i_1$ and $\partial \tilde y_{j_3} \cap \partial \tilde y_{j_4} = i_2$, is less than $\left(\frac{D}{m} \right)^4$. There are
\beq
{\abs*{\partial x_i^{(4)}} \choose 2} \left({\abs*{\partial x_i^{(4)}} \choose 2} - 1\right) (m-1)^2
\eeq
possibilities of such choices, so conditioned on $S_1$, by union bound we have
\beq
\Pr{\exists\text{ two bad pairs in }\partial x_i^{(4)}} \leq {C_1 \log m \choose 2} \left({C_1 \log m \choose 2} - 1\right) (m-1)^2 \left(\frac{D}{m} \right)^4 = o(1/m).
\eeq
\end{IEEEproof}
%Note that the upper bound of $S_1$ is required in the proof of Lemm \ref{good2}, since if we have too many queries in $\partial x_i^{(4)}$, there could be more semi-bad pairs in $\partial x_i^{(4)}$.

\subsection{Good Event Regarding Step 3) of Generating $G_i$}
\begin{lemma}\label{lem:good3}
Let us define the event $S_3$ as
\begin{center}
$S_3$: $c_3 \log m \left(\frac{\log m}{ \log \log m} \right) \leq \abs{\partial^3 \tilde x_i^{(4)}} \leq C_3 \log m \left(\frac{\log m}{ \log \log m} \right)$.
\end{center}
Then, we have $\Pr{S_3} \geq 1 - o(1/m)$ for some $C_3>c_3>0$.
\end{lemma}
\begin{IEEEproof}
Note that $\abs*{\partial^2 \tilde x_i^{(4)}} = \Theta(D \abs*{\partial x_i^{(4)}})$. A query $y_j$ in $A^{(2)}$ selects one of the labels in $\partial^2 \tilde x_i^{(4)}$ independently with probability $p_j$ defined as
\beq
p_j:= 1 - \frac{{m - \abs*{\partial^2 \tilde x_i^{(4)}} \choose d(j)}}{{m \choose d(j)}}.
\eeq
With the bound $|\partial^2 \tilde x_i^{(4)}|=\Theta(\log m)$ provided by $S_1$, one can easily prove that
\begin{equation}
\label{p2xpj}
p_j = \Theta\left(\frac{\log m}{m} \right).
\end{equation}
Since $|A^{(2)}|=m\frac{\log m}{\log\log m}$, the expectation of $\abs{\partial^3 \tilde x_i^{(4)}}$ is $p_j\cdot |A^{(2)}|=\Theta \left(\log m \left(\frac{\log m}{\log \log m}\right) \right) = \omega( \log m)$, and there exists $C_3,c_3 > 0$ such that $\Pr{S_3} \geq 1 - o(1/m)$ by Lemma \ref{ber}.
\end{IEEEproof}

\subsection{Good Event Regarding Step 4) of Generating $G_i$}
\begin{lemma}\label{good4}
Let us define the event $S_4$ as

\begin{center}
$S_4$: There exists at most one semi-bad pair and at most one semi-bad query with badness equal to two in $\partial^3 \tilde x_i^{(4)}$. \\There is no semi-bad query with badness larger than two.
\end{center}

Then, we have $\Pr{S_4} \geq 1 - o(1/m)$.
\end{lemma}
\begin{IEEEproof}
We first prove that there is at most one semi-bad query. The probability that a query $y_j$ in $\partial^3 \tilde x_i^{(4)}$ selects one of the labels in $\partial^2 \tilde x_i^{(4)}$ when selecting rest of $d(j) - 1$ labels is given as
\beq
p_{j,1}:= 1 - \frac{{m - \abs*{\partial^2 \tilde x_i^{(4)}} \choose d(j) - 1}}{{m - 1 \choose d(j) - 1}},
\eeq
and similar to \eqref{p2xpj}, we have
\beq
p_{j,1} = \Theta\left(\frac{\log m}{m} \right).
\eeq
Hence, by the union bound, one yields
\beq
\Pr{\exists\text{ two semi-bad queries in }\partial^3 \tilde x_i^{(4)}}
\leq {\abs*{\partial^3 \tilde x_i^{(4)}} \choose 2} \Theta \left(\frac{\log^2 m}{m^2}\right)
= O \left( \frac{(\log m)^6}{m^2 (\log \log m)^2}\right)
= o(1 / m),
\eeq
conditioned on $S_3$, which says that $\abs*{\partial^3 \tilde x_i^{(4)}}=\Theta\left(\frac{\log^2 m}{\log\log m}\right)$.

The proof for semi-bad pair with badness equal to two is very similar to the proof of Lemma~\ref{good2} except that we have $\abs{\partial^3 \tilde x_i^{(4)}}$ queries instead of $\abs{\partial \tilde x_i^{(4)}}$ and there are $m - \abs{\partial^2 \tilde x_i^{(4)}}$ choices of labels instead of $m - 1$.

We next prove that there is no semi-bad query with badness larger than two. For a query $y_j$ in $\partial^3 \tilde x_i^{(4)}$ to have badness larger than two, it should select more than one label in $\partial^2 \tilde x_i^{(4)}$ when selecting rest of $d(j) - 1$ labels. The probability is explicitly written as
\begin{equation*}
p_{j,2} := 1 - \frac{{m - \abs*{\partial^2 \tilde x_i^{(4)}} \choose d(j) - 1}}{{m - 1 \choose d(j) - 1}} - (\abs*{\partial^2 \tilde x_i^{(4)}} - 1) \frac{{m - \abs*{\partial^2 \tilde x_i^{(4)}} \choose d(j) - 2}}{{m - 1 \choose d(j) - 1}},
\end{equation*}
and it is asymptotically bounded as
\begin{equation*}
p_{j,2} = O \left(\frac{\log^2 m}{m^2} \right).
\end{equation*}
The union bound conditioned on $S_3$ again gives
\begin{equation*}
\Pr{\exists\text{ one such semi-bad query in }\partial^3 \tilde x_i^{(4)}}
\leq  \abs{\partial^3 \tilde x_i^{(4)}} O \left(\frac{\log^2 m}{m^2}\right)
= O \left( \frac{\log^6 m}{m^2 (\log \log m)^2}\right)
= o(1 / m).
\end{equation*}

\end{IEEEproof}

\subsection{Good Event Regarding Step 8) of Generating $G_i$}
%In Step 6), queries in $A^{(3)}$ are randomly assigned to $[w]$ workers, and  the set of queries in $A^{(3)}$ that are assigned to the worker $(k,d)\in \partial_w x_i^{(4)}$ is denoted by $\partial w_{k,d}^{(3)}:=\partial w_{k,d} \cap A^{(3)}$. The lemma bounds the size of $\partial w_{k,d}^{(3)}$.
\begin{lemma}\label{good5}
Let us define the event $S_5$ as
\begin{center}
$S_5$: For any $(k, d) \in \partial_w x_i ^{(4)}$, $ c_5 (\log m) (\log \log m) < \abs{\partial w_{k,d}^{(3)}} < C_5 (\log m) (\log \log m)$
\end{center}
where $\partial w_{k,d}^{(3)}=\partial w_{k,d} \cap A^{(3)}$. Then, we have $\Pr{S_5} \geq 1 - o(1/m)$ for some $C_5>c_5 > 0$.
\end{lemma}
\begin{IEEEproof}
For any given $(k, d) \in \partial_w x_i^{(4)}$, $\abs{\partial w_{k,d}^{(3)}}$ follows $\Bin{n^{(3)}, \frac{\Phi_d}{w}}$, and for $n^{(3)}=w(\log m)(\log\log m)$ the expectation of $\abs{\partial w_{k,d}^{(3)}}$ is $\Phi_d (\log m)(\log \log m)$. Hence, by Lemma \ref{ber}, there exists $c_5$ and $C_5$ such that
\beq
\Pr{c_5 (\log m)(\log \log m) < \abs{\partial w_{k,d}^{(3)}} < C_5 (\log m)(\log \log m)} \geq 1 - \Theta\left( e^{-\Phi_d(\log m)(\log \log m)} \right).
\eeq
The union bound over $(k, d) \in \partial_w x_i ^{(4)}$ gives
\beq
\Pr{S_5} \geq 1 - \Theta\left( (\log m) e^{-\Phi_d(\log m)(\log \log m)} \right) = 1 - o(1/m).
\eeq
\end{IEEEproof}
Note that Lemma \ref{good5} together with $S_1$ implies
\beq
\abs{\partial x_i^{(3)}} \leq \abs{\partial_w x_i^{(4)}} \Theta (\log m \log \log m) \leq \abs{\partial x_i^{(4)}} \Theta (\log m \log \log m) = O(\log^2 m \log \log m).
\eeq

\begin{lemma}\label{good6}
Let us define the event $S_6$ as
\begin{center}
$S_6$: There is at most one query in $\partial x_i^{(3)}$ that selects a label in $\partial^2 \tilde x_i^{(4)} \cup \partial^4 \tilde x_i^{(4)}$.
\end{center}
Then, we have $\Pr{S_6} \geq 1 - o(1/m)$.
\end{lemma}

\begin{IEEEproof}
The probability that a query $y_j$ in $\partial x_i^{(3)}$ selects one of the labels in $\partial^2 \tilde x_i^{(4)} \cup \partial^4 \tilde x_i^{(4)}$ is given as
\beq
p_j:= 1 - \frac{{m - \abs*{\partial^2 \tilde x_i^{(4)} \cup \partial^4 \tilde x_i^{(4)}} \choose d(j)}}{{m \choose d(j)}},
\eeq
and similar to \eqref{p2xpj}, we have
\beq
p_j = \Theta \left( \frac{\log^2 m}{m \log \log m } \right),
\eeq
where we used $S_3$ and the relation $\abs{\partial^4 \tilde x_i^{(4)}} \leq D \abs{\partial^3 \tilde x_i^{(4)}}$.
By the union bound, we get
\beq
\Pr{\exists\text{ two such queries in }\partial \tilde x_i^{(3)}} \leq {\abs*{\partial x_i^{(3)}} \choose 2}  \Theta \left(\frac{\log^4 m }{m^2 (\log \log m)^2}\right) = O \left( \frac{\log^8 m}{m^2}\right) = o(1 / m),
\eeq
where we used $\abs{\partial x_i^{(3)}} = O(\log^2 m \log \log m)$.
\end{IEEEproof}

%At the next step, each query node in $\partial x_i^{(3)}$ selects $d$ labels from $[m]$ uniformly at random, where $d$ is sampled from the degree distribution $\{\Phi_d\}$. At this step, cycles are generated if there exists a bad pair of queries that select the same label node or there exists a bad query in $\partial x_i^{(3)}$ that selects any label node in $\partial^2 x_i^{(4)}$. 
\begin{lemma}\label{good7}
Let us define the event $S_7$ as
\begin{center}
$S_7$: There is at most one bad pair in $\partial x_i^{(3)}$.
\end{center}
Then, we have $\Pr{S_7} \geq 1 - o(1/m)$.
\end{lemma}
\begin{IEEEproof}
The proof is exactly the same as the proof of Lemma~\ref{good2} except that we use $\partial x_i^{(3)}=O(\log^2 m \log\log m)$ instead of $\abs{\partial x_i^{(4)}} = \Theta (\log m)$.
\end{IEEEproof}

\subsection{Good Event Regarding Step 10) of Generating $G_i$}
\begin{lemma}\label{good8}
Let us define the event $S_{8}$ as
\begin{center}
$S_{8}$: $c_{8} \abs{\partial^2 x_i^{(3)}} \left(\frac{\log m}{ \log \log m} \right) \leq \abs{\partial^3 \tilde x_i^{(3)}} \leq C_{8} \abs{\partial^2 x_i^{(3)}} \left(\frac{\log m}{ \log \log m} \right)$.
\end{center}
Then, we have $\Pr{S_{8}} \geq 1 - o(1/m)$ for some small positive constant $c_{8}$ and large positive constant $C_{8}$.
\end{lemma}
\begin{IEEEproof}
The proof is basically the same as the proof of Lemma~\ref{lem:good3} except that we are given $n^{(2)} - \abs{\partial^3 \tilde x_i^{(4)}}$ queries that select one of $\abs{\partial^2 x_i^{(3)}}$ labels among $m - \abs{\partial^2 \tilde x_i^{(4)}}$ labels. To bound $\abs{\partial^3 \tilde x_i^{(4)}}$, we use the condition $S_3$.
\end{IEEEproof}

\begin{lemma}\label{good9}
Let us define the event $S_9$ as
\begin{center}
$S_9$: There is at most one query in $\partial^3 \tilde x_i^{(3)}$ that selects a label in $\partial^4 \tilde x_i^{(4)} \setminus \partial^2 x_i^{(3)}$.
\end{center}
Then, we have $\Pr{S_9} \geq 1 - o(1/m)$.
\end{lemma}
\begin{IEEEproof}
The proof is the same as the proof of Lemma~\ref{good6} except the queries select labels from $[m] \setminus \partial^2 \tilde x_i^{(4)}$ instead of $[m]$ and $\partial^4 \tilde x_i^{(4)} \setminus \partial^2 x_i^{(3)}$ is used instead of $\partial^2 \tilde x_i^{(4)} \cup \partial^4 \tilde x_i^{(4)}$. The final union bound of Lemma~\ref{good6} is written as
\beq
\Pr{\exists\text{ two such queries in }\partial^3 \tilde x_i^{(3)}} \leq {\abs*{\partial^3 \tilde x_i^{(3)}} \choose 2} \Theta \left(\frac{(\log m)^2 }{m (\log \log m)}\right)^2 = O \left( \frac{(\log m)^6}{m^2 \log \log m}\right) = o(1 / m),
\eeq
where we used $S_8$ to bound $\abs{\partial^3 \tilde x_i^{(3)}}$.
\end{IEEEproof}

\begin{lemma}\label{good10}
Let us define event $S_{10}$ as

\begin{center}
$S_{10}$: There is at most one bad pair and at most one bad query with badness equal to two in $\partial^3 \tilde x_i^{(3)}$.
 \\There is no bad query with badness larger than two.
\end{center}

Then, we have $\Pr{S_{10}} \geq 1 - o(1/m)$.
\end{lemma}
\begin{IEEEproof}
The proof for the bad pair and the bad queries is the same as the proof of Lemma~\ref{good4} except that the queries in $\partial^3 \tilde x_i^{(3)}$ selects labels from $[m] \setminus \partial^2 \tilde x_i^{(4)}$ instead of $[m]$ and we use $S_8$ to bound $\abs{\partial^3 \tilde x_i^{(3)}}$.
\end{IEEEproof}

\section{Proof of the Lemma~\ref{lem:app:pfLemmagood2}}\label{app:pfLemmagood2}

We prove the good events (\romannumeral 5) and (\romannumeral 6) by using some of the good events proved in Appendix \ref{app:pfLemmagood1}. Note that the first parts of good events (\romannumeral 5) and (\romannumeral 6) bounding the size of the sets $\partial x_i^{(4)}$ and $\partial w_{k,d}^{(3)}$ has already been proved in Appendix \ref{app:pfLemmagood1}. Hence, we only prove the good events related to the number of perfect queries and good queries in this section.

\subsection{Good Event Regarding Step 4) of Generating $G_i$}
\begin{lemma}\label{good11}
Let us define the event $S_{11}$ as
\begin{center}
$S_{11}$: The number of good labels in $\partial^2 \tilde x_i^{(4)}$ is larger than or equal to $|\partial^2 \tilde x_i^{(4)}|-C_{11} \log \log m$.
\end{center}
Then, we have $\Pr{S_{11}} \geq 1 - o(1/m)$ for some large constant $C_{11} > 0$.
\end{lemma}
\begin{IEEEproof}
It is sufficient to consider the first query-to-label assignment step in step 4), where each query in $\partial^3 \tilde x_i^{(4)}$ selects one label randomly from $\partial^2 \tilde x_i^{(4)}$, since the number of good labels will only increase after we proceed further. The assignment step is equivalent to the balls and bins problem with $\abs*{\partial^3 \tilde x_i^{(4)}}$ balls and $\abs*{\partial^2 \tilde x_i^{(4)}}$ bins. Let $n_1 = \abs*{\partial^2 \tilde x_i^{(4)}}$ and $n_2 = \abs*{\partial^3 \tilde x_i^{(4)}}$. If we let $\{X_i \}_{i \in [n_1]}$ as the number of balls in each bin, the random variables $\{X_i \}_{i \in [n_1]}$ are negatively correlated \cite{joagdev1983negative}. The negatively correlated random variables have the property that for any set $I \subset [n_1]$,
\begin{equation}
\label{negacorrel}
\Pr{X_i < t,\ \forall i \in I} \leq \prod_{i \in I}\Pr{X_i < t}.
\end{equation}
Each $X_i$ follows the distribution $\Bin{n_2 , \frac{1}{n_1}}$, and since $n_1=\abs*{\partial^2 \tilde x_i^{(4)}}=\Theta(\log m)$ and $n_2=\abs*{\partial^3 \tilde x_i^{(4)}}=\log m \left(\frac{\log m}{ \log \log m} \right)$ conditioned on $S_1$ and $S_3$, respectively, the expectation of it is $\Theta\left(\frac{\log m}{\log \log m}\right)$. Hence, by Lemma \ref{ber}, there exists small positive constant $\delta > 0$ such that
\beq
\Pr{X_i < c_{11} \frac{\log m}{\log \log m}} \leq \Exp{-\delta \frac{\log m}{\log \log m}}.
\eeq
for any $i \in [n_1]$. We prove that the probability that there exists a subset $I$ with $\abs{I} = C_{11} \log \log m$ such that $X_i < c_{11} \frac{\log m}{\log \log m}$ for all $i \in I$ is $o(1/m)$. For a given $I$, this probability is bounded as
\beq
\Pr{X_i < c_{11} \frac{\log m}{\log \log m},\ \forall i \in I} \leq \Exp{-\delta C_{11} \log m}
\eeq
by \eqref{negacorrel}. We use the union bound for ${n_1 \choose c_{11} \log \log m}$ possibilities of $I$ to obtain
\beq
\begin{split}
\Pr{\exists I \subset [n_1]\text{ s.t. }X_i < c_{11} \frac{\log m}{\log \log m},\ \forall i \in I}
&\leq {C_1 D \log m \choose c_{11} \log \log m} \Exp{-\delta C_{11} \log m} \\
&\leq \Exp{-\delta C_{11} (\log(C_1 D) + \log\log m)(c_{11} \log \log m) \log m} \\
&= o(1/m)
\end{split}
\eeq
\end{IEEEproof}

Conditioned on $S_{11}$, there are at most $C_{11}\log\log m$ no-good labels in $\partial^2 \tilde{x}_i^{(4)}$. If there exists no semi-bad pair in $\partial x_i^{(4)}$ selecting the same label from $\partial^2 \tilde{x}_i^{(4)}$, Lemma \ref{good11} implies that there are at most $C_{11} \log \log m$ queries in $\partial x_i^{(4)}$ containing a no-good label node in its label-selection process. By Lemma~\ref{lem:app:pfLemmagood1}, there are at most three queries in $\partial x_i^{(4)}$ that form a semi-bad pair, or are the parent nodes of a semi-bad pair or a semi-bad query. Therefore, these two arguments conclude that there are at least $|\partial x_i^{(4)}|-C_{11}\log\log m - 3$ perfect queries in $\partial x_i^{(4)}$.

\subsection{Good Event Regarding Step 10) of Generating $G_i$}
In the lemma below, $\partial^2 x_i^{(3)}$ refers to the set defined in~\eqref{eqn:partial2 x_i^3}, before the removal of the nodes by Step 10).

\begin{lemma}\label{good12}
Let us define the event $S_{12}$ as
\begin{center}
$S_{12}$: The number of good labels in $\partial^2 x_i^{(3)}$ is larger than or equal to $|\partial^2 x_i^{(3)}|-C_{12} \log \log m -1$.
\end{center}
Then, we have $\Pr{S_{12}} \geq 1 - o(1/m)$ for some large constant $C_{12} > 0$.
\end{lemma}
\begin{IEEEproof}
At most one label in $\partial^2 \tilde x_i^{(4)}$ can be included in $\partial^2 x_i^{(3)}$  and the label will not receive any query from $\partial^3 \tilde x_i^{(3)}$. Thus, we consider the worst case where such a label is not a good label, and then the proof is exactly the same as the proof of Lemma~\ref{good11} except that we have $\abs{\partial^3 \tilde x_i^{(3)}}$ balls and $\abs{\partial^2 x_i^{(3)}} - 1$ bins. We use $S_8$ to bound $\abs{\partial^3 \tilde x_i^{(3)}}$.
\end{IEEEproof}
Conditioned on  Lemma \ref{good12}, there are at most $(C_{12} \log \log m +1)$ no-good labels in $\partial^2 x_i^{(3)}$. 
After removing at most five queries from $\partial x_i^{(3)}$, each label node in $\partial^2 x_i^{(3)}$ is connected to only one query node in $\partial x_i^{(3)}$ by Lemma~\ref{lem:app:pfLemmagood1}. Thus, there are at least $|\partial w_{k,d}^{(3)}|-C_{12}\log\log m - 6$ good queries in $\partial w_{k,d}^{(3)}$ for each worker $(k, d) \in \partial_w x_i^{(4)}$.

\bibliographystyle{IEEEtran}
% argument is your BibTeX string definitions and bibliography database(s)
\bibliography{IEEEabrv,main}

\end{document}